\documentclass[a4paper,pdftex,reqno,11pt]{amsart}

\usepackage{times}
\usepackage{amsmath,amssymb}
\usepackage{hyperref}
\usepackage{wasysym}

\usepackage{geometry}
\newtheorem{theorem}{theorem}[section]

\newtheorem{conjecture}[theorem]{Conjecture}
\newtheorem{corollary}[theorem]{Corollary}
\newtheorem{lemma}[theorem]{Lemma}

\newtheorem{definition}[theorem]{Definition}

\def\game{v} 
\def\Gone{\game_{n,m}^{k,l}}
\def\Gtwo{\game_{n,n+1}^{k,l}}
\def\Gthree{\game_{n,0}^{k,l}}

\def\voter{player}

\begin{document}

%\date{\today}

\title{The inverse problem for power distributions in committees}

\author{Sascha Kurz}
\address{Sascha Kurz, Department of Mathematics, Physics, and Computer Science, University of Bayreuth, 95440 Bayreuth, Germany.  
Tel.: +49-921-557353, Fax: +49-921-557352, sascha.kurz@uni-bayreuth.de}

\maketitle

\begin{abstract}
   Several power indices have been introduced in the literature in order to measure the influence of individual
   committee members on the aggregated decision. Here we ask the inverse question and aim to design voting rules
   for a committee such that a given desired power distribution is met as closely as possible. We present
   an exact algorithm for a large class of different power indices based on integer linear programming. With
   respect to negative approximation results we generalize the approach of Alon and Edelman who studied power
   distributions for the Banzhaf index, where most of the power is concentrated on few coordinates. It turned out 
   that each Banzhaf vector of an $n$-member committee that is near to such a desired power distribution, has to be
   also near to the Banzhaf vector of a $k$-member committee. We show that such Alon-Edelman type results are
   possible for other power indices like e.g.\ the Public Good index or the Coleman index to prevent actions, while
   they are principally impossible for e.g.\ the Johnston index. 
   
  \bigskip

\end{abstract}

\section{Introduction}
\noindent
Consider a committee that attains its decisions by voting like e.g.\ the U.S.\ Senate or the European Parliament.
Whenever the committee members have different capabilities to influence decisions, the question for the measurement
of individual power arises. Dating back to the late 18th century, Luther Martin was probably the first approaching
this issue, see e.g.\ \cite{felsenthal2005voting}. Since the second half of the 20th century, several power indices
were formally introduced in order to measure the voting power in collective decision making procedures, see e.g.\ 
\cite{bertini2013comparing,1141.91005,0954.91019,1208.91004,laruelle2013voting} for some surveys and discussions. 
A huge amount of literature has been devoted to the study of the properties, shortcomings, paradoxes, and axiomatic
foundations of several power indices. Within the scientific community the general approval of a unique
power index as the optimal compromise is not yet in sight\footnote{Arguably, some power indices, like the Banzhaf and
the Shapley-Shubik index, are generally more accepted and applied than others. On the other hand, the pros and cons
of several power indices are frequently discussed in the older and latest literature.}. So some researches argue to choose
the appropriate index according to the situation\footnote{As an example, we mention \cite{brams1988presidential}
arguing that the Johnston index is best suited for measuring presidential power.}, see e.g.\ \cite{straffin1977homogeneity}. 

Having a suitable power index at hand, one can ask the fundamental problem of how to design a voting procedure, such
that the resulting distribution of power (according to the chosen power index) among the committee members meets or
almost meets a fixed vector\footnote{There exists a stream of literature discussing the question of a \textit{fair}
power distribution within a committee, see e.g.\ \cite{le2012voting,Penrose,laruelle1998allocation,shapley1954method}.}.
This is commonly called the inverse power index problem\footnote{There are several more recent papers using this denomination
and trying to algorithmically attack this issue. Considerations about the problem itself date back for a longer time,
see. e.g.\ \cite{imrie1973impact,nurmi1982problem,papayanopoulos1983}.}. As an algorithmic answer, several heuristics,
mostly with an iterative nature, and exact algorithms have been developed, see \cite{kurz2012heuristic} for a recent
overview. The proposed heuristics are generally quite fast and produce good numerical results, meaning that they achieve
small approximation errors on practical instances. From the theoretical point of view
there is a fundamental problem, since those algorithms are not able to give a priori bounds what approximation errors are
achievable. Also non-trivial a posteriori bounds are missing. Thus one simply does not know whether an iterative heuristic 
has converged sufficiently, i.e.\ its approximation
error is near the minimum possible approximation. Methods based on exhaustive enumeration of voting procedures, see
e.g.\ \cite{de2010enumeration,de2012solving,kurz2012minimum}, can in principle determine the exact solution of the
inverse problem in finite time. Since for most classes of voting procedures their number grows faster than exponential, 
those exhaustive algorithms are applicable for small numbers of voters only. Unfortunately, no intermediate lower bounds
are produced so that one either has to completely determine the exact solution or to be contented with the best solution
found in a limited time horizon, i.e.\ using the approach as an elaborate heuristic. Approaches based on integer linear
programming, see \cite{kurz2012inverse}, automatically come with lower and upper bounds in the intermediate computation
steps. The drawback is the lack of theoretical bounds on the running time and the approximation quality, when interrupted
after a polynomial running time. Recently, several algorithms have been designed that can achieve a sufficiently accurate
approximation with running time polynomial in terms of the number of {\voter}s times a factor depending on the desired
approximation quality, see \cite{De:2012:NOS:2213977.2214043,o2011chow,de2012inverse}.\footnote{For completeness, we 
remark that there is also a stream of literature that characterizes the sets of transferable utility games, i.e.\ more
general objects than the voting procedures that we will study here, whose distribution vector of a certain value,
i.e.\ a more general object than a power index, exactly coincides with a given vector, see e.g.\ 
\cite{1126.91005,dragan2012inverse,dragan2013inverse}.}

Since the number of voting procedures is generally finite as the number of voters is fixed, one clearly cannot approximate
certain power distributions too closely if the number of voters is small. As mentioned before, the number of voting procedures
grows faster than exponential for the most relevant classes, so that one might expect that for large committees each
desired power distribution can be approximated to a nicety. Indeed there exists a stream of literature concerning with limit
results saying that under certain technical conditions one can simply choose the entries of a given power distribution 
as weights in a weighted voting game, whose power distribution then is relatively close to the initial vector, see e.g.\
\cite{chang2006ls,dubey1979mathematical,kurz2013nucleolus,lindner2004ls,lindner2007cases,milnor1978values,
neyman1982renewal,shapiro1978values}.   

It was quite a surprise when Alon and Edelman showed that many power distributions, where most of the power is concentrated
on a small number of {\voter}s, are hard to approximate independently of the number of {\voter}s. In their seminal paper
\cite{pre05681536}, the authors give explicit bounds stating that a Banzhaf vector, whose weight is concentrated on $k<n$
{\voter}s, has to be near to the Banzhaf vector of a game with $n-k$ null {\voter}s, which is essentially a game with just $k$
{\voter}s. In \cite{kurz2012inverse}, their theorem was applied to deduce that for $n\ge 2$ {\voter}s the 
$\Vert\cdot\Vert_1$-distance between the Banzhaf vector of a simple game and the power distribution $(0.75,0.25,0,\dots)$
has to be at least $\frac{1}{9}$. Computationally a lower approximation bound of at least $0.37846$ was shown for all simple
games up to $11$ {\voter}s.

The contribution of this paper consists of a slight tightening of the Alon-Edelman bound. As a consequence, we can deduce
that for $n\ge 2$ {\voter}s the $\Vert\cdot\Vert_1$-distance between the Banzhaf vector of a simple game and 
$(0.75,0.25,0,\dots)$ has to be at least $\frac{1}{8}$. We slightly generalize their theorem so that it is applicable
for some classes of voting procedures other than simple games. More importantly, we continue their study of
inverse power index problems and prove similar bounds for several power indices defined in the literature. As argued
before, this is relevant for the practical application of power indices, since different power indices are suitable
in different settings and the Banzhaf index is just a single, while rather important, example of a power index. To this end
we break down the approach of Alon and Edelman into smaller pieces, which can be studied from a more general point of
view. Several parts of this new theoretic framework can then be used for different power indices and the considerations
depending on the precise definition, of the power index under study, are somewhat minimized. It will turn out that many power
indices, including e.g.\ the Public Good index, the Shift index, and the Deegan-Packel index, admit an Alon-Edelman type
approximation bound, while the Johnston index does not admit such a result. For other concepts of power indices, like
$p$-binomial semivalues, the rounding procedure proposed in \cite{pre05681536} does not work directly, but may be adjusted
in a meaningful way. So there are classes of power indices that admit Alon-Edelman type results and others that do not. 
With respect to the classification of power indices and their comparison, this structural distinction might be interesting 
in itself and deserves further study. While we can decide to which class certain power indices belong for many cases, we 
leave a few examples, like the Shapley-Shubik index, open and just state conjectures.

The rest of the paper is organized as follows. In Section~\ref{sec_games} we formally introduce the
classes of binary voting procedures along with some notational conventions. The power indices we are studying
in this paper are defined in Section~\ref{sec_power_indices}. The main general framework for Alon-Edelman type
results is stated in Section~\ref{sec_alon_edelmann_type}. An integer linear programming formulation for the
inverse power index problem for power indices based on counting functions is stated in Section~\ref{sec_ilp_formulations}.
The counterpart of desired power distributions, where most of the power is concentrated on a small number
of coordinates, is studied in Section~\ref{sec_not_concentrated}. It turns out that with the aid of so-called
limit results the (approximate) inverse power index problem becomes almost trivial under certain technical conditions. 
We end the main part of the paper with a conclusion and some remarks on possible future work in Section~\ref{sec_conclusion}.
The technical proofs for our Alon-Edelman type results from Section~\ref{sec_alon_edelmann_type} are shifted to
Appendix~\ref{sec_details_quality_functions}. The technical details for the parametric class of weighted voting games,
e.g.\ used to prove the non-existence of an Alon-Edelman type result for the Johnston index, are presented in 
Appendix~\ref{sec_details_parameterized_wvg}.

\section{Binary voting procedures}
\label{sec_games}
In this paper we restrict ourselves to the study of binary voting procedures. By this we mean that each committee member
has the possibility to vote either {\lq}yes{\rq} or {\lq}no{\rq} for each proposal. The aggregated committee decision
then also has to be either {\lq}yes{\rq} or {\lq}no{\rq}. This assumption is certainly an oversimplification in several
practical applications and more general models with several levels of approvals in the input and output have been
stated in the literature, see e.g.\ \cite{bolger1986power,freixas2003weighted,freixas2009anonymous}. Also continuous
models or the presence of communication, interaction structures, or a priori unions have been studied so far. We will
abstain from all those refinements and leave the study of the corresponding Alon-Edelman type bounds as significant open
research problems. Next we will give a brief introduction in the taxonomy of binary voting procedures and refer the
interested reader to e.g.\ \cite{0943.91005} for a more detailed treatment. 

Within this paper we denote by $N$ the set of committee members, voters, or players. For brevity, we will only use the term 
\emph{{\voter}s} in the following, which is commonly used in the literature. Its cardinality
$|N|$ is denoted by $n$. Since in our context the names of the {\voter}s play no role, we will typically use
the set $N=\{1,2,\dots,n\}=:[n]$. Similarly let $(k,n]:=\{k+1,\dots,n\}$, for positive integers $k<n$, and 
$2^X$ denote the set of all subsets of a (finite) set $X$. Subsets of $N$, i.e.\ elements of $2^N$, are also called
\emph{coalitions}. To keep notation as simple as possible we will frequently use the following set-theoretic abbreviations: 
$N-S=N\backslash S$, $S-i=S\backslash\{i\}$, $S+i=S\cup i=S\cup\{i\}$, where $S,N$ are sets and $i\in S$.

\begin{definition}
  \label{def_boolean_game}
  A \emph{Boolean game} is a function $\game:2^N\rightarrow\{0,1\}$ with $\game(\emptyset)=0$ and $\game(N)=1$. The
  set of all Boolean games on $n$ players, i.e.\ one can assume $N=[n]$, is denoted by $\mathcal{B}_n$.
\end{definition}  

The restriction $\game(\emptyset)=0$ is some kind of a minimal reasonable assumption in the context of voting in a
committee. In a situation where no {\voter} is in favor of a proposal it would be pretty factitious if the aggregated 
overall committee decision would be {\lq}yes{\rq}. A similar consideration can be performed for the situation 
where all committee members are in favor of a proposal. More such reasonability arguments come to one's mind
immediately, e.g.\ an enlarged group of supporters should not turn the aggregated decision from {\lq}yes{\rq} to
{\lq}no{\rq}: 

\begin{definition}
  \label{def_simple}
  A \emph{simple game} is a Boolean game $\game:2^N\rightarrow\{0,1\}$ such that $\game(S)\le \game(T)$ for all
  $\emptyset\subseteq S\subseteq T\subseteq N$. % and $f(N)=1$. 
  The set of all simple games on $n$ players is denoted by $\mathcal{S}_n$.
\end{definition}

The coalitions $S\subseteq N$ of a Boolean game $\game$ with $\game(S)=1$ are called \emph{winning} coalitions and
the other ones are called \emph{loosing} coalitions. If $S$ is a winning coalition, but all of its proper subsets
are losing, we call $S$ a \emph{minimal winning} coalition. Similarly we call a losing coalition $S$ 
\emph{maximal losing} if all of its proper supersets are winning. 

\begin{definition}
  \label{definition_minimal_winning_and_co}
  Let $\game$ be a Boolean game. By $\mathcal{W}$ we denote the set of all winning and by $\mathcal{W}^m$ we
  denote the set of all minimal winning coalitions of $\game$. Similarly, by $\mathcal{L}$ we denote the set of
  all losing and by $\mathcal{L}^M$ we denote the set of all maximal losing coalitions of $\game$.  
\end{definition}

We remark that each simple game is uniquely characterized by specifying $N$ and either one of the sets $\mathcal{W}$, $\mathcal{W}^m$,
$\mathcal{L}$, or $\mathcal{L}^M$. So in the following we will from time to time also use the pair $(\mathcal{W},N)$
to denote a simple game instead of the function notation $\game:2^N\rightarrow\{0,1\}$. A {\voter} that is not contained
in any minimal winning coalition is called a \emph{null} {\voter}. Boolean games are uniquely characterized by either
of the sets $\mathcal{W}$ or $\mathcal{L}$. 

Logically inverting the statement of a proposal can result in counterintuitive outcomes in the context of voting, 
if all simple games are permitted. So far it is possible that both, a coalition $S$ and its complementary coalition
$N\backslash S$, can carry through a proposal.

\begin{definition}
  A simple game is called \emph{proper} if the complement $N-S$ of any winning coalition $S$ is losing. It is called
  \emph{strong} if the complement $N-S$ of any losing coalition $S$ is winning. A simple game that is both proper and
  strong is called \emph{constant-sum} (or self-dual, or decisive).
\end{definition} 

The \textit{desirability relation}, introduced in \cite{0083.14301}, assumes a certain intuitive ordering of the {\voter}s:

\begin{definition}
  \label{def_desirability_relation} Given a simple game $(\mathcal{W},N)$ we say that {\voter}~$i\in N$ is \emph{more desirable}
  as {\voter}~$j\in N$, denoted by $i\succeq j$, if
  \begin{enumerate}
    \item[(1)] for all $S\subseteq N-\{i,j\}$ with $S+j\in\mathcal{W}$, we have $S+i\in\mathcal{W}$;  
    \item[(2)] for all $S\subseteq N-\{i,j\}$ with $S+i\in\mathcal{L}$, we have $S+j\in\mathcal{L}=2^N-\mathcal{W}$.
  \end{enumerate}
  We write $i\simeq j$ if $i\succeq j$,  $j\succeq i$ and use $i\succ j$ as abbreviation for $i\succeq j$, $i\not\simeq j$. 
\end{definition}

\begin{definition}
  A simple game $(\mathcal{W},N)$ is called \emph{complete} if for each pair of {\voter}s $i,j\in N$ we have $i\succeq j$
  or $j\succeq i$. The set of all complete (simple) games on $n$ {\voter}s is denoted by $\mathcal{C}_n$.
\end{definition}

\begin{definition}
  Let $(\mathcal{W},[n])$ be a complete simple game, where $1\succeq 2\succeq \dots\succeq n$, and $S\subseteq N$.
  A coalition $T\subseteq N$ is a \emph{direct left-shift} of $S$ whenever there exists a {\voter} $i\in S$ with
  $i-1\notin S$ such that $T=S-\{i\}+\{i-1\}$ for $i>1$ or $T=S+\{n\}$ for $n\notin S$. Similarly, a coalition 
  $T\subseteq N$ is a \emph{direct right-shift} of $S$ whenever there exists a {\voter} $i\in S$ with
  $i+1\notin S$ such that $T=S-\{i\}+\{i+1\}$ for $i<n$ or $T=S-\{n\}$ for $n\in S$. A coalition $T$ is a
  \emph{left-shift} of $S$ if it arises as a sequence of direct left-shifts. For brevity we denote this case by 
  $T\succeq S$. Similarly, it is a \emph{right-shift} of $S$ if it arises as a sequence of direct right-shifts. 
  For brevity we denote this case by $T\preceq S$. A winning coalition $S$ such that all right-shifts of
  $S$ are losing is called \emph{shift-minimal} (winning). Similarly, a losing coalition $S$ such that all 
  left-shifts of $S$ are winning is called \emph{shift-maximal} (losing). By $\mathcal{W}^{sm}$ we denote the
  set of all shift-minimal winning coalitions of $(\mathcal{W},N)$ and by $\mathcal{L}^{sM}$ the set of all
  shift-maximal losing coalitions. 
\end{definition}

We remark that every shift-minimal winning coalition has to be a minimal winning coalition and each shift-maximal losing
coalition has to be a maximal losing coalition.

\begin{definition}
  A simple game $(\mathcal{W},[n])$ is \emph{weighted} if there exists a quota $q>0$ and weights $w_1,\dots,w_n\ge 0$
  such that $S$ is winning if and only if $w(S)=\sum_{i\in S} w_i\ge q$. We denote the game by $[q;w_1,\dots,w_n]$.
  The set of all weighted (simple) games on $n$ {\voter}s is denoted by $\mathcal{T}_n$.
\end{definition}

Since the weights clearly induce the desirability relation, each weighted game is complete. We remark that weighted
representations are not unique, e.g.\ $[2;2,1,1]=[3;3,2,1]$. Not every simple game is weighted, but it can be represented 
as a finite intersection of weighted games.

The number of Boolean games on $n$ {\voter}s, not taking symmetry into account,  is clearly given by $2^{2^n-1}$. Enumeration 
results for the other subclasses of binary voting procedures can be found in \cite{de2012solving,kurz2012minimum,kurz2013dedekind}.

\section{Power indices}
\label{sec_power_indices}
Power indices are a formal way to measure the influence of a single {\voter} on the outcome of a committee decision.
They became necessary when researchers relatively early have discovered that the relative influences are not directly
proportional to voting weights in a weighted game, see e.g.\ \cite{kurz2012heuristic} for a more extensive discussion.
In most cases, power indices can be defined for larger classes than weighted games, sometimes even for transferable
utility games. Many power indices have been proposed so far in the literature. Here, we want to briefly present a large
collection of those and arrange them in a somewhat systematic order. Our classification and listing is based on 
\cite{AlonsoMeijide20123395,bertini2013comparing,1141.91005,0954.91019,felsenthal2005voting,holler1982forming,
1208.91004,laruelle2013voting,nurmi1980game}\footnote{Also the website http://powerslave.val.utu.fi/indices.html 
has served as a source.}. In general terms, a \emph{power index} on a class
$\mathcal{V}_n\subseteq \mathcal{B}_n$ of binary games, consisting of $n$ {\voter}s, is a mapping
$g:\mathcal{V}_n\rightarrow \mathbb{R}^n$. Mostly we use $\mathcal{V}_n\in\{\mathcal{B}_n,\mathcal{S}_n\}$
and in some cases the restriction to $\mathcal{V}_n\in\{\mathcal{C}_n,\mathcal{T}_n\}$. Usually we state
the definition of the vector-valued power indices just for its components $g_i$, i.e.\ for an arbitrary {\voter}~$i\in N$.

\subsection{Power indices derived from values}
Some power indices are derived from values or solution concepts for transferable utility games, see e.g.\ 
\cite{0993.90001,1142.91017} for the more general point of view. Here we will state only the restriction of
the underlying values to power indices, i.e.\ to simple or Boolean games instead of transferable utility games.
Let us start with one of the most famous values or power indices: In \cite{0050.14404}, the Shapley value was
axiomatically introduced and shortly after, see \cite{shapley1954method}, restricted to simple games:
\begin{definition}
  \label{def_shapley_shubik}
  The \emph{Shapley-Shubik} index of a Boolean game $\game:2^N\rightarrow\{0,1\}$ for voter $i\in N$ is given by
  $$
    \operatorname{SSI}_i(\game)=\sum_{S\subseteq N-i} \frac{|S|!(|N|-1-|S|)!}{|N|!}\cdot\left(\game(S\cup i)-\game(S)\right).
  $$ 
\end{definition}  

The Shapley-Shubik index satisfies some nice properties, i.e.\ it is symmetric, positive, efficient on $\mathcal{S}_n$
and satisfies the null {\voter} property.

\begin{definition}
  \label{def_nice_power_index_properties}
  Let $g:\mathcal{V}_n\rightarrow \mathbb{R}^n=(g_i)_{i\in[n]}$ be a power index on a class $\mathcal{V}_n$ of
  binary games. We say that
  \begin{enumerate}
    \item[(1)] $g$ is \emph{symmetric}: if for all $\game\in\mathcal{V}_n$ and any bijection $\tau:[n]\rightarrow [n]$
               we have $g_{\tau(i)}(\tau \game)=g_i(\game)$, where $\tau \game(S)=\game(\tau(S))$ for any coalition
               $S\subseteq[n]$;
    \item[(2)] $g$ is \emph{positive}: if for all $\game\in\mathcal{V}_n$ and all $i\in[n]$ we have $g_i(\game)\ge 0$
               and $g(\game)\neq 0$;
    \item[(3)] $g$ is \emph{efficient}: if for all $\game\in\mathcal{V}_n$ we have $\sum_{i=1}^n g_i(\game)=1$;
    \item[(4)] $g$ satisfies the \emph{null {\voter} property}: if for all $\game\in\mathcal{V}_n$ and all null {\voter}s
               $i$ of $\game$ we have $g_i(\game)=0$.           
  \end{enumerate}
\end{definition}

We remark that the Shapley-Shubik index is not efficient on $\mathcal{B}_n$ in general (for $n\ge 3$) by looking
at the following example: Define the Boolean game $\game:2^{[3]}\rightarrow\{0,1\}$, which attains the value $1$
exactly for the coalitions $\{1\}$, $\{2\}$, and $\{1,2,3\}$. Inserting into the definition gives 
$\operatorname{SSI}_1(\game)=\operatorname{SSI}_2(\game)=\frac{1}{2}$, and $\operatorname{SSI}_3(\game)=\frac{1}{3}$,
which sums up to $\frac{4}{3}\neq 1$. Nevertheless, Boolean games might be seen as a rather 
obscure class of binary voting procedures, we remark that many of the subsequent power indices are not efficient.
Since efficiency is a desirable property (in some contexts), we generally define a normalization of a power index 
in order to make it efficient:

\begin{definition}
  \label{def_normalization}
  For a positive power index $P:\mathcal{V}_n\rightarrow\mathbb{R}^n$ we define the corresponding 
  \emph{normalized} power index $\widehat{P}:\mathcal{V}_n\rightarrow\mathbb{R}^n$ as
  $
    \widehat{P}_i(\game)=\frac{P_i(\game)}{\sum_{j=1}^n P_j(\game)} 
  $
  for all $\game\in\mathcal{V}_n$.
\end{definition}
 
\begin{lemma}
  \label{lemma_efficient}
  With the notation and assumptions from Definition~\ref{def_normalization}, $\widehat{P}$ is efficient (on $\mathcal{V}_n$).
\end{lemma} 

If we call each {\voter}~$i\in N$ a \emph{vetoer} of a simple game $\game$ whenever $\game(N-i)=0$, then we can restate the index
from \cite{tijs1981bounds}:
\begin{definition} 
  \label{def_tijs}
  The \emph{Tijs index} of a simple game $\game:2^N\rightarrow\{0,1\}$, containing at least one vetoer, for 
  {\voter}~$i\in N$ is given by $\operatorname{Tijs}_i(\game)=1$ is $i$ is a vetoer and zero otherwise.
\end{definition}

The Tijs index is symmetric and satisfies the null {\voter} property. If the class of games is restricted to
those containing at least one vetoer, the Tijs index is also positive, so that its normalized version is
efficient due to Lemma~\ref{lemma_efficient}. For our purposes the possibility $\operatorname{Tijs}(\game)=0$
causes minor technical difficulties. Nevertheless vetoers are an important concept in political science.
In \cite{tsebelis2002veto} analyses of existing political institutions have been formalized into the so-called
veto-player-theorem, e.g.\ stating that departure from the status quo in political institutions is more likely, 
the smaller is the number of vetoers (which makes also sense if no vetoer is present).

In \cite{0191.49502}, the author has introduced the nucleolus $\operatorname{Nuc}$ 
as the lexicographically minimal imputation. It has been proposed as a power index as recently as the last decade,
see e.g.\ \cite{le2012voting,montero2006noncooperative,montero2013nucleolus}. We remark that the nucleolus satisfies
all four properties of Definition~\ref{def_nice_power_index_properties} (on Boolean games).

Semivalues were introduced in \cite{weber1979} for simple games and later generalized in \cite{dubey1981value}. 
They can be seen as weighted averages of a {\voter}'s marginal contribution to coalitions.

\begin{definition}
  \label{def_semi_value}
  Let $\mathbf{p}=(p_0,\dots,p_{n-1})$ be a vector of non-negative real numbers satisfying 
  $\sum_{j=0}^{n-1} p_j{{n-1}\choose{j}}=1$. The \emph{semivalue} of a Boolean game $\game:2^N\rightarrow\{0,1\}$
  with respect to $\mathbf{p}$ for {\voter}~$i\in N$ is given by
  $$
    \Psi_i^{\mathbf{p}}(\game)=\sum_{S\subseteq N-i} p_{|S|}\cdot \left[\game(S\cup i)-\game(S)\right].
  $$
\end{definition}  
The Shapley-Shubik index is given by $p_j=\frac{1}{n\cdot{{n-1}\choose j}}$ and the later on defined absolute Banzhaf index
is given by $p_j=\frac{1}{2}$ for all $0\le j\le n-1$. Semivalues are symmetric, positive and satisfy the null {\voter} property
on $\mathcal{S}_n$, but typically are not efficient.
 
\begin{definition}
  \label{def_p_binomial_semivalue}
  For $p\in(0,1)$ the semivalue $\Psi^p=\Psi^{\mathbf{p}}$, with $p_j=p^j(1-p)^{n-1-j}$ for $0\le j\le n-1$, is called a 
  \emph{$p$-binomial semivalue}.
\end{definition} 
 
\subsection{Power indices bases on winning, swing, and critical coalitions}
\label{subsec_banzhaf_related}
In \cite{bertini2013comparing} the authors present a large list of autonomously generated power indices, as they call
them, which are, in contrast to those from the previous section, not obtained from values. We repeat and enlarge
their list at this place and the following subsection. All examples are based on winning, critical, minimal winning, 
shift-minimal winning, and swing coalitions and in some sense more or less related to the Banzhaf index. A possibly 
more general point of view is taken in Subsection~\ref{subsec_counting_functions}.

In order to state the definitions of the announced power indices we need a bit more notation:
\begin{definition}
  \label{def_swing}
  Let $\game=(\mathcal{W},N)$ be a simple game. A coalition $S\subseteq N-i$ with $\game(S\cup i)-\game(S)=1$ is called
  an \emph{$i$-swing} or a \emph{swing} for {\voter}~$i$. The number of $i$-swings of {\voter}~$i\in N$ is denoted by $\eta_i$.
  As abbreviation we use $\eta=\sum_{j=1}^n \eta_j$.
\end{definition}   

Let $\mathcal{W}_i$ denote the set of winning coalitions that contain {\voter}~$i$. For an $i$-swing $S\subseteq N-i$
the coalition $S\cup i$ is a winning coalition. We say that $i$ is \emph{critical} in $S\cup i$ since removing {\voter}
$i$ from $S\cup i$ turns the coalition from winning to losing. For a losing coalition $S\subseteq N-i$, we similarly
say that {\voter}~$i$ is critical if $S\cup i$ is winning. 

\begin{definition}
  \label{def_absolute Banzhaf}
  The \emph{absolute Banzhaf index}\footnote{Named after \cite{banzhaf1964weighted}, but originally going back to
  \cite{Penrose}, so that some authors speak of the Penrose-Banzhaf index.} of a simple game $\game=(\mathcal{W},[n])$
  for {\voter}~$i\in N$ is given by
  $\operatorname{BZ}_i(\game)=\eta_i(\game)/2^{n-1}$.
\end{definition}
Its normalization $\widehat{\operatorname{BZ}}$ is usually called the \emph{relative Banzhaf index} and satisfies
all four properties of Definition~\ref{def_nice_power_index_properties}.

In \cite{coleman1971control} the author, among other things, defined two power indices:
\begin{definition}
  \label{def_coleman}
  Let $\game:2^N\rightarrow\{0,1\}$ be a simple game. The \emph{Coleman power} of a member \emph{to
  prevent action} for {\voter}~$i\in N$ is given by 
  $$
    \operatorname{ColPrev}_i(\game)=\frac{\text{\# winning coalitions in which $i$ is critical}}
    {\text{total number of winning coalitions}}=\frac{\eta_i(\game)}{|\mathcal{W}|}
  $$ 
  and the \emph{Coleman power} of a member \emph{to initiate action} is given by
  $$
    \operatorname{ColIni}_i(\game)=\frac{\text{\# losing coalitions in which $i$ is critical}}
    {\text{total number of losing coalitions}}.
  $$
\end{definition}

As noted e.g.\ in \cite{dubey1979mathematical}, the normalized versions of the power indices 
from Definition~\ref{def_absolute Banzhaf} and Definition~\ref{def_coleman} coincide, i.e.\ we have
$\widehat{\operatorname{BZ}}=\widehat{\operatorname{ColPrev}}=\widehat{\operatorname{ColIni}}$.
In other words those three power indices are scaled versions of each other. We remark that for
different purposes normalizing power indices indeed destroys some information since the scaling
variant depends non-trivially on the game. So in practice all three power indices 
may have their justification. 

\begin{definition}
  \label{def_rae}
  The \emph{Rae index}\footnote{Originally, the Rae index was introduced in \cite{rae1969decision}. It
  coincides with the satisfaction index studied in a different context in \cite{brams1978power}
  and is sometimes called the Brams-Lake index.} of a Boolean game
  $\game=2^N\rightarrow\{0,1\}$ for {\voter}~$i\in N$ is given by
  $$
    \operatorname{Rae}_i(\game)=\frac{\left|S\subseteq N\,:\,i\in S, \game(S)=1\right|+
    \left|S\subseteq N\,:\,i\notin S, \game(S)=0\right|}{2^{|N|}}.
  $$
\end{definition}
In \cite[Eq. 53, p. 124]{dubey1979mathematical}, the authors prove the identity 
$$
  \operatorname{Rae}(\game)=\frac{1}{2}+\frac{1}{2}\cdot \operatorname{Bz}(\game),
$$
which was also known to Penrose long before the publication of Rae, see e.g.\ \cite{felsenthal2005voting}.
In other words, the Rae index arises as a \emph{linear transformation} from the absolute Banzhaf index.
In our context we can easily conclude Alon-Edelman type bounds for linearly transformed power indices,
see Lemma~\ref{lemma_scaled_counting_function} and Lemma~\ref{lemma_additive_shift}. 

\begin{definition}
  \label{def_koenig_braeuninger}
  The \emph{K\"onig-Br\"auninger index}\footnote{Introduced in \cite{KB} and also called inclusiveness index. This
  index is equivalent to the Zipke index, see \cite{nevison1978naive}.} of a Boolean game
  $\game=2^N\rightarrow\{0,1\}$ for {\voter}~$i\in N$ is given by
  $$
    \operatorname{KB}_i(\game)=\frac{\text{\# winning coalitions that contain $i$}}
    {\text{total number of winning coalitions}}=
    \frac{|\mathcal{W}_i|}{|\mathcal{W}|}.
  $$
\end{definition}

A rather similar power index has been defined in \cite{bertini2008public}:
\begin{definition}
  \label{def_public_help_index}
  The \emph{Public Help index} of a Boolean game $\game=2^N\rightarrow\{0,1\}$ for {\voter}~$i\in N$ is given by
  $$
    \operatorname{PHI}_i(\game)=\frac{|\mathcal{W}_i|}{\sum_{j\in N}|\mathcal{W}_j|}.
  $$
\end{definition}

The last two power indices are closely related to the so-called Chow parameters introduced in
\cite{chow1961characterization}. In their original form, they are given by the $n+1$ numbers
$|\mathcal{W}_1|,\dots,|\mathcal{W}_n|$ and $|\mathcal{W}|$.\footnote{Some authors have redefined them
as $\eta_1,\dots,\eta_n$, i.e.\ the number of $i$-swings being the numerators of the absolute Banzhaf
index, and $|\mathcal{W}|$. The relation between both versions is given by $\eta_i=2|\mathcal{W}_i|-|\mathcal{W}|$,
see e.g.\ \cite{dubey1979mathematical}.} We remark that the (original) Chow parameters uniquely characterize
each simple game. Dropping the number of winning coalitions gives a vector of $n$ numbers, which
uniquely characterizes each weighted game, and can be used as a power index. We have 
$\widehat{\operatorname{KB}}=\widehat{\operatorname{PHI}}=\widehat{\operatorname{Chow}}$ and remark that
the vector $(\eta_1,\dots,\eta_n)$ of the number of swings is called Banzhaf score by some authors.

\begin{definition}
  \label{def_chow}
  The \emph{Chow index} of a Boolean game $\game=2^N\rightarrow\{0,1\}$ for {\voter}~$i\in N$ is given by
  $$
    \operatorname{Chow}_i(\game)=|\mathcal{W}_i|.
  $$
\end{definition}

The power indices described so far in this subsection, are all linear transforms of counting winning,
losing, swing, or critical coalitions. Here, an object like a winning coalition for the Chow index can be 
counted for multiple {\voter}s, i.e.\ it is counted for all members of the respective coalition. Another concept
is to distribute the \textit{contribution}, of a winning coalition $S$ in our example, equally among the
contributing {\voter}s, i.e.\ to just count $\frac{1}{|S|}$ instead of $1$ for all members of $S$. In 
Subsection~\ref{subsec_counting_functions}, or more precisely in Definition~\ref{def_equal_division}, we
describe the underlying idea in more detail. Applying this rather general concept to the absolute Banzhaf
index after multiplication with $2^{n-1}$ gives:
  
\begin{definition}
  \label{def_absolute_johnston}
  The \emph{absolute Johnston index}\footnote{Introduced in \cite{johnston1978measurement} and also called
  Johnston score by some authors.} of a Boolean game $\game=2^N\rightarrow\{0,1\}$ for voter $i\in N$
  is given by
  $$
    \operatorname{JS}_i(\game)=\sum_{\{i\}\subseteq S\subseteq N\,:\,\text{$i$ is critical in $S$, \game(i)=1}}
    \frac{1}{\text{\# of critical {\voter}s in }S}.
  $$
\end{definition}

For the weighted game $[2;2,1,1]$ with {\voter} set $[3]$, the absolute Johnston index is given by
$\left(3,\frac{1}{2},\frac{1}{2}\right)$. Its normalization $\widehat{\operatorname{JS}}$, called
relative Johnston index, is given by $\left(\frac{3}{4},\frac{1}{8},\frac{1}{8}\right)$. For the
weighted voting game $[3;2,1,1]$ with {\voter} set $[3]$ we have 
$\operatorname{JS}=\left(2,\frac{1}{2},\frac{1}{2}\right)$ and 
$\widehat{\operatorname{JS}}=\left(\frac{2}{3},\frac{1}{6},\frac{1}{6}\right)$, which coincides with
the Shapley-Shubik vector.

%For the weighted voting game $[3;2,1,1]$ with {\voter} set $[3]$ the winning coalitions are $\{1,2,3\}$, $\{1,2\}$,
%and $\{1,3\}$. In the first coalition only {\voter}~$1$ is critical while in the other winning coalitions both respective
%{\voter}s are critical, so that $C_1=1+\frac{1}{2}+\frac{1}{2}=2$ and $C_2=C_3=\frac{1}{2}$. The corresponding normalized 
%power distribution then is given by $\left(\frac{2}{3},\frac{1}{6},\frac{1}{6}\right)$, i.e.\ coincides with the 
%Shapley-Shubik index.

\subsection{Power indices derived from minimal or shift-minimal winning coalitions}
\label{subsec_m_and_sm_winning}

In the previous subsection we have considered power indices based on winning, critical, or swing coalitions. The underlying
idea is that the distinction between a winning and a losing coalition is \textit{the} crucial difference which
should be mirrored in the definition of a power index. On the other hand, Riker introduced so the so-called
Riker's size principle in \cite{riker1962theory}, claiming that parties attempt to increase the size of a 
coalition supporting a proposal only until the point where it gets minimally winning. Thus, we consider power 
indices based on minimal or shift-minimal winning coalitions in this subsection and remark that there are
also other lines of argumentation in order to justify those concepts, see e.g.\ 
\cite{holler1983,holler1998two,widgren2001probabilistic}. As abbreviation we use $\mathcal{W}^m_i$ 
for the set of minimal winning coalitions containing {\voter}~$i$ and $\mathcal{W}^{sm}_i$ for the set of 
shift-minimal winning coalitions containing {\voter}~$i$.    

\begin{definition}
  \label{def_PGI}
  The \emph{absolute Public Good index}\footnote{The (relative) Public Good index was introduced in
  \cite{holler1982forming} and is also known as the Holler-Packel index due to an axiomatization
  of Holler and Packel.} of a simple game $\game=2^N\rightarrow\{0,1\}=(\mathcal{W},N)$ for voter
  $i\in N$ is given by
  $$
    \operatorname{PGI}_i(\game)=\left|\mathcal{W}^m_i\right|.
  $$
  Its normalization $\widehat{\operatorname{PGI}}$ is called \emph{(relative) Public Good index}.
\end{definition}

As the Johnston index arises from the Banzhaf index by a certain \textit{payoff-distribution} rule,
see Definition~\ref{def_equal_division}, there is also a counterpart to the Public Good index:

\begin{definition}
  \label{def_Deeagan_Packel}
  The \emph{absolute Deegan-Packel index}\footnote{The (relative) Deegan-Packel index was introduced in \cite{deegan1978new}.}
  of a simple game $\game=2^N\rightarrow\{0,1\}=(\mathcal{W},N)$ for {\voter}~$i\in N$ is given by
  $$
    \operatorname{DP}_i(\game)=\sum_{S\in \mathcal{W}^m_i} \frac{1}{|S|}.
  $$
  Its normalization $\widehat{\operatorname{DP}}=\frac{1}{\left|\mathcal{W}^m\right|}\cdot \operatorname{DP}$ is
  called \emph{(relative) Deegan-Packel index}.
\end{definition}

For shift-minimal winning coalitions we have a similar pair of power indices:

\begin{definition}
  \label{def_shift_index}
  The \emph{absolute Shift index}\footnote{The (relative) Shift index was introduced in \cite{alonso2010new}.
  As a justification for power indices based on shift-minimal winning coalitions one may go along Riker's
  size principle and additionally assume that parties try to avoid the inclusion of powerful members.}
  of a complete simple game $\game=2^N\rightarrow\{0,1\}=(\mathcal{W},N)$ for {\voter}~$i\in N$ is given by
  $$
    \operatorname{Shift}_i(\game)=\left|\mathcal{W}^{sm}_i\right|.
  $$
  Its normalization $\widehat{\operatorname{Shift}}$ is called \emph{(relative) Shift index}.
\end{definition}

\begin{definition}
  \label{def_Shift_Deeagan_Packel}
  The \emph{absolute Shift-Deegan-Packel index}\footnote{The (relative) Shift-Deegan-Packel index was introduced
  in \cite{AlonsoMeijide20123395}.} of a complete simple game $\game=2^N\rightarrow\{0,1\}=(\mathcal{W},N)$
  for {\voter}~$i\in N$ is given by
  $$
    \operatorname{SDP}_i(\game)=\sum_{S\in \mathcal{W}^{sm}_i} \frac{1}{|S|}.
  $$
  Its normalization $\widehat{\operatorname{SDP}}=\frac{1}{\left|\mathcal{W}^{sm}\right|}\cdot \operatorname{SDP}$ is
  called \emph{(relative) Shift-Deegan-Packel index}.
\end{definition}

\subsection{Power indices for weighted games}
\label{subsec_weighted}
Most of the power indices introduced so far were defined for simple games or even Boolean games. The
Shift index and the Shift-Deegan-Packel index from the previous subsection are an exception that can only be 
defined for complete simple games. For the most restrictive class of binary voting procedures
that we are considering in this paper, i.e.\ weighted games, also a very few power indices have been introduced
in the literature.

In \cite{colomer1995paradox} the authors defined a different variant of the Banzhaf index by reweighting
a swing based counting. For a given weighted voting game $\game:2^{[n]}\rightarrow\{0,1\}=[q;w_1,\dots,w_n]$
the Colomer index, which is also known under the name executive power index, is defined as
$$
  \operatorname{Colomer}_i(\game)=
  \frac{\sum\limits_{\{i\}\subseteq S\subseteq N} \frac{w_i\cdot\big(\game(S)-\game(S-i)\big)}
  {\sum_{j\in S}w_j\cdot\big(\game(S)-\game(S-j)\big)}}
  {\sum_{S\subseteq N}f(S)}
  =\frac{\sum\limits_{S\in\mathcal{W}_i}
  \frac{w_i\cdot\big(\game(S)-\game(S-i)\big)}
  {\sum_{j\in S}w_j\cdot\big(\game(S)-\game(S-j)\big)}
  }
  {|\mathcal{W}|}.
$$  
We remark that this index directly depends on the weighted representation of the game and not only the
underlying simple game, i.e.\ the function $\game$ does not suffices to determine the values of the power
distribution, so that our notation $\operatorname{Colomer}_i(\game)$ is slightly misleading.

An index that is not harmed by the ambiguity of weighted representations is defined in \cite{freixas2013minimum}.
The underlying theoretical concepts are minimum sum integer representations of weighted games, see e.g.\
\cite{kurz2012minimum,freixaskurz2013minimum}. While such representations do not need to be unique in general, at
the very least the set of such distinguished representations is finite in all cases. The so-called MSR index is
defined as the average over all minimum sum integer representations of weighted game. In some sense this index
brings us back to the original motivation for power indices, i.e.\ to have a measure for influence better than 
the original weights, which can vary to a large extend of what is generally considered as being meaningful.

\subsection{Counting functions} 
\label{subsec_counting_functions} 
Many of the power indices defined in the previous subsections have a common structure, i.e.\ they arise by
counting a certain quantity like swing coalitions or winning coalitions\footnote{Several authors have tried
to provide a description of such a common structure, see e.g.\ \cite{bertini2013comparing,countingpowerindices}.
Our approach does not claim to be superior and has of course many similarities, but it seems to be more convenient
in our situation.}. In some cases these counts are weighted like
for the $p$-binomial semivalues. Almost always those counts can be decomposed as a sum over all coalitions.

\begin{definition}
  \label{def_counting_function}
  Let $\mathcal{V}_n$ be a class of binary games on $n$ {\voter}s. A \emph{counting function} $C$ (on $\mathcal{V}_n$)
  is a mapping from $\mathcal{V}_n\times 2^N\times N$ to $\mathbb{R}_{\ge 0}^n$. We write
  $C_i:\mathcal{V}_n\times 2^N\rightarrow\mathbb{R}_{\ge 0}$ for the restriction of $C$ to an arbitrary
  {\voter}~$i\in N$ and $\overline{C}:\mathcal{V}_n\times 2^N\rightarrow\mathbb{R}_{\ge 0}$ with
  $\overline{C}(\game,S)=\sum_{j\in N} C_i(\game,S)$.   
\end{definition}

An example of a counting function, on the set $\mathcal{S}_n$ of simple games consisting of $n$~{\voter}s, is given by
$$
  C_i(\game,S)=\left\{\begin{array}{rcl}1/2^{n-1}&:&i\in S,\game(S)=1,\game(S-i)=0,\\0&:&\text{otherwise}.\end{array}\right.
$$
We remark that it counts the number $\eta_i$ of $i$-swings via $\eta_i(\game)=2^{n-1}\cdot\sum_{S\subseteq N} C_i(\game,S)$.
A different counting function that counts the same quantities is given by
$$
  \widetilde{C}_i(\game,S)=\left\{\begin{array}{rcl}1/2^{n-1}&:&i\notin S,\game(S)=0,\game(S+i)=1,\\0&:&\text{otherwise}.\end{array}\right.
$$

\begin{definition}
  \label{def_induced_power_index}
  Let $\mathcal{V}_n$ be a class of binary games on $n$ {\voter}s. Given a counting function $C$ on $\mathcal{V}_n$,
  the \emph{induced power index} $P:\mathcal{V}_n\rightarrow \mathbb{R}^n_{\ge 0}$ (on $\mathcal{V}_n$) is given by
  $$
    P_i(\game)=C_i(\game,2^N):=\sum_{S\in 2^N}C_i(\game,S)
  $$ 
  for all $i\in [n]$ and all $\game\in\mathcal{V}_n$. By $\hat{P}$ we denote the normalized version of $P$, see
  Definition~\ref{def_normalization}. 
\end{definition}

The induced power indices of both $C$ and $\widetilde{C}$, as stated above, are equivalent to the absolute Banzhaf
index. In the following we will write $C^{\operatorname{Bz}}$ when referring to this counting function $C$. When summed
up over all coalitions $C^{\operatorname{Bz}}_i$ counts the number of coalitions where $i$ is a critical {\voter}. 
Similarly $\sum_{i=1}^n\sum_{S\subseteq N} C^{\operatorname{Bz}}_i(\cdot,S)$ gives the number of coalitions with
at least a critical {\voter}, including multiplicities.  

We remark that it is not too hard to give counting functions for all power indices of Section~\ref{sec_power_indices},
except for the nucleolus, the Colomer index and the MSR index, such that the corresponding induced power index is equivalent
to the respective power index. Two such examples are given by 
$$
  C_i^{\operatorname{PGI}}(\game,S)=\left\{\begin{array}{rcl}1&:&S\in\mathcal{W}^m_i,\\0&:&\text{otherwise}\end{array}\right.
$$
for the absolute Public Good index and
$$
  C_i^{\operatorname{Shift}}(\game,S)=\left\{\begin{array}{rcl}1&:&S\in \mathcal{W}^{sm}_i,\\0&:&\text{otherwise}\end{array}\right.
$$
for the absolute Shift index. The first function \textit{counts} minimal winning coalitions and the second
\textit{counts} shift-minimal winning coalitions. 

Of course it is easily possible to write down a counting function such that the induced power index is equivalent  
to the absolute Johnston index:
$$
  C_i^{\operatorname{JS}}(\game,S)=\left\{
  \begin{array}{rcl}
    \frac{1}{\left|\{j\in S\,:\,\game(S-j)=0\}\right|} &:&i\in S,\game(S)=1, \game(S-i)=0,\\
    0 &:& \text{otherwise.}
  \end{array}
  \right.
$$
But as we mentioned in Subsection~\ref{subsec_banzhaf_related} a more general concept is underlying:
\begin{definition}
  \label{def_equal_division}
  Let $\mathcal{V}_n$ be a class of binary games on $n$ {\voter}s and $C$ be a counting function on $\mathcal{V}_n$.
  The \emph{equal division} counting function $C'$ is given by
  $$
    C'_i(\game,S)=\left\{\begin{array}{rcl}
    \frac{\max_{j\in S} C_j(\game,S)}{\left|\{j\in S\,:\,C_j(\game,S)>0\}\right|}&:&C_i(\game,S)>0,\\
    0&:&\text{otherwise} 
    \end{array}\right.
  $$  
  for all $i\in[n]$ and all $\game\in\mathcal{V}_n$.
\end{definition} 

In other words, we have $C'_i(\game,S)=C'_j(\game,S)$ for all $i,j\in[n]$ with $C_i(\game,S)$, $C_j(\game,S)>0$, i.e.\
an equal division of the \textit{payoff} $\max_{j\in S} C_j(\game,S)$ of coalition $S$ to all \textit{contributing}
{\voter}s. With this terminology at hand we can state that the counting function of the absolute Johnston index
$C^{\operatorname{JS}}$ arises as the equal division version of $2^{n-1}\cdot C^{\operatorname{Bz}}$. This
correspondence is the reason why we have defined $C^{\operatorname{Bz}}$ as $C$ and not as $\widetilde{C}$,
see the equations before Definition~\ref{def_induced_power_index}. We have that 
$\sum_{i=1}^n\sum_{S\subseteq N}C^{\operatorname{JS}}(\game,S)$ equals the number of winning coalitions with at least
one critical {\voter} (without multiplicities).

Similarly, the counting function of the absolute Deegan-Packel index $C^{\operatorname{DP}}$ arises  as the equal
division version of $C^{\operatorname{PGI}}$ and the counting function of the absolute Shift-Deegan-Packel index
$C^{\operatorname{SDP}}$ arises as the equal division version of $C^{\operatorname{Shift}}$.

Directly from the definition of a counting function and the definition from a normalized power index we conclude:
\begin{lemma}
  \label{lemma_normalized_counting_function}
  Let $\mathcal{V}_n$ be a class of binary games on $n$ {\voter}s, $C$ be a counting function on $\mathcal{V}_n$, and
  $P:\mathcal{V}_n\rightarrow\mathbb{R}_{\ge 0}^n$ be the induced power index. Then the normalized
  power index $\widehat{P}$ is induced by the counting function $\widehat{C}=\frac{1}{\Lambda}\cdot C$, where
  $$
    \lambda(\game)=\sum_{i=1}^n\sum_{S\subseteq N} C_i(\game,S)
  $$ 
  for each $\game\in\mathcal{V}_n$, i.e.\ the scaling factor $\frac{1}{\Lambda}$ may depend on the game $\game$.
\end{lemma}

\section{Alon-Edelman type bounds}
\label{sec_alon_edelmann_type}
In this section we establish a list of what we call Alon-Edelman type bounds or results for the 
$\Vert\cdot\Vert_1$-distance between Boolean games and somewhat \textit{simplified} or \textit{reduced}
Boolean games. We start with the original result from \cite{pre05681536} (slightly rewritten in our notation).
To this end, we call a Boolean game $\game:2^{[n]}\rightarrow\{0,1\}$ on $n$ {\voter}s \emph{$k$-pure} if all
{\voter}s in $(k,n]$ are null {\voter}s.  

\begin{theorem}
  \label{thm_alon_edelman}
  \textbf{(Alon and Edelman, 2010)}
  Let $n>k$ be positive integers, let $\varepsilon<\frac{1}{k+1}$ be a positive real, and let
  $\game=(\mathcal{W},[n])$ be a simple game on $n$ {\voter}s. If 
  $\sum_{i=k+1}^n \widehat{\operatorname{Bz}}_i(\game)\le \varepsilon$,
  then there exists a $k$-pure simple game $\game'$ on $n$ {\voter}s so that
  $$
    \Vert \widehat{\operatorname{Bz}}(\game')-\widehat{\operatorname{Bz}}(\game)\Vert_1=
    \sum_{i=1}^n \left|\widehat{\operatorname{Bz}}_i(\game')-\widehat{\operatorname{Bz}}_i(\game)\right|
    \le \frac{(2k+1)\varepsilon}{1-(k+1)\varepsilon}+\varepsilon.
  $$ 
\end{theorem}  

In the following we will generalize this theorem for other power indices besides the (relative)
Banzhaf index. As a side effect, we will tighten the stated bound to a nicer expression.
We go along the lines of the original proof and generalize the underlying ideas. To this end, 
we separate smaller parts and introduce additional notation aiming at a general framework for   
Alon-Edelman type results.

Instead of the (relative) Banzhaf index we want to use an almost arbitrary power index, which then 
of course has to satisfy some technical conditions, since there are power indices for which a
Alon-Edelman type bound cannot exist (in the precise formulation that we will state shortly), as
we will see later on. To condense the mentioned technical conditions we introduce the property of
being \emph{locally approximable} of a power index based on counting functions in 
Definition~\ref{def_locally_approximable}. The quality of the local approximation is quantified 
with the aid of two functions, which in turn determine the bounds of the Alon-Edelman type results.
We will specify how the \textit{simplified} $k$-pure game $\game'$ arises from $\game$, i.e.\ we
make the existence result constructive\footnote{This constructive reformulation was already implicitly
contained in the proof of \cite{pre05681536}.}, see Definition~\ref{def_shortening}. 

Anticipating the necessary notation given in Subsection~\ref{subsec_shortenings_and_local_approximability},
we can state our main theorem as follows:
\begin{theorem}
  \label{main_thm}
  Let $0<k<n$ be integers, $P$ be a power index induced by a counting function $C$ that
  is locally approximable for the shortening function $\Gamma:\mathcal{V}_n\rightarrow\mathcal{V}_n$
  with quality functions $f_1$ and $f_2$. If $P$ further satisfies the \textit{null {\voter} 
  property} and is \textit{positive}\footnote{see Definition~\ref{def_nice_power_index_properties}.(4)
  and Definition~\ref{def_nice_power_index_properties}.(2)}, then for each game
  $\game\in\mathcal{V}_n$ and its shortening $\game'=\Gamma(\game)\in\mathcal{V}_n$ with 
  $\sum_{i=k+1}^n P_i(\game)\le\varepsilon\cdot\sum_{i=1}^n P_i(\game)$ we have
  $$
    \Vert P(\game')-P(\game)\Vert_1
    =\sum_{i=1}^k\left|P_i(\game')-P_i(\game)\right|+\sum_{i=k+1}^n \left|P_i(\game)\right|
    \le \left(kf_1(k)+1\right)\cdot\varepsilon\cdot\sum_{i=1}^n P_i(\game).
  $$
  If $\sum\limits_{i=k+1}^n \widehat{P}_i(\game)\le\varepsilon'$, then we have
  $$
    \Vert \widehat{P}(\game')-\widehat{P}(\game)\Vert_1
    =\sum_{i=1}^k\left|\widehat{P}_i(\game')-\widehat{P}_i(\game)\right|+\sum_{i=k+1}^n \left|\widehat{P}_i(\game)\right|
    \le (f_2(k)+kf_1(k)+1)\varepsilon'
  $$ 
  for the normalized power index $\widehat{P}$.
\end{theorem}
\begin{proof}
  Since $P_i(\game)=C_i(\game,2^N)$ the condition $\sum_{i=k+1}^n P_i(\game)\le\varepsilon\cdot\sum_{i=1}^n P_i(\game)$ 
  can be rewritten as $\sum_{i=k+1}^n C_i(\game,2^N)\le\varepsilon\cdot\overline{C}(\game,2^N)$.
  From local approximability we conclude
  $$
    \sum_{i=1}^k\left|P_i(\game')-P_i(\game)\right|
    =\sum_{i=1}^k\left|C_i(\game',2^N)-C_i(\game,2^N)\right|
    \le k\cdot f_1(k)\cdot\varepsilon \cdot\overline{C}(\game,2^N),
  $$    
  where the ride hand side equals $kf_1(k)\cdot\varepsilon\cdot\sum_{i=1}^n P_i(\game)$. Since $P$ is positive
  we have $\sum_{i=k+1}^n \left|P_i(\game)\right|=\sum_{i=k+1}^n P_i(\game)\le\varepsilon\cdot\sum_{i=1}^n P_i(\game)$,
  so that we can combine both inequalities to obtain the proposed bound for $\Vert P(\game')-P(\game)\Vert_1$.
  
  \medskip
  
  \begin{eqnarray*}
    && \Vert \widehat{P}(\game')-\widehat{P}(\game)\Vert_1=
    \sum_{i=1}^k \left|\frac{C_i(\game,2^N)}{\overline{C}(\game,2^N)}-\frac{C_i(\game',2^N)}{\overline{C}(\game',2^N)}\right|
    +\sum_{i=k+1}^n \widehat{P}_i(\game)\\
    &\le& \sum_{i=1}^k \left|\frac{C_i(\game,2^N)}{\overline{C}(\game,2^N)}-\frac{C_i(\game',2^N)}{\overline{C}(\game,2^N)}\right|+
          \sum_{i=1}^k \left|\frac{C_i(\game',2^N)}{\overline{C}(\game,2^N)}-\frac{C_i(\game',2^N)}{\overline{C}(\game',2^N)}\right|
          +\varepsilon'\\
    &\le& \sum_{i=1}^k \left|\frac{C_i(\game,2^N)-C_i(\game',2^N)}{\overline{C}(\game,2^N)}\right|+
          \left|\frac{\overline{C}(\game',2^N)-\overline{C}(\game,2^N)}{\overline{C}(\game',2^N)\cdot
           \overline{C}(\game,2^N)}\right|\cdot \sum_{i=1}^k C_i(\game',2^N)
          +\varepsilon'\\
    &\le& \left(kf_1(k)+f_2(k)+1\right)\cdot\varepsilon',    
  \end{eqnarray*}
  where we have used the triangle inequality for absolute values, $\sum\limits_{i=1}^k C_i(\game',2^N)=\overline{C}(\game',2^N)$,
  and local approximability.
\end{proof}

The requirements of Theorem~\ref{main_thm} are almost in one-to-one correspondence to those from Theorem~\ref{thm_alon_edelman}
except the additional requirement $\varepsilon<\frac{1}{k+1}$ for the Alon-Edelman result. This condition is indeed
necessary also in our context if we want to prevent from the case that $\Gamma(\game)=\game'$ equals either 
one of the non-Boolean games $(\emptyset,[n])$ or $(2^{[n]},[n])$, see Corollary \ref{corollary_k_rounding}. 
As it will turn out in Lemma~\ref{lemma_quality_functions_absolute_banzhaf} we can choose $f_2(k)=k+1$ for the absolute
Banzhaf index, the following lemma completes the correspondence (except for the tightness of the proposed upper bound).

\begin{lemma}
  \label{lemma_not_trivial_game}
  With the notation from Theorem~\ref{main_thm} let $\varepsilon<\frac{1}{f_2(k)}$ or $\varepsilon'<\frac{1}{f_2(k)}$.
  If we additionally assume that $P((\emptyset,[n]))=P((2^{[n]},[n]))=\mathbf{0}$ and $P(\game)\neq\mathbf{0}$, then
  $\game'=\Gamma(\game)\notin\left\{(\emptyset,[n]),(2^{[n]},[n])\right\}$.
\end{lemma}  
\begin{proof}
  At first we remark that $\sum_{i=k+1}^n P_i(\game)\le\varepsilon\cdot\sum_{i=1}^n P_i(\game)$ is equivalent to
  $\sum\limits_{i=k+1}^n \widehat{P}_i(\game)\le\varepsilon'$ for $\varepsilon=\varepsilon'$, since $\widehat{P}$ is
  efficient. Thus we assume $\varepsilon<\frac{1}{f_2(k)}$ and it suffices to show $P(\game')\neq\mathbf{0}$.
  If to the contrary $P(\game')=\mathbf{0}$, then 
  $$
    \Vert P(\game)\Vert_1=\Vert P(\game')-P(\game)\Vert_1=\left|\sum_{i=1}^n P_i(\game)\right|
    \le (f_2(k)+1)\varepsilon\cdot\Vert P(\game)\Vert_1
  $$
  due to local approximability, the relation $P_i(\game)=C_i(\game,2^N)$, and $\sum_{i=1}^n P_i(\game)=\Vert P(\game)\Vert_1$.
  For $\varepsilon<\frac{1}{f_2(k)}$ and $\Vert P(\game)\Vert_1>0$ this is impossible. 
\end{proof}
  
Anticipating Lemma~\ref{lemma_quality_functions_absolute_banzhaf} we remark that Theorem~\ref{main_thm} yields
a tighter bound for the special case of the Banzhaf index than Theorem~\ref{thm_alon_edelman}, since
we have $$\frac{(2k+1)\varepsilon}{1-(k+1)\varepsilon}+\varepsilon>\left(2k+2\right)\varepsilon$$ for all $k,\varepsilon>0$
\footnote{We remark that the differences in the upper bounds are due to the tighter estimate of 
$\left|\widehat{\operatorname{Bz}}_i(\game,2^N)-\widehat{\operatorname{Bz}}_i(\game',2^N)\right|$ in the second part
of the proof of Theorem~\ref{main_thm} compared to the estimation in \cite{pre05681536}. The generalized bound would have
been $\Vert \widehat{P}(\game')-\widehat{P}(\game)\Vert_1\le
\frac{(2kf_1(k)+1)\varepsilon'}{1-(kf_1(k)+1)\varepsilon'}+\varepsilon'$ using the original proof.}.  

An application of Theorem~\ref{subsec_applications} is given in Subsection~\ref{subsec_applications}. In the following
subsections we give the necessary technical framework and determine for which power indices the stated conditions are
satisfied (some cases will remain open nevertheless). For some of the power indices introduced in Section~\ref{sec_power_indices}
we can easily see that it makes no sense to ask for an Alon-Edelman type result like Theorem~\ref{main_thm}. Our key assumption
is that most of the \textit{power} is concentrated on the first $k$ voters. Some power indices from the literature
have the property that  {\voter}s individual power can differ only up to a fix multiplicative constant independently from the
given game.

\begin{lemma}
  \label{lemma_no_sense_to_aks_for_alon_edelman_result}
  For each simple game $\game\in\mathcal{S}_n$ we have
  \begin{eqnarray*}
    \frac{1}{2}\operatorname{KB}_j(\game)\le\operatorname{KB}_i(\game)\le 2\operatorname{KB}_j(\game),\\
    \frac{1}{4}\operatorname{PHI}_j(\game)\le\operatorname{PHI}_i(\game)\le 4\operatorname{PHI}_j(\game),\\
    \frac{1}{2}\operatorname{Chow}_j(\game)\le\operatorname{Chow}_i(\game)\le 2\operatorname{Chow}_j(\game),\\
    \frac{1}{2}\widehat{\operatorname{KB}}_j(\game)\le\widehat{\operatorname{KB}}_i(\game)\le 2\widehat{\operatorname{KB}}_j(\game),\\
    \frac{1}{4}\widehat{\operatorname{PHI}}_j(\game)\le\widehat{\operatorname{PHI}}_i(\game)\le 4\widehat{\operatorname{PHI}}_j(\game)\text{ and}\\
    \frac{1}{2}\widehat{\operatorname{Chow}}_j(\game)\le\widehat{\operatorname{Chow}}_i(\game)\le 2\widehat{\operatorname{Chow}}_j(\game)
  \end{eqnarray*}
  for all $i,j\in[n]$.
\end{lemma} 
\begin{proof}
  Let $\game=(\mathcal{W},[n])\in\mathcal{S}_n$, $i\in[n]$ be arbitrary , and set $A_i=\mathcal{W}\backslash\mathcal{W}_i$,
  i.e.\ the set of winning coalitions that do not contain {\voter}~$i$. Since $\game$ is simple we have 
  $\left|A_i\right|\le \left|\mathcal{W}_i\right|$, so that
  $$
    \frac{1}{2}\cdot\left|\mathcal{W}\right|\le\left|\mathcal{W}_i\right|\le\left|\mathcal{W}\right|\le 2\cdot\left|\mathcal{W}_i\right|.
  $$
  This directly gives the two inequalities for $\operatorname{Chow}$ and $\operatorname{KB}$. Estimating the numerator
  and denominator of $\operatorname{PHI}$ separately, gives the stated inequality for the Public Help index. Since the normalization
  is the same for every {\voter}, we can also conclude the three remaining inequalities.
\end{proof}

We remark that more explicitly $\operatorname{KB}_i(\game)\in\left[\frac{1}{2},1\right]$,
$\widehat{\operatorname{KB}}_i(\game)\in\left(\frac{1}{2n},\frac{2}{n}\right)$, 
$\operatorname{PHI}_i(\game)\in\left[\frac{1}{2n},\frac{2}{n}\right]$, and
$\widehat{\operatorname{PHI}}_i(\game)\in\left(\frac{1}{4n},\frac{4}{n}\right)$ for all simple games
$\game\in\mathcal{S}_n$ and all {\voter}s $i\in[n]$. 

\subsection{Shortenings and locally approximable power indices}
\label{subsec_shortenings_and_local_approximability}

\begin{definition}
  \label{def_reduced_game}
  For a Boolean game $\game=(\mathcal{W},[n])$ and for $A\subseteq [k]$ we define
  the \emph{reduced game}\footnote{See also \cite[Definition 1.4.4, 1.4.7]{0943.91005}.} 
  $\game_A=(\mathcal{W}_A,(k,n])$, where 
  $\mathcal{W}_A=\{B\subseteq (k,n]\,:\,A\cup B\in\mathcal{W}\}$.
\end{definition}

We remark that $\game_A$ consists of $n-k$ {\voter}s. It may happen that $\game_A$ is not a Boolean game,
but it can be easily figured out that there are just three cases:

\begin{lemma}
  \label{lemma_reduced_game}
  For a Boolean game $\game=(\mathcal{W},[n])\in\mathcal{B}_n$ and a subset $A\subseteq [k]$, where $0<k<n$, 
  (using the notation from Definition~\ref{def_reduced_game}) we have one of the following possibilities:
  \begin{itemize}
    \item[(1)] $\mathcal{W}_A=\emptyset$;
    \item[(2)] $\mathcal{W}_A=2^{(k,n]}$;
    \item[(3)] $v_A=(\mathcal{W}_A,(k,n])\in \mathcal{B}_{n-k}$.
  \end{itemize}  
\end{lemma}

We remark that the definition of $k$-pure for a Boolean game $\game=(\mathcal{W},[n])$ is equivalent to
requiring $\mathcal{W}_A\in\big\{\emptyset,2^{(k,n]}\big\}$ for all $A\subseteq [k]$. If we aim to 
only perform a relatively small number of modifications to obtain a $k$-pure game $\game'$ from
a given $n$-pure game $\game$, it makes sense to modify $\mathcal{W}_A$ only if it is not contained in
$\big\{\emptyset,2^{(k,n]}\big\}$. Whenever $\mathcal{W}_A\in \mathcal{B}_{n-k}$ we have two choices for
the modification.

\begin{definition}
  \label{def_shortening}
  Let $\mathcal{V}_n$ be a class of binary games on $n$ {\voter}s. A mapping 
  $\Gamma:\mathcal{V}_n\times[n]\rightarrow\mathcal{V}_n$ is a \emph{shortening function} (on $\mathcal{V}_n$)
  if 
  \begin{enumerate}
    \item[(1)] $\game'=\Gamma(\game,k)$ is $k$-pure for all $k\in[n]$, $\game\in\mathcal{V}_n$ and
    \item[(2)] for all $U\subseteq[k]$ with $\mathcal{W}_U\in\big\{\emptyset,2^{(k,n]}\big\}$ we have
               $\game(U\cup V)=\game'(U\cup V)$ for all $V\in(k,n]$.   
  \end{enumerate}  
\end{definition}  

One special shortening function, that was already used in \cite{pre05681536}, is given by some
kind of rounding procedure:

\begin{definition}
  \label{def_k_rounding}
  For a given Boolean game $\game=(\mathcal{W},[n])$ and an integer $1\le k\le n$ we denote by $\game'=(\mathcal{W},[n])$
  the game that arises from $\game$ as follows: for every $A\subseteq [k]$ we set $\mathcal{W}'_A=\emptyset$ if 
  $|\mathcal{W}_A|\le |\overline{\mathcal{W}_A}|$ and $\mathcal{W}_A'=2^{(k,n]}$ if 
  $|\mathcal{W}_A|> |\overline{\mathcal{W}_A}|$, where $\overline{\mathcal{W}_A}=2^{(k,n]}-\mathcal{W}_A$ denotes
  the complement. We call the mapping $\Gamma$ that maps $(\game,k)$ to $\game'$ the \emph{$k$-rounding}.   
\end{definition}

The idea behind Definition~\ref{def_k_rounding} is to obtain a shortening function\footnote{The operation $k$-rounding is
indeed a shortening function $\mathcal{V}_n\times[n]\rightarrow\mathcal{V}_n$ for many of the most meaningful of classes
$\mathcal{V}_n$ of binary games mildly modified to satisfy a technical condition, see Lemma~\ref{lemma_k_rounding}
and Corollary~\ref{corollary_k_rounding}.} that modifies
the minimal number of coalitions, i.e.\ the number of switches from winning to losing, or the other way round, is
minimized. We remark that $k$-rounding of a Boolean game $\game=(\mathcal{W},[n])$ may result in a game
$\game'=(\mathcal{W}',[n])$, where $\emptyset\in \mathcal{W}'$ or $[n]\notin\mathcal{W}'$, i.e.\ that
$\game'$ is not a Boolean game. But these are the only exceptions and $k$-rounding preserves a number of meaningful
properties:  

\begin{lemma}
  \label{lemma_k_rounding}
  Let $\game=(\mathcal{W},[n])\in\mathcal{B}_n$ be a Boolean game and $1\le k\le n$ be an integer such that
  the $k$-rounding $\game'=(\mathcal{W}',[n])$ of $\game$ is not equal to either $(\emptyset,[n])$ or $(2^{[n]},[n])$,
  then we have the following implications 
  \begin{enumerate}
    \item[(1)] $\emptyset\notin \mathcal{W}'$, $[n]\in\mathcal{W}'$ $\Rightarrow$ $\game'\in\mathcal{B}_n$;
    \item[(2)] $\game$ simple $\Rightarrow$ $\game'$ simple;
    \item[(3)] $\game$ complete $\Rightarrow$ $\game'$ complete;
    \item[(4)] $\game$ weighted $\Rightarrow$ $\game'$ weighted;
    \item[(5)] $\game$ proper $\Rightarrow$ $\game'$ proper;
    \item[(6)] $\game$ strong and simple $\Rightarrow$ $\game'$ strong.
  \end{enumerate}  
\end{lemma}
\begin{proof} $\,$\\[-5mm] 
  \begin{enumerate}
    \item[(1)] Obviously $\game'$ is a function of the form $2^{N}\rightarrow \{0,1\}$ with $N=[n]$.
               The two other conditions form exactly the missing part of Definition~\ref{def_boolean_game}.
    \item[(2)] Let $S\subseteq T\subseteq [k]$ be two fixed coalitions and $U\subseteq(k,n]$ be arbitrary.
               If $S\cup U$ is winning in $(\mathcal{W},[n])$, then $T\cup U$ is winning in $(\mathcal{W},[n])$ since
               $\mathcal{W}$ is simple. Thus we have
               $$
                 \left|\left\{U\subseteq (k,n] \,:\, S\cup U\in\mathcal{W}\right\}\right|  \le
                 \left|\left\{U\subseteq (k,n] \,:\, T\cup U\in\mathcal{W}\right\}\right|. 
               $$
               If $S$ is winning in $(\mathcal{W}',[n])$, so is $T$. It remains to show $\emptyset\notin\mathcal{W}'$
               and $[n]\in\mathcal{W}'$. Assume $\emptyset\in\mathcal{W}'$, then all coalitions $\emptyset\subseteq S
               \subseteq [n]$ are winning, which contradicts $\mathcal{W}'\neq 2^{[n]}$. Similarly assume 
               $[n]\notin\mathcal{W}'$, then all coalitions $\emptyset\subseteq S\subseteq [n]$ are losing, which
               contradicts $\mathcal{W}'\neq \emptyset$. Thus $\game'=(\mathcal{W}',[n])$ is simple.
    \item[(3)] Since every complete game is simple we conclude $\game'\in\mathcal{S}_n$ from~(2).
               Now let $S,T\subseteq [k]$ be two fixed coalitions with $S\preceq T$ in $\mathcal{W}$ and
               $U\subseteq (k,n]$ be arbitrary. We have that $S\cup U\preceq T\cup U$ in $\mathcal{W}$ so that
               $$
                 \left|\left\{U\subseteq (k,n] \,:\, S\cup U\in\mathcal{W}\right\}\right|  \le
                 \left|\left\{U\subseteq (k,n] \,:\, T\cup U\in\mathcal{W}\right\}\right|. 
               $$
               Thus $S\preceq T$ in $\game'$ and $\game'=(\mathcal{W}',[n])$ is complete.
    \item[(4)] Since every weighted game is simple, we conclude $\game'\in\mathcal{S}_n$ from~(2).
               Now let $[q;w_1,\dots,w_n]$ be a weighted representation of $\game$, where we w.l.o.g.\ assume $w_i\ge 0$
               for all $1\le i\le n$. Next we use the weights $w'_1=w_1,\dots,w'_k=w_k,w'_{k+1}=0,\dots,w'_n=0$, i.e.\
               $w'(X)=w'(X\cap[k])$ for all coalitions $X\subseteq [n]$. With this we denote the minimum weight
               $w'(X)$ of a winning coalition in $X\in\mathcal{W}'$ by $u$ and the maximum weight $w'(X)$ of a
               losing coalition $X\in\mathcal{W}'$ by $l$. Let $S\subseteq [1,k]$ be a winning coalition
               $S\in\mathcal{W}'$ with weight $w'(S)=u$ and $T\subseteq [1,k]$ be a losing coalition
               $T\notin\mathcal{W}'$ with weight $w'(T)=l$. For an arbitrary coalition $V\subseteq [1,k]$,
               with $w'(V)\ge w'(S)=u$, we have $w(V)\ge w(S)=u$ and $w(V\cup U)\ge w(S\cup U)$ for all
               $U\subseteq (k,n]$, so that
               $$
                 \left|\left\{U\subseteq (k,n] \,:\, V\cup U\in\mathcal{W}\right\}\right|  \ge
                 \left|\left\{U\subseteq (k,n] \,:\, S\cup U\in\mathcal{W}\right\}\right|. 
               $$
               Thus, $V$ has to be winning in $\game'=(\mathcal{W'},[n])$ and we clearly have $l<u$.
            
               \medskip            
               
               For an arbitrary coalition $V'\subseteq [1,k]$ with $w'(V')< w'(S)=u$, we have 
               $w(V')< w(S)=u$ and coalition $V'$ is losing in $\game'$ due to the definition of $u$. Thus, 
               $w'(V')\le w'(T)=l$ due to the definition of $l$ and $[q';w_1,\dots,w_k,0,\dots,0]$ is a
               weighted representation of $\game'$ for each $q'\in(l,u]$.
    \item[(5)] Let $S\subseteq [1,k]$ be an arbitrary coalition and $T=[1,k]\backslash S$.
               By $m_1$ we denote the number of coalitions $U\subseteq (k,n]$ such that $S\cup U$ is winning in $\game$.
               Similarly, by $m_2$ we denote the number of coalitions $U\subseteq (k,n]$ such that $T\cup ((k,n]\backslash U)$
               is losing in $\game$. If $S\cup U$ is winning in $\game$, where $U\subseteq (k,n]$, then
               $T\cup ((k,n]\backslash U)$ is losing in $\game$ due to the fact that $\game$ is proper. Thus, we
               have $m_1+m_2\le 2^{n-k}$. Due to the rounding procedure, $S$ is winning in $\game'$ iff $m_1>2^{n-k-1}$.
               Similarly, $T$ is winning in $\game'$ iff $m_2>2^{n-k-1}$. Both cases cannot occur simultaneously
               so that we conclude that $\game'$ is proper.
    \item[(6)] We use the same notation as in~(5). If $S\cup U$ is losing in $\game$, where $U\subseteq (k,n]$, then
               $T\cup ((k,n]\backslash U)$ is winning in $\game$ due to the fact that $\game$ is strong.
               Thus $m_1+m_2\ge 2^{n-k}$. Due to the rounding procedure $S$ is losing in $\game'$ iff $m_1\le 2^{n-k-1}$
               and $T$ is losing in $\game'$ iff $m_2\le 2^{n-k-1}$. Thus both coalitions can be losing simultaneously
               in $\game'$ if and only if $m_1=m_2=2^{n-k-1}$, i.e.\ exactly one of the two coalitions $S\cup U$ and 
               $T\cup ((k,n]\backslash U)$ is winning in $\game$ for all $U\subseteq (k,n]$. Since $game$ is strong,
               $m_1>0$, and $\game$ is simple the coalition $S\cup(k,n]$ has to be winning in $\game$. Thus $T\cup\emptyset$
               has to be losing in $\game$. But then all coalitions $T\cup ((k,n]\backslash U)$ would be losing so that $m_2=0$.
               Thus the case $m_1=m_2=2^{n-k-1}$ is not possible for simple games and we can conclude that $\game'$ is
               strong.
  \end{enumerate}
\end{proof}

\noindent
For $|\mathcal{W}_A|=|\overline{\mathcal{W}_A}|$ the chosen tie-breaking rule in Definition~\ref{def_k_rounding} makes
the additional assumption of being a simple game in implication (6) of Lemma~\ref{lemma_k_rounding} necessary. If we
modify the tie-breaking rule in the other direction implications (1)-(4) and implication~(6), without the assumption of
being a simple game, remain valid. For implication~(5) we then need the assumption of being a simple game (or possibly a
relaxation). 

\begin{corollary}
  \label{corollary_k_rounding}
  For each integer $1\le k<n$ the operation of $k$-rounding, interpreted as a mapping 
  $\mathcal{V}_n\times[n]\rightarrow\mathcal{V}_n$, is a shortening function for 
  $$
    \mathcal{V}_n\in\Big\{
    \left\{\game\,:\, \game:2^{[n]}\rightarrow\{0,1\}\right\},
    \mathcal{S}_n\cup\mathcal{X}_n,
    \mathcal{C}_n\cup\mathcal{X}_n,
    \mathcal{T}_n\cup\mathcal{X}_n   
    \Big\},
  $$
  where $\mathcal{X}_n=\left\{(\emptyset,[n]),(2^{[n]},[n])\right\}$.   
\end{corollary}

\begin{definition}
  \label{def_locally_approximable}
  Let $\mathcal{V}_n$ be a class of binary games on $n$ {\voter}s, $\Gamma:\mathcal{V}_n\times[n]\rightarrow\mathcal{V}_n$
  be a shortening function on $\mathcal{V}_n$, and $P:\mathcal{V}_n\rightarrow \mathbb{R}^n_{\ge 0}$ be a power index
  induced by a counting function $C:\mathcal{V}_n\times 2^{[n]}\times [n]\rightarrow\mathbb{R}_{\ge 0}^n$ on $\mathcal{V}_n$.
  We say that $P$ and $C$ are \emph{locally approximable} for $\Gamma$ (on $\mathcal{V}_n$), if there exist functions
  $f_1:\mathbb{N}_{>0}\rightarrow\mathbb{R}_{\ge 0}$ and $f_2:\mathbb{N}_{>0}\rightarrow\mathbb{R}_{\ge 0}$ such that for
  all integers $1\le k<n$, all real constants $\varepsilon\ge 0$, and all games $\game=(\mathcal{W},[n])\in\mathcal{V}_n$ with 
  $\sum_{i=k+1}^n C_i(\game,2^N)\le \varepsilon \cdot \overline{C}(\game,2^N)$, we have
  \begin{enumerate}
    \item[(1)] $C_i(\game,2^N)-f_1(k)\cdot\varepsilon\cdot \overline{C}(\game,2^N)
               \le C_i(\Gamma(\game),2^N)\le 
               C_i(\game,2^N)+f_1(k)\cdot\varepsilon\cdot \overline{C}(\game,2^N)$
    \item[(2)] $\overline{C}(\game,2^N)-f_2(k)\cdot\varepsilon\cdot \overline{C}(\game,2^N)
               \le \overline{C}(\Gamma(\game),2^N)\le 
               \overline{C}(\game,2^N)+f_2(k)\cdot\varepsilon\cdot \overline{C}(\game,2^N)$ 
  \end{enumerate}
  for all $1\le i\le k$, where $N=[n]$. 
\end{definition}   

In other words, the assumptions from Definition~\ref{def_locally_approximable} say that, if for a given game $\game$
the \textit{contribution} $\sum_{i=k+1}^n C_i(\game,2^N)$ of the last $n-k$ {\voter}s is relatively small, compared to the
total \textit{contribution} $\overline{C}(\game,2^N)$, then the aggregated counting function $C_i(\Gamma(\game),2^N)$
of the \textit{shortened} game $\Gamma(\game)$ is \textit{near}\footnote{with respect to the $\Vert\cdot \Vert_1$-norm}
to the aggregated counting function $C_i(\game,2^N)$ of the game itself. 

If a given locally approximable counting function is multiplied by a positive constant factor, not depending on
the respective game, then the quality functions remain valid:

\begin{lemma}
  \label{lemma_scaled_counting_function}
  Let $\mathcal{V}_n$ be a class of binary games on $n$ {\voter}s, $\Gamma:\mathcal{V}_n\times[n]\rightarrow\mathcal{V}_n$
  be a shortening function on $\mathcal{V}_n$, and $P:\mathcal{V}_n\rightarrow \mathbb{R}^n_{\ge 0}$ be a power index
  induced by a counting function $C:\mathcal{V}_n\times 2^{[n]}\times [n]\rightarrow\mathbb{R}_{\ge 0}^n$ on $\mathcal{V}_n$.
  If $P$ and $C$ are locally approximable for $\Gamma$ on $\mathcal{V}_n$ with quality functions $f_1$ and $f_2$, then
  $C'=C\cdot\Lambda$ and the corresponding induced power index $P'=P\cdot \Lambda$ are locally approximable
  for $\Gamma$ on $\mathcal{V}_n$ with the same quality functions $f_1$ and $f_2$, where $\Lambda\in\mathbb{R}_{>0}$. 
\end{lemma}

We remark that there is a canonical choice for the function $f_2$ when we are only given a function $f_1$ satisfying the
respective requirements from Definition~\ref{def_locally_approximable}. 

\begin{lemma}
  \label{lemma_canonical_f_2}
  If $f_1$ satisfies condition~(1) of Definition~\ref{def_locally_approximable}, then $f_2(k):=k\cdot f_1(k)+1$
  satisfies condition~(2) of Definition~\ref{def_locally_approximable}.
\end{lemma}

\begin{corollary}
  \label{corollay_main_thm_without_f_2}
  Given the assumptions of Theorem~\ref{main_thm} with $f_2(k):=k\cdot f_1(k)+1$ according to Lemma~\ref{lemma_canonical_f_2},
  we have
  $$
    \Vert \widehat{P}(\game')-\widehat{P}(\game)\Vert_1
    \le (2kf_1(k)+2)\varepsilon'
  $$
\end{corollary}

For many power indices we will subsequently only state a suitable choice for $f_1$ and not for $f_2$, whenever 
we do not know a feasible function $f_2$ such that the corresponding bounds of Theorem~\ref{main_thm} are tighter
than those from Corollary~\ref{corollay_main_thm_without_f_2}. For some power indices a \textit{tighter} choice
for $f_2$ pays off. 

\begin{lemma}
  \label{lemma_quality_functions_absolute_banzhaf}
  The counting function $C^{\operatorname{Bz}}$ of the absolute Banzhaf index is locally approximable, 
  with $f_1^{\operatorname{Bz}}(k)=1$ and $f_2^{\operatorname{Bz}}(k)=k+1$, for $k$-rounding on all binary voting
  classes $\mathcal{V}_n$, where $k$-rounding is a shortening function, i.e.\ where games remain in $\mathcal{V}_n$
  after rounding.  
\end{lemma}

To ease notation we will, instead of this precise technical statement, say that $f_1^{\operatorname{Bz}}(k)=1$ and
$f_2^{\operatorname{Bz}}(k)=k+1$ are \emph{feasible} for $C^{\operatorname{Bz}}$ and $k$-rounding in the following
and in similar situations. As in Theorem~\ref{main_thm} we call $f_1$, $f_2$ \emph{quality functions}.

\begin{proof}
  (of Lemma~\ref{lemma_quality_functions_absolute_banzhaf}, compare \cite{pre05681536})\\
  Let $\game=(\mathcal{W},[n])$ be a game.
  Suppose that for $i\in (k,n]$ and $B\subseteq[n]$ we have that $i\notin B$, $B\notin\mathcal{W}$, but 
  $B\cup i\in\mathcal{W}$. Let $A=B\cap[1,k]$ and $B'=B-A$, then $i\notin B'$, $B'\notin\mathcal{W}_A$, but
  $B'\cup i\in\mathcal{W}_A$. This correspondence implies that for every $k+1\le j\le n$ we have
  $$
    C^{\operatorname{Bz}}_j(v,2^N)=\sum_{A\subseteq[k]} C^{\operatorname{Bz}}_j\Big((\mathcal{W}_A,(k,n]),2^N\Big). 
  $$
  From the edge-isoperimetric inequality for the cube-graph\footnote{See also \cite[Corollary 1]{dubey1979mathematical}.}
  one concludes
  $$
    \overline{C}^{\operatorname{Bz}}\Big((\mathcal{F},[n]),2^N\Big)\ge |\mathcal{F}|(n-\log_2|\mathcal{F}|)  
  $$
  for all $\mathcal{F}\subseteq 2^{[n]}$, so that
  $$
    \overline{C}^{\operatorname{Bz}}\Big((\mathcal{F},[n]),2^N\Big)\ge
    \min\left\{|\mathcal{F}|,\left|\overline{\mathcal{F}}|\right|\right\}.
  $$
  Thus we have
  \begin{eqnarray*}
    \varepsilon\cdot \overline{C}^{\operatorname{Bz}}(\game,2^N)&\ge&
    \sum_{i=k+1}^n C_i^{\operatorname{Bz}}(\game,2^N)
    =\sum_{i=k+1}^n\sum_{A\subseteq [k]} C^{\operatorname{Bz}}_i\Big((\mathcal{W}_A,(k,n]),2^N\Big)\\
    &\ge& \sum_{A\subseteq [k]} \min\left\{|\mathcal{W}_A|,\left|\overline{\mathcal{W}_A}\right|\right\}.
  \end{eqnarray*}
  From Lemma~\ref{lemma_swing_recursion}, see below, we then conclude the feasibility of $f_1^{\operatorname{Bz}}(k)=1$, so that
  Lemma~\ref{lemma_canonical_f_2} then gives the feasibility of 
  $f_2^{\operatorname{Bz}}(k)=kf_1^{\operatorname{Bz}}(k)+1=k+1$.
\end{proof}

\begin{lemma} \cite[Lemma~3.3.12]{0954.91019}\quad
  \label{lemma_swing_recursion}
  Let $\game=(\mathcal{W},N)$ be a simple game and $N\neq T\in \mathcal{W}$ be a minimal winning 
  coalition. For the simple game $\game'$ arising by deleting $T$ from the set of winning coalitions we have
  $$
    \eta_i(\game')=\left\{\begin{array}{rcl}\eta_i(\game)-1 &:& i\in T,\\\eta_i(\game)+1&:&i\in N\backslash T.\end{array}\right.
  $$   
\end{lemma}

We remark that for all of the power indices introduced in Section~\ref{sec_power_indices}, except for the nucleolus
and the MSR index, similar recursion formulas as in Lemma~\ref{lemma_swing_recursion} exist.

From Lemma~\ref{lemma_scaled_counting_function} and Lemma~\ref{lemma_quality_functions_absolute_banzhaf} we conclude:
\begin{corollary}
  \label{cor_quality_functions_Banzhaf_score}
  The quality functions $f_1(k)=1$ and $f_2(k)=k+1$ are \emph{feasible} for
  the vector of the number of $i$-swings $(\eta_1,\dots, \eta_n)$, i.e.\ the \textit{Banzhaf score}, given
  by the counting function 
  $$
    C_i(\game,S)=\left\{\begin{array}{rcl}1&:&i\in S,\game(S)=1,\game(S-i)=0,\\
    0&:&\text{otherwise},\end{array}\right.
  $$ and $k$-rounding.
\end{corollary}

\begin{lemma}
  \label{lemma_additive_shift}
  Let the quality functions $f_1$ and $f_2$ be feasible for a power index $P$ induced by a counting function $C$ and
  $\alpha\in\mathbb{R}_{\ge 0}$ be a constant. The shifted counting function $C'_i(\game,S)=C_i(\game,S)+\frac{\alpha}{2^n}$,
  where $n$ is the number of {\voter}s in $\game$, induces the power index $P'=P+\alpha$ and $f_1$, $f_2$ are feasible for $C'$
  and $P'$ too.
\end{lemma}
\begin{proof}
  At first we observe 
  $$
    P'_i(\game)=\sum_{S\subseteq [n]} C'_i(\game,S)=\alpha+\sum_{S\subseteq [n]} C_i(\game,S)=\alpha+P_i(\game).
  $$ 
  With the notation of Definition~\ref{def_locally_approximable} we have
  \begin{eqnarray*}
    C_i'(\Gamma(\game),2^N) &=&C_i(\Gamma(\game),2^N)+\alpha \\
    &\le& C_i(\game,2^N)+\alpha+f_1(k)\cdot\varepsilon\cdot \overline{C}(\game,2^N)\\
    &=& C'_i(\game,2^N)+f_1(k)\cdot\varepsilon\cdot \overline{C}(\game,2^N)\\
    &\le & C'_i(\game,2^N)+f_1(k)\cdot\varepsilon\cdot \overline{C}'(\game,2^N),
  \end{eqnarray*}
  since $f_1(k),\varepsilon,\alpha\ge 0$. Similarly we obtain
  $$
    C_i'(\Gamma(\game),2^N) \ge C'_i(\game,2^N)-f_1(k)\cdot\varepsilon\cdot \overline{C}'(\game,2^N)
  $$
  and
  $$
    \overline{C}'(\game,2^N)-f_2(k)\cdot\varepsilon\cdot \overline{C}'(\game,2^N)
               \le \overline{C}'(\Gamma(\game),2^N)\le 
               \overline{C}'(\game,2^N)+f_2(k)\cdot\varepsilon\cdot \overline{C}'(\game,2^N).
  $$
\end{proof}

In combination Lemma~\ref{lemma_scaled_counting_function} and Lemma~\ref{lemma_additive_shift} say that 
we can transfer quality functions for a given power index $P$, induced by counting function, to
all power indices $P'=\alpha+\beta P$, where $\alpha\in\mathbb{R}_{\ge 0}$ and $\beta\in\mathbb{R}_{>0}$.
As an application we mention the following corollary based on the identity 
$\operatorname{Rae}=\frac{1}{2}+\frac{1}{2}\cdot\operatorname{Bz}$:

\begin{corollary}
  \label{cor_quality_functions_Rae_index}
  The quality functions $f_1(k)=1$ and $f_2(k)=k+1$ are \emph{feasible} for
  the Rae index and $k$-rounding.
\end{corollary}

In Subsection~\ref{subsec_quality_functions} we will state the respective quality functions, so that
Theorem~\ref{main_thm} can be applied, for most of the power indices introduced in Section~\ref{sec_power_indices}. Most
of the proofs, based on elementary counting and bookkeeping considerations, will be delayed to the appendix. For $p$-binomial
semivalues it will turn out that no such result based on $k$-rounding
can exist, for the Johnston index even the larger class of shortening functions admits no such result. In order to prove
those negative results, we introduce a parameterized class of weighted games in Subsection~\ref{subsec_parameterized_wvg}.
Now we proceed with an application of Theorem~\ref{main_thm} in the next subsection.

\subsection{Applications}
\label{subsec_applications}

We would like to demonstrate how to apply Theorem~\ref{main_thm} to deduce some information on power vectors that
are hard to approximate. Assume that we want to find a simple game with normalized Banzhaf vector 
$\sigma_n=\left(\frac{3}{4},\frac{1}{4},\dots\right)$, where the number of {\voter}s $n\ge 2 $ is arbitrary. It will
turn out that no simple game can exactly meet the desired power distribution $\sigma_n$, so that we aim at 
minimizing the $\Vert\cdot\Vert_1$-distance between $\sigma_n$ and the realized Banzhaf distribution 
$\operatorname{Bz}(\game)$. The idea is to use negative approximation results for games on $k<n$ {\voter}s 
and then to apply Theorem~\ref{main_thm} to conclude lower bounds for arbitrary $n$. For this purpose
we need to bridge the gap between a $k$-pure game on $n$ {\voter}s (as in Theorem~\ref{main_thm}) and
the resulting $k$-{\voter} game after removing the null {\voter}s.

\begin{definition}
  \label{def_null_player_removable}
  A power index $P:\mathcal{V}_i\rightarrow \mathbb{R}^i_{\ge 0}$ is called \emph{null {\voter} removable} if
  for all $\game\in\mathcal{V}_n$ and $\game'\in\mathcal{V}_k$, arising from $\game$ by deleting the null {\voter}s,
  we have $P_j(\game)=P_j(\game')$ for all {\voter}s~$j$, which are not null {\voter}s in $\game$.   
\end{definition}

Both, the normalized and the absolute Banzhaf index are null {\voter} removable while the index based on the number of
$i$-swings, is not. For the later case the number of $i$-swings for all non-null {\voter}s
doubles for each added extra null {\voter}.

\begin{theorem}
  \label{thm_approximation_bound}
  Let $0<k<n$ be integers, $\sigma\in\mathbb{R}_{\ge 0}^n$ be a desired power vector with $\Vert\sigma\Vert_1=1$,
  $P$ be a power index induced by a counting function $C$ that is locally approximable for the shortening function 
  $\Gamma:\mathcal{V}_n\rightarrow\mathcal{V}_n$ with quality functions $f_1$ and $f_2$. If $P$
  further is \textit{positive}, satisfies the \textit{null {\voter} property}, $\widehat{P}$ is
  \textit{null {\voter} removable}, and $\min\left\{\frac{\Lambda+\alpha}{\beta+2},\frac{1}{f_2(k)}\right\}\ge\alpha$,
  then we have
  $$
    \Vert \widehat{P}(\game)-\sigma\Vert_1\ge\min\left\{2\cdot\left(\frac{1}{f_2(k)}-\alpha\right),
    \frac{2}{\beta+2}\cdot\Lambda-\frac{2\beta+2}{\beta+2}\cdot\alpha\right\},
  $$
  where $\sigma'=(\sigma_1,\dots,\sigma_k)$, $\alpha=1-\Vert\sigma'\Vert_1$, $\beta=f_2(k)+kf_1(k)+1$,
  and $\Lambda=\min_{\game'\in\mathcal{V}_k} \Vert \sigma'-\widehat{P}(\game')\Vert_1$, for all $\game\in\mathcal{V}_n$.
\end{theorem}
\begin{proof}
  Let $\game\in\mathcal{V}_n$ be arbitrary and $\game'=\Gamma(\game)\in\mathcal{V}_k$. For
  $\varepsilon<\frac{1}{f_2(k)}$ with $\sum\limits_{i=k+1}^{i} \widehat{P}_i(\game)\le \varepsilon$ 
  we obtain
  $$
    \Vert \widehat{P}(\game')-\widehat{P}(\game)\Vert_1\le (f_2(k)+kf_1(k)+1)\varepsilon
  $$
  from Theorem~\ref{main_thm}. Applying the triangle inequality for $\Vert\cdot\Vert_1$ gives
  $$
    \Vert \widehat{P}(\game)-\sigma\Vert_1\ge \Vert \widehat{P}(\game')-\sigma'\Vert_1
    -\Vert \widehat{P}(\game')-\widehat{P}(\game)\Vert_1
    - (1-\Vert\sigma'\Vert_1).
  $$ 
  If otherwise $\sum\limits_{i=k+1}^{i} \widehat{P}_i(\game)> \varepsilon$, then
  $$
    \Vert \widehat{P}(\game)-\sigma\Vert_1=\sum_{i=1}^k\left|\widehat{P}_i(\game)-\sigma_i\right|+
    \sum_{i=k+1}^n\Big|\widehat{P}_i(\game)-\sigma_i\Big|\ge
    2\cdot\max\left\{\Vert\sigma'\Vert_1-1+\varepsilon,0\right\}.
  $$
  Now it remains to choose a suitable $\varepsilon$ such that the weakest of the two lower bounds is
  maximized. If $\frac{\Lambda+\alpha}{\beta+2}<\frac{1}{f_2(k)}$, then we can choose 
  $\varepsilon=\frac{\Lambda+\alpha}{\beta+2}$ so that both right hand sides are equal to 
  $\frac{2}{\beta+2}\cdot\Lambda-\frac{2\beta+2}{\beta+2}\cdot\alpha$. Now assume 
  $\frac{\Lambda+\alpha}{\beta+2}\ge\frac{1}{f_2(k)}$. If $\alpha=\frac{1}{f_2(k)}$, then the proposed
  inequality is true due to $\Vert \widehat{P}(\game)-\sigma\Vert_1\ge 0$, so that we can additionally
  assume $\frac{1}{f_2(k)}>\alpha$. Finally we consider
  $$
    \lim_{\varepsilon\to \frac{1}{f_2(k)}:\alpha<\varepsilon<\frac{1}{f_2(k)}}
    \min\Big\{\Lambda-\alpha-\beta\epsilon,2(\varepsilon-\alpha)\Big\}
    =2\cdot\left(\frac{1}{f_2(k)}-\alpha\right).
  $$    
\end{proof}

As an example we choose $n\ge 3$, $\sigma=\left(\frac{29}{40},\frac{9}{40},\frac{1}{20(n-2)},\dots,\frac{1}{20(n-2)}\right)$, 
$P=\operatorname{Bz}$, $\Gamma$ as $k$-rounding and
$\mathcal{V}_i=\mathcal{S}_i\cup\left\{(\emptyset,[i]),(2^{[i]},[i])\right\}$, i.e.\ the set of simple
games supplemented by the two trivial non-Boolean games for technical reasons. For $k=2$ we have
$\alpha=\frac{1}{20}$ and $\Lambda=\frac{1}{2}$, since the only achievable Banzhaf distributions
are $(1,0)$, $\left(\frac{1}{2},\frac{1}{2}\right)$, and $(0,1)$. From
Lemma~\ref{lemma_quality_functions_absolute_banzhaf} we read of $\beta=6$ and $f_2(2)=3$. Inserting
yields $\Vert \widehat{\operatorname{Bz}}(\game)-\sigma\Vert_1\ge\frac{3}{80}$ for all $\game\in\mathcal{S}_n$,
i.e.\ $\sigma$ cannot be approximated too well by Banzhaf distributions within the class of simple games.

\begin{corollary}
  \label{cor_approximation_bound}
  Let $1<k<n$ be integers, $\sigma=\left(\sigma_1,\dots,\sigma_k,0,\dots,0\right)\in\mathbb{R}_{\ge 0}^n$ be
  a desired power vector with $\Vert\sigma\Vert_1=1$, $P\in\{\operatorname{Bz},%\operatorname{Rae}, NOT positive, ...
  \operatorname{ColPrev},\operatorname{ColIni},\operatorname{PGI},\operatorname{DP},\operatorname{Shift},\operatorname{SDP}\}$
  be a power index, $\mathcal{V}\in\left\{\mathcal{S},\mathcal{C},\mathcal{T}\right\}$ be a class of games, then
  $$
    \Vert \widehat{P}(\game)-\sigma\Vert_1\ge\min\left\{\frac{2}{f_2^P(k)},
    \frac{2\Lambda}{f_2^P(k)+kf_1^P(k)+3}\right\},
  $$
  $\sigma'=(\sigma_1,\dots,\sigma_k)$, $\Lambda=\min_{\game'\in\mathcal{V}_k} \Vert \sigma'-\widehat{P}(\game')\Vert_1$
  and the corresponding $f_1^P$, $f_2^P$ can be found in Table~\ref{table_quality_functions}, for all $\game\in\mathcal{V}_n$.
\end{corollary}
\begin{proof}
  For $\mathcal{V}_i'=\mathcal{V}_i\cup\left\{(\emptyset,[i]),(2^{[i]},[i])\right\}$ the operation $k$-rounding
  is a shortening function $\mathcal{V}_n'\rightarrow \mathcal{V}_i'$. For $\game\in\left\{(\emptyset,[n]),
  (2^{[n]},[n])\right\}$ we have $\Vert \widehat{P}(\game)-\sigma\Vert_1=1$ so that the proposed bound is valid
  due to $f_2^P(2)\ge 2$ for all cases of $P$. Thus we can assume $\game\in\mathcal{V}_n$. With the
  notation from Theorem~\ref{thm_approximation_bound} we have $\alpha=0$ and all requirements on $P$ are
  satisfied. Since $\min\left\{\frac{\Lambda+\alpha}{\beta+2},\frac{1}{f_2(k)}\right\}\ge 0$ and  
  $\min_{\game'\in\mathcal{V}_k} \Vert \sigma'-\widehat{P}(\game')\Vert_1
  \ge \min_{\game'\in\mathcal{V}_k'} \Vert \sigma'-\widehat{P}(\game')\Vert_1$ we can conclude the proposed bound
  from the bound of Theorem~\ref{thm_approximation_bound}.
\end{proof}

As an example we consider $\sigma_n=\left(\frac{3}{4},\frac{1}{4},0,\dots,0\right)$ for $n\ge 3$, $P=\operatorname{Bz}$,
and $k=2$. Corollary~\ref{cor_approximation_bound} then yields 
$\Vert \widehat{\operatorname{Bz}}(\game)-\sigma\Vert_1\ge\frac{1}{8}$ for all simple games $\game\in\mathcal{S}_n$. 
We remark that in \cite{kurz2012inverse} the bound $\Vert \widehat{\operatorname{Bz}}(\game)-\sigma\Vert_1\ge\frac{1}{9}$ 
was proven based on the original Alon-Edelman bound. Computational results for the corresponding  exact solutions of the
inverse power index problem for simple games with up to $11$ {\voter}s lead to the conjecture 
$\Vert \operatorname{Bz}(\game)-\sigma_n\Vert_1\ge \frac{1}{3}$ (and indeed a even tighter one).

\subsection{Quality functions for several power indices}
\label{subsec_quality_functions}

Having the theoretical framework of the preceding subsections at hand it remains to prove the validity of suitable
quality functions for the different power indices in order to conclude the respective Alon-Edelman type results.
In contrast to the proof of Lemma~\ref{lemma_quality_functions_absolute_banzhaf} for the Banzhaf index,
where the edge-isoperimetric inequality for the cube-graph was used, we will only need  elementary counting 
and bookkeeping arguments. Exemplarily, we state the proof for the quality functions for the Public Good Index. In the 
following we will abbreviate the term minimal winning coalition by MWC.
 
\begin{lemma}
  \label{lemma_quality_functions_PGI}
  For the Public Good Index we can choose $f_1(k)=k+1$ and $f_2(k)=\frac{(k+1)(k+5)}{4}$.
\end{lemma}
\begin{proof}
  Let $\game=(\mathcal{W},[n])$ be an arbitrary simple game. As defined before, $C_i^{\operatorname{PGI}}(\game,2^{[n]})$ 
  counts the number of minimal winning coalitions in $\game$ containing {\voter}~$i$ and we have 
  $\overline{C}^{\operatorname{PGI}}(\game,2^{[n]})=\sum_{i=1}^n C_i^{\operatorname{PGI}}(\game,2^{[n]})$. For brevity we just 
  write $C_i$ and $\overline{C}$. As further abbreviations we use $\game'=(\mathcal{W}',[n])$ for the $k$-rounding of $\game$,  
  $\widehat{C}=\sum_{i=k+1}^n C_i$, i.e.\ the restriction of $\overline{C}$ to the {\voter}s in $(k,n]$, and by 
  $\widehat{M}$ we denote the number of minimal winning coalitions in $\mathcal{W}$ that contain at least one {\voter}
  from $(k,n]$. With this we have $\widehat{M}\le\widehat{C}$.    
  
  Now we want to study the changes in the number $C_i$ of minimal winning coalitions by going from $\game$ to 
  $\game'$. At first we consider the cases where a coalition $S$ is a MWC in $\game$ but not in $\game'$.
  \begin{enumerate}
    \item[(1)] If $S\cap (k,n]\neq\emptyset$, then $S$ is a MWC containing a {\voter} from $(k,n]$, so that for a 
               given index $1\le i\le k$, $C_i$ decreases by at most $\widehat{M}$.
    \item[(2)] Let us assume $S\subseteq [k]$. Since $S$ is a MWC in $\game$, it is indeed winning in $\game$. Thus 
               $S\cup T$ is winning in $\game$ for all $T \subseteq (k,n]$, so that $S$ is also winning in $\game'$.
               Since $S$ is not a MWC in $\game'$, there exists a {\voter}~$j\in S$ such that $S-j$ is also winning in
               $\game'$. This can only happen if there is a subset $\emptyset\neq T\subseteq (k,n]$ such that $S-j+T$
               is a MWC in $\game$. So the change of a MWC $S-j+T$ can cause the change of a MWC $S$. Now we want to count
               how often this can happen. To this end we set $r:=|S-j|$ and remark $1\le r\le k-1$, since $\emptyset$ can
               not be winning in $\game'$ (our conditions ensure that also $\game'$ is a simple game). 
               If $i\in S-j$, then we have $k-r$ possibilities and one possibility otherwise. Thus, for a 
               given index $1\le i\le k$, $C_i$ decreases by at most $\widehat{M}(k-1)$. 

               \medskip               
               
               The overall decrease of all $C_i$ with $1\le i\le k$ can be bounded a bit tighter than $\widehat{M}(k-1)k$
               as follows. For $S$ we have $k-r$ possibilities for the $r$ {\voter}s in $S-j$ and $1$ possibility for the 
               remaining $k-r$ {\voter}s in $[k]$. Since $(r+1)(k-r)\le\frac{(k+1)^2}{4}$, the overall decrease is at most
               $\widehat{M}\cdot\frac{(k+1)^2}{4}$.
  \end{enumerate}      
  
  \noindent
  Next we consider the cases where a coalition $S$ is not a MWC in $\game$ but a MWC in $\game'$. Since
  the {\voter}s in $(k,n]$ are null {\voter}s in $\game'$, we deduce $S\subseteq [k]$.
  \begin{enumerate}
    \item[(3)] Assume that $S$ is losing in $\game$ but winning in $\game'$. According to the rounding procedure, there exists
               a MWC $S\cup T$ in $\game$, where $T\subseteq (k,n]$. Thus, for a given index $1\le i\le k$, $C_i$ increases by
               at most $\widehat{M}$.  
    \item[(4)] If $S$ is winning in $\game$, then there must be a {\voter}~$j\in S$ so that $S-j$ is also winning in $\game$,
               but losing in $\game'$. According to the rounding procedure, there exists a MWC $S-j\cup T$ in $\game$,
               where $T\subseteq (k,n]$. We proceed similarly as in case~(2) and set $r:=|S-j|$, so that $0\le r\le k-1$. 
               Thus, for a given index $1\le i\le k$, $C_i$ increases by at most $\widehat{M}\cdot k$ and the overall increase is 
               at most $\widehat{M}\cdot\frac{(k+1)^2}{4}$.
  \end{enumerate}
  Summarizing the four cases gives
  $$
    C_i^{\operatorname{PGI}}(\game,2^N)-(k+1)\widehat{C}\le 
    C_i^{\operatorname{PGI}}(\game,2^N)-k\cdot\widehat{M}\le C_i^{\operatorname{PGI}}(\game',2^N)
  $$
  and
  $$  
    C_i^{\operatorname{PGI}}(\game',2^N)
    \le C_i^{\operatorname{PGI}}(\game,2^N)+(k+1)\widehat{M}\le C_i^{\operatorname{PGI}}(\game,2^N)+(k+1)\widehat{C} 
  $$   
  for all $1\le i\le k$. Since $\widehat{C}\le \varepsilon\cdot \overline{C}^{\operatorname{PGI}}(\game,2^N)$, we can
  choose $f_1(k)=k+1$.
  
  Using $\widehat{C}\le \varepsilon\cdot \overline{C}^{\operatorname{PGI}}(\game,2^N)$ and 
  $\sum\limits_{i=k+1}^n C_i^{\operatorname{PGI}}(\game',2^n)=0$,
  we obtain
  $$
    \overline{C}^{\operatorname{PGI}}(\game,2^N)-\frac{(k+1)(k+5)-4}{4}\cdot\widehat{M}
    \le \overline{C}^{\operatorname{PGI}}(\game',2^N)
  $$
  and
  $$
    \overline{C}^{\operatorname{PGI}}(\game',2^N)  
    \le \overline{C}^{\operatorname{PGI}}(\game,2^N)+ \frac{(k+1)(k+5)-4}{4}\cdot\widehat{M}
    +\varepsilon\cdot \overline{C}^{\operatorname{PGI}}(\game,2^N),
  $$  
  so that we can choose $f_2(k)=\frac{(k+1)(k+5)}{4}$.
\end{proof}

In Table~\ref{table_quality_functions} we state suitable quality functions for 11~different power indices,
whose proofs all are rather similar to that of Lemma~\ref{lemma_quality_functions_PGI} and therefore shifted
to Appendix~\ref{sec_details_quality_functions}. We have to remark, that for the Tijs index we do not use
$k$-rounding but $k$-up-rounding, i.e.\ we set $\mathcal{W}_A'=2^{(k,n]}$ if $\left|\mathcal{W}_A\right|\ge 1$
and $\mathcal{W}_A'=\emptyset$ (using the notation from Definition~\ref{def_k_rounding}).

\begin{table}[htp!]
  \begin{center}
    \begin{tabular}{rrrr}
      \hline
      \textbf{power index} & $\mathbf{f_1(k)}$ & $\mathbf{f_2(k)}$ & \textbf{justification} \\  
      \hline
      $\operatorname{Tijs}$    & $0$    & $1$                    & Lemma~\ref{lemma_quality_functions_Tijs} \\
      $\operatorname{Bz}$      & $1$    & $k+1$                  & Lemma~\ref{lemma_quality_functions_absolute_banzhaf} \\
      $\operatorname{swing}=2^{n-1}\cdot\operatorname{Bz}$ & $1$ & $k+1$ & Lemma~\ref{lemma_scaled_counting_function} \\
      $\Psi^P$ with $c_1\le P_i\le c_2$,\\
      where $c_1,c_2\in\mathbb{R}_{>0}$, $0\le i\le n-1$ & $\frac{c_2}{c_1}$ & $\frac{c_2}{c_1}\cdot(k+1)$ 
      & Lemma~\ref{lemma_quality_functions_bounded_semivalue}\\
      $\operatorname{ColPrev}$ & $2$             & $2k+1$                 & Lemma~\ref{lemma_quality_functions_ColPrev} \\
      $\operatorname{ColIni}$  & $2$             & $2k+1$                 & Lemma~\ref{lemma_quality_functions_ColIni} \\
      $\operatorname{Rae}$     & $1$             & $k+1$                  & Lemma~\ref{lemma_scaled_counting_function} and~\ref{lemma_additive_shift} (or Corollary~\ref{cor_quality_functions_Rae_index})\\
      $\operatorname{PGI}$     & $k+1$           & $\frac{(k+1)(k+5)}{4}$ & Lemma~\ref{lemma_quality_functions_PGI}\\
      $\operatorname{DP}$      & $k+2$           & $\frac{k^2+6k+9}{4}$   & Lemma~\ref{lemma_quality_functions_DP}\\
      $\operatorname{Shift}$   & $\frac{k+3}{2}$ & $\frac{k^2+3k+2}{2}$   & Lemma~\ref{lemma_quality_functions_Shift}\\
      $\operatorname{SDP}$     & $2k+1$          & $2k^2+k+1$             & Lemma~\ref{lemma_quality_functions_SDP}\\
      \hline
    \end{tabular}
    \caption{Quality functions for several power indices}
    \label{table_quality_functions}
  \end{center}
\end{table}

\subsection{A parameterized class of weighted games and negative results}
\label{subsec_parameterized_wvg}

In this section we construct a parametric example in order to obtain negative results for $p$-binomial semivalues
and the Johnston index with respect to the existence of Alon-Edelman type results. 

\begin{definition}
  \label{def_parametric_game}
  For integers $l,k,m,n$ with $1\le l\le k-1$, $n\ge 1$, and $0\le m\le n+1$, we define the 
  simple game $\game_{n,m}^{k,l}$ given by its set of minimal winning coalitions
  $$
    \{U\subseteq [k] : |U|=l,\, U\neq T\}\cup \{T\cup V : V\subseteq (k,k+n],\, |V|=m\},
  $$
  where $T=(k-l,k]$.
\end{definition}

We remark that the number of {\voter}s of $\game_{n,m}^{k,l}$ is given by $n+k$ and not by $n$ (as everywhere
else within this paper). The {\voter}s come in at most three different \textit{types}.

\begin{lemma}
  \label{lemma_parametric_game_weighted} 
  For integers $l,k,m,n$ satisfying the restrictions from Definition~\ref{def_parametric_game} the
  simple game $\game_{n,m}^{k,l}$ is weighted.
\end{lemma}
\begin{proof}
  For $m\notin\{0,n+1\}$ a weighted representation is given by 
  $[q;w_1=a,\dots,w_{k-l}=a,w_{k-l+1}=b,\dots,w_k=b,w_{k+1}=1,\dots,w_{n+k}=1]$,
  where $b=(l-2)m+n+1$, $a=(l-1)m+n+1$, and $q=lb+m=(l-1)^2m+l(n+1)$. For $m=n+1$ a weighted representation is
  given by $[q=l^2+1;w_1=l+1,\dots,w_{k-l}=l+1,w_{k-l+1}=l,\dots,w_k=l,w_{k+1}=0,\dots,w_{n+k}=0]$ and for
  $m=0$ a weighted representation is given by $[q=l;w_1=1,\dots,w_{k-l}=1,w_{k-l+1}=1,\dots,w_k=1,
  w_{k+1}=0,\dots,w_{n+k}=0]$.
\end{proof}

The family of parametric games $\game_{n,m}^{k,l}$ has several useful properties: As just mention the games weighted,
i.e.\ belong to the most narrow class $\mathcal{T}_n$ of binary games that we are considering in this paper. With the
help of the parameters $l,k,m,n$ one can, to some degree, control the \textit{power} of the {\voter}s in $[1,k-k]$
and of those in $(k,k+n]$. The games $\game_{n,m}^{k,l}$ are almost $k$-pure, i.e.\ $T$ is the only coalition
$S\subseteq [k]$ such that the reduced game $\mathcal{W}_S$ is a simple game. In other words, each shortening function
can just modify coalitions of the form $T\cup V$, where $V\subseteq (k,k+n]$. Since in a $k$-pure game it only depends
on $T$ whether the $T\cup V$ are winning or losing, there are just two possible choices. Turning $T\cup V$ into losing
coalitions for all $V\subseteq (k,k+n]$ results in $\game_{n,n+1}^{k,l}$ and turning $T\cup V$ into losing winning for
all $V\subseteq (k,k+n]$ results in $\game_{n,0}^{k,l}$. For the special shortening function $k$-rounding we can explicitly
state which case occurs:

\begin{lemma}
  Let $\Gamma$ denote the operation of $k$-rounding and $l,k,m,n$ be integers satisfying the restrictions from
  Definition~\ref{def_parametric_game}. Then we have
  $$
    \Gamma\left(\game_{n,m}^{k,l}\right)=\left\{
    \begin{array}{rcl}
      \game_{n,0}^{k,l} &:& m\le \lfloor n/2\rfloor,\\
      \game_{n,n+1}^{k,l} &:& \text{otherwise.}\\
    \end{array}
    \right.
  $$ 
\end{lemma}

We want to study the absolute difference of individual power between the three games $\Gone$, $\Gtwo$, and
$\Gthree$ for the (at most) three types of {\voter}s $[1,k-l]$, $(k-l,j]$, and $(k,k+n]$. As an abbreviation we
introduce:
\begin{definition}
  \label{def_delta}
  For each positive power index $P:\mathcal{V}_n\rightarrow\mathbb{R}_{\ge 0}$ satisfying the null {\voter} property and
  each integers $l,k,m,n$ satisfying the restrictions from Definition~\ref{def_parametric_game}, we set
  \begin{eqnarray*}
    \Delta^{1,P}_{1,2} &=& \left|P_i(\Gone)-P_i(\Gtwo)\right|\quad\forall i\in [1,k-l], \\
    \Delta^{2,P}_{1,2} &=& \left|P_i(\Gone)-P_i(\Gtwo)\right|\quad\forall i\in (k-1,k], \\
    \Delta^{3,P}_{1,2} &=& \left|P_i(\Gone)-P_i(\Gtwo)\right|\quad\forall i\in (k,k+n], \\
    \Delta^{1,P}_{1,3} &=& \left|P_i(\Gone)-P_i(\Gthree)\right|\quad\quad\,\forall i\in [1,k-l], \\
    \Delta^{2,P}_{1,3} &=& \left|P_i(\Gone)-P_i(\Gthree)\right|\quad\quad\,\forall i\in (k-1,k], \\
    \Delta^{3,P}_{1,3} &=& \left|P_i(\Gone)-P_i(\Gthree)\right|\quad\quad\,\forall i\in (k,k+n],\text{ and} \\
    \xi^P              &=& \sum_{i=k+1}^{k+n} P_i(\Gone).
  \end{eqnarray*}
\end{definition} 

In other words, $\xi^P$ is the aggregated power of the {\voter}s in $(k,k+n]$ of the game $\Gone$. For $\Gtwo$ and
$\Gthree$ the corresponding aggregated power equals zero so that we have $\xi^P=n\Delta^{3,P}_{1,2}=n\Delta^{3,P}_{1,3}$.

Our first aim is to show that there can be no Alon-Edelman type result for $p$-binomial semivalues with $p\neq\frac{1}{2}$
and $k$-rounding, i.e.\ the Banzhaf index is unique within the class of $p$-binomial semivalues with respect
to this property. To this end we prove the following limit results in Subsection~\ref{subsec_technical_p_binomial}.

\begin{lemma}
  \label{lemma_limit_results_p_binomial}
  Let $P=\Psi^p$ and $k,l$ be integers with $1\le l\le k-1$.
  \begin{itemize}
    \item[(1)] For $m=\left\lceil\frac{n+1}{2}\right\rceil$, $p\in \left(\frac{1}{2},1\right)$ we have
               $\lim\limits_{n\to\infty} \frac{\Delta_{1,2}^{1,P}}{\xi^P}\to\infty$ and
               $\lim\limits_{n\to\infty} \frac{\Delta_{1,2}^{2,P}}{\xi^P}\to\infty$.
    \item[(2)] For $m=\left\lfloor\frac{n}{2}\right\rfloor$, $p\in \left(0,\frac{1}{2}\right)$ we have
               $\lim\limits_{n\to\infty} \frac{\Delta_{1,3}^{1,P}}{\xi^P}\to\infty$ and
               $\lim\limits_{n\to\infty} \frac{\Delta_{1,2}^{3,P}}{\xi^P}\to\infty$.
    \item[(3)] For $m\in\big\{\left\lceil\frac{n+1}{2}\right\rceil,\left\lfloor\frac{n}{2}\right\rfloor\big\}$
               and $p\in(0,1)\backslash\left\{\frac{1}{2}\right\}$ we have $\lim\limits_{n\to\infty} \xi^P=0$.
  \end{itemize}
\end{lemma}
We remark that (3) is not valid for $p=\frac{1}{2}$. Next we bound the overall $\Vert\cdot\Vert_1$-distance
between $\Gone$ and $\Gtwo$ or $\Gthree$:

\begin{lemma}
\label{lemma_bounded_change_p_binomial}
 For $p\in (0,1)\backslash\left\{\frac{1}{2}\right\}$ and integers $l,k,m,n$ satisfying the restrictions from
 Definition~\ref{def_parametric_game} we have
 $$
   \left|\sum_{i=1}^{n+k} \Psi_i^p(\Gone)-\sum_{i=1}^{n+k} \Psi_i^p(\Gtwo)\right|\le k\cdot p^{l-1}(1-p)^{k-l-1}+\xi^{\Psi^p}
 $$
 and
 $$
   \left|\sum_{i=1}^{n+k} \Psi_i^p(\Gone)-\sum_{i=1}^{n+k} \Psi_i^p(\Gthree)\right|\le k\cdot p^{l-1}(1-p)^{k-l-1}+\xi^{\Psi^p}.
 $$
 \end{lemma}   
 
 For the proof we refer the reader to Subsection~\ref{subsec_technical_p_binomial}. By suitably choosing $k$ and $l$ 
 we can achieve that the bound of Lemma~\ref{lemma_bounded_change_p_binomial} guarantees a small chance between
 $\Gone$ and its $k$-rounding. $l=\left\lceil\frac{k}{2}\right\rceil$ and $k$ sufficiently large but fix will do the
 job. Next we can increase $n$ and apply Lemma~\ref{lemma_limit_results_p_binomial}. Since the values $f_1(k)$, 
 $f_2(k)$ of the hypothetical quality functions would be finite and we can choose $\varepsilon=\xi^P$, an Alon-Edelman
 type result cannot exist for $k$-rounding.    
 
 \begin{theorem}
   \label{thm_no_alon_edelman_type_result_for_p_binomial_and_k_rounding}
   For $k\in\mathbb{N}$, $p\in(0,1)\backslash\left\{\frac{1}{2}\right\}$, $\mathcal{T}_n\subseteq\mathcal{V}_n$ being compatible
   with $k$-rounding, no Alon-Edelman type result for $\Psi^p$ and $k$-rounding can exist.
 \end{theorem}  
 
 Since the changes can be made arbitrarily small, due to Lemma~\ref{lemma_bounded_change_p_binomial}, the corresponding
 statement for the normalized $p$-binomial semivalue is also valid.
 Having a closer look at the proofs of the technical bounds in Subsection~\ref{subsec_technical_p_binomial} we observe
 that the argumentation will break down if we slightly adopt the shortening function $k$-rounding.
 
 \begin{definition}
   \label{def_p_k_rounding}
  For a given Boolean game $\game=(\mathcal{W},[n])$, a real constant $p\in(0,1)$, and an integer $1\le k\le n$ we
  denote by $\game'=(\mathcal{W},[n])$ the game that arises from $\game$ as follows: for every $A\subseteq [k]$ we
  set $\mathcal{W}'_A=\emptyset$ if $|\mathcal{W}_A|\le p\cdot 2^{n-k}$ and $\mathcal{W}_A'=2^{(k,n]}$ if 
  $|\mathcal{W}_A|> p\cdot 2^{n-k}$. We call the mapping $\Gamma$ that maps $(\game,k)$ to $\game'$ the \emph{$(p,k)$-rounding}. 
 \end{definition}
 
In this new notation our old $k$-rounding is denoted as $\left(\frac{1}{2},k\right)$-rounding. There might be an intuition
for this special choice. Taking as $p$ the probability for voting yes for each {\voter}, the expected number of winning
coalitions in $\mathcal{W}_A$ would be $p\cdot 2^{n-k}$. In some sense one would compare the present situation of the reduced
game to the expected situation and rounds with respect to this threshold. Admittedly, this interpretation might not carry 
very far or even be wrong. Nevertheless we state:

\begin{conjecture}
  \label{conj_alon_edelmann_for_p_binomial}
  For $k\in\mathbb{N}$, $p\in(0,1)$, $\mathcal{T}_n\subseteq\mathcal{V}_n$ being compatible
  with $(p,k)$-rounding, Alon-Edelman type results for $\Psi^p$ and $(p,k)$-rounding do exist.
\end{conjecture} 

Although we have tested the   validity of Conjecture~\ref{conj_alon_edelmann_for_p_binomial} for several
examples, it remains unclear, if just counting the number of winning coalitions in the respective reduces
games is detailed enough so that $(p,k)$-rounding works as expected. We remark that at the very least, there
are generalizations of the edge-isoperimetric inequality for the cube-graph with a binomial probability distribution
with expectation $p$, see e.g.\ \cite{ellis2011almost,falik2007edge}.

We want to complete this subsection with the discussion of the Johnston index, which does, interestingly enough,
not admit an Alon-Edelman type result for any shortening function. To this end we consider the
games $\game_{n,m}^{k,1}$ and state the following inequalities :

\begin{lemma}
  \label{lemma_Johnston_bounds} 
  For integers $n=2\tilde{n}+1\ge 3$ odd, $l=1$, $m=\tilde{n}+1$, and $k\ge 2$ we have
  \begin{eqnarray*}
    \Vert \operatorname{JS}(\Gone)-\operatorname{JS}(\Gtwo)\Vert_1 
    &\ge& 
    %k\cdot 4^{\tilde{n}}+\xi^{\operatorname{JS}}
    %\ge 
    \left(k\cdot\frac{\sqrt{\tilde{n}}}{\sqrt{2}}+1\right)\xi^{\operatorname{JS}},\\
    \Vert \operatorname{JS}(\Gone)-\operatorname{JS}(\Gthree)\Vert_1&\ge&
    %k\cdot 4^{\tilde{n}} -\frac{\sqrt{2}\cdot 4^{\tilde{n}}}{\sqrt{\tilde{n}+1}}+\xi^{\operatorname{JS}}
    %\ge (k-1)4^{\tilde{n}}+\xi^{\operatorname{JS}}\\
    %&\ge&
    \left((k-1)\cdot\frac{\sqrt{\tilde{n}}}{\sqrt{2}}+1\right)\xi^{\operatorname{JS}},\text{ and}\\
    \lim\limits_{\tilde{n}\to\infty} \frac{\xi^{\operatorname{JS}}}{\sum\limits_{i=1}^{n+k} \operatorname{JS}_i(\Gone)}&=& 0.    
   \end{eqnarray*}
\end{lemma}  

Having Lemma~\ref{lemma_Johnston_bounds} at hand, whose proof is delayed to Subsection~\ref{subsec_technical_johnston},
we can conclude that for the absolute Johnston index rounding up or down does not yield the desired approximation result
for all $n\in\mathbb{N}$. 

 \begin{theorem}
   \label{thm_no_alon_edelman_type_result_for_johnston}
   For $\mathcal{T}_n\subseteq\mathcal{V}_n$ being compatible with an arbitrary shortening function $\Gamma$ on
   $\mathcal{V}_n$, no Alon-Edelman type result for $\operatorname{JS}$ and $\Gamma$ can exist.
 \end{theorem}

Theorem \ref{thm_no_alon_edelman_type_result_for_johnston} leaves little room for an approximation result similar to
Theorem~\ref{main_thm}, since our only restriction is that of using a shortening function. So in principle it may exist
an approximation result but only with a more dramatic change of the game. For the normalized Johnston index
$\widetilde{\operatorname{JS}}$ the following bounds let us draw a similar conclusion: 

\begin{lemma}
  \label{lemma_normalized Johnston_bounds} 
  For integers $n=2\tilde{n}+1$ odd, $l=1$, $m=\tilde{n}+1$, and $k\ge 2$, where $\tilde{n}$ is sufficiently
  large, we have
  \begin{eqnarray*}
    \Vert \widehat{\operatorname{JS}}(\Gone)-\widehat{\operatorname{JS}}(\Gtwo)\Vert_1
    %=\sum_{i=1}^{n+k} \left|\frac{\operatorname{JS}_i(\Gone)}{\sum_{j=1}^{n+k} \operatorname{JS}_j(\Gone)}
    %-\frac{\operatorname{JS}_i(\Gtwo)}{\sum_{j=1}^{n+k} \operatorname{JS}_j(\Gtwo)}\right| 
    &\ge& \frac{1}{5k},\\
    \Vert \widehat{\operatorname{JS}}(\Gone)-\widehat{\operatorname{JS}}(\Gthree)\Vert_1
    %=\sum_{i=1}^{n+k} \left|\frac{\operatorname{JS}_i(\Gone)}{\sum_{j=1}^{n+k} \operatorname{JS}_j(\Gone)}
    %-\frac{\operatorname{JS}_i(\Gthree)}{\sum_{j=1}^{n+k} \operatorname{JS}_j(\Gthree)}\right|
    &\ge& \frac{1}{5k},\text{ and}\\
    \lim\limits_{\tilde{n}\to\infty} \xi^{\widehat{\operatorname{JS}}}
    %=\lim\limits_{\tilde{n}\to\infty}\frac{\xi^{\operatorname{JS}}}{\sum_{j=1}^{n+k} \operatorname{JS}_j(\Gone)}
    &=&0.
  \end{eqnarray*}
\end{lemma}

\section{ILP formulations for the exact solution of the inverse power index problem}
\label{sec_ilp_formulations}
In this section we will develop a generic integer linear programming formulation for the exact
solution of the inverse power index problem for power indices based on the idea of counting
functions. They mimic and generalize the ideas from \cite{kurz2012inverse}, where the Shapley-Shubik
and the Banzhaf index have been treated.

\subsection{The classes of underlying Boolean games}
\label{subsec_ILP_underlying_game}
In order to model a Boolean game $\game:2^N\rightarrow \{0,1\}$ we introduce binary variables $x_S\in\{0,1\}$
for all coalitions $S\subseteq N$ with the meaning $x_S=v(S)$, i.e.\ $x_S=1$ for winning and $x-S=0$ for
losing coalitions. The two conditions, $v(\emptyset)=0$ and $v(N)$, of Definition~\ref{def_boolean_game} can then
be written as $x_{\emptyset}=0$ and $x_{N}=0$. Next we state the conditions to model the refinements of Boolean
games described in Section~\ref{sec_games}.

The incidence vectors $x_S$ of the winning coalitions of a Boolean game correspond to a simple game if 
the constraints
\begin{equation}
  \label{ie_simple_game}
  x_{S\backslash \{i\}} \le x_S\quad\forall \emptyset\neq S\subseteq N,\, i\in S
\end{equation}
are satisfied. For complete simple games we assume $1\succeq 2\succeq\dots\succeq n$, i.e.\ the {\voter}s
are numbered from the most powerful to the least powerful. Here a set of sufficient conditions is given 
by
\begin{eqnarray}
  \label{ie_csg_1}
  x_{S\cup \{k(S)+1\}\backslash \{k(S)\}}&\le& x_S\quad \forall \emptyset\neq S\subseteq N\backslash n\text{ and}\\
  \label{ie_csg_2}
  x_{S \backslash \{n\}}&\le& x_S\quad \forall \{n\}\subseteq S\subseteq N,
\end{eqnarray}
where $k(S)=\max\{i\mid i\in S\}$ denotes the maximum index of a {\voter} in coalition~$S$. We remark that the inequalities
(\ref{ie_csg_1}) and (\ref{ie_csg_2}) dominate the inequalities~(\ref{ie_simple_game}). 

For weighted games we use the inequalities~(\ref{ie_csg_1}) and (\ref{ie_csg_2}) for complete simple games
and additionally introduce the variables $q,w_1,\dots,w_n\in\mathbb{R}_{\ge 0}$ for the quota and the weights 
of the {\voter}s. It is a well known fact that each weighted game admits a representation where all weights and the
quota are integers. Being slightly less restrictive we assume a weighted representation where the weight of every
winning coalition is at least one more than the weight of an arbitrary losing coalition, i.e.\ we assume $w(S)\ge q$
for all winning coalitions $S\subseteq N$ and $w(T)\le q-1$ for all losing coalitions $T\subseteq N$. To this end we
can require that the quota is at least one:
\begin{equation}
  \label{ie_wvg_1}
  q\ge 1.
\end{equation} 
The assumed ordering of the {\voter}s of the underlying complete simple game induces an ordering on the weights:
\begin{equation}
  \label{ie_wvg_2}
  w_i\le w_{i+1}\quad\forall 1\le i\le n-1.
\end{equation}
It remains to interlink the weights $w_i$ with the incidences $x_S$:
\begin{eqnarray}
  \label{ie_wvg_3}
  q-(1-x_S)\cdot M-\sum_{i\in S} w_i &\le& 0\quad\forall S\subseteq N\text{ and}\\
  \label{ie_wvg_4}
  q+x_S\cdot M-\sum_{i\in S} w_i &\ge& 1\quad\forall S\subseteq N,
\end{eqnarray}
where $M$ is an suitably large constant fulfilling $M\ge\sum\limits_{i=1}^n w_i$. (According to 
\cite[Theorem 9.3.2.1]{0243.94014} we may choose $M=4n\left(\frac{n+1}{4}\right)^{(n+1)/2}$.) For a
winning coalition $S$, i.e.\ $x_S=1$, inequality~(\ref{ie_wvg_3}) is equivalent to $\sum_{i\in S}w_i\ge q$.
Similarly, for a losing coalition $T$, i.e.\ $x_T=0$, inequality~(\ref{ie_wvg_4}) is equivalent to
$\sum_{i\in S}w_i\le q-1$. The two other combinations are equivalent to automatically true inequalities
if $M$ is chosen sufficiently large. More concretely, for a winning coalition $S$ inequality~(\ref{ie_wvg_4})
is equivalent to $\sum_{i\in S}w_i\le q-1+M$, which in turn is true due to $M\le q-1+M$ and $M\ge\sum\limits_{i=1}^n w_i$.
For each losing coalition $T$ inequality~(\ref{ie_wvg_3}) is equivalent to $\sum_{i\in S}w_i\ge q-M$, which is
true if we additionally require $q\le M$.

Sufficient conditions for proper games can be easily stated as
\begin{equation}
  \label{ie_proper}
  x_S+x_{N-S}\le 1\quad \forall S\subseteq N, |S|\le n/2.
\end{equation}   
Similarly we have
\begin{equation}
  \label{ie_strong}
  x_S+x_{N-S}\ge 1\quad \forall S\subseteq N, |S|\le n/2
\end{equation} 
for strong games. For constant-sum games we can take inequalities~(\ref{ie_proper}) and (\ref{ie_strong}).   
 
\subsection{Power indices based on a counting function}
\label{subsec_ILP_power_index}

Let $\mathcal{V}_n$ be a class of binary voting games on $n$ {\voter}s. In the previous subsection we have provided
ILP formulations for all types of binary voting games described in Section~\ref{sec_games} and remark that other
classes require, expectably relatively easy and standard, adaptations. For the power index we here assume that it
is based on a counting function $C:\mathcal{V}_n\times 2^N\times N\rightarrow\mathbb{R}_{\ge 0}^n$. So we introduce
the nonnegative real variables $y_{i,S}\in\mathbb{R}_{\ge 0}$, with the interpretation $y_{i,S}=C_i(\game,S)$, for
all $i\in N$, $S\subseteq N$ and a game $\game\in\mathcal{V}_n$ represented by the $x_S$ and the constraints from
Subsection~\ref{subsec_ILP_underlying_game}. We model the induced power index 
$P:\mathcal{V}_n\rightarrow\mathbb{R}_{\ge 0}^n$ by the variables $p_i\in\mathbb{R}_{\ge 0}$ and the constraints
\begin{equation}
  \label{ie_power_index}
  p_i= \sum_{S\subseteq N} y_{i,S}
\end{equation} 
for all $i\in N$. The constraints for the values $y_{i,S}=C_i(\game,S)$ depend on the precise definition of the 
respective counting function. For the power indices introduced in Section~\ref{sec_power_indices}, which are based on 
counting functions, they can easily be stated, since the underlying concepts of e.g.\ winning, losing, minimal winning,
or shift-minimal winning coalitions can be directly expressed using the binary variables $x_S$.

For the Shapley-Shubik index, given by the counting function
$$
  C_i^{\operatorname{SSI}}(\game,S) =
  \left\{\begin{array}{rcl}\frac{(|S|-1)!(n-|S|)!}{n!}&:&i\in S,\game(S)=1,\game(S\backslash\{i\})=0,\\0&:&\text{otherwise},\end{array}\right.
$$
and a subclass $\mathcal{V}_n\subseteq \mathcal{S}_n$ of simple games we can state
$y_{i,S}=\frac{(|S|-1)!(n-|S|)!}{n!}\cdot\left(x_{S}-x_{S\backslash\{i\}}\right)$ for all $\{i\}\subseteq S\subseteq N$ 
and $y_{i,S}=0$ for all 
$S\subseteq N\backslash\{i\}$. If we cannot guarantee $x_{S}-x_{S\backslash\{i\}}\ge 0$, i.e.\ when the game
described by the $x_S$ is not simple, we have to model the first set of constraints differently:
\begin{eqnarray*}
  y_{i,S} &\le& \frac{(|S|-1)!(n-|S|)!}{n!}\cdot x_S,\\
  y_{i,S} &\le& \frac{(|S|-1)!(n-|S|)!}{n!}\cdot \left(1-x_{S\backslash\{i\}}\right),\\
  y_{i,S} &\ge& \frac{(|S|-1)!(n-|S|)!}{n!}\cdot \left(x_S-x_{S\backslash\{i\}}\right)  
\end{eqnarray*}  
for all $\{i\}\subseteq S\subseteq N$. If either $i\notin S$, $\game(S)=x_S=0$, or
$\game(S\backslash\{i\})=x_{S\backslash\{i\}}=1$, then we clearly have $y_{i,S}=0$ due to $y_{i,S}\ge 0$. From the
third set of the above inequalities we conclude $y_{i,S}=\frac{(|S|-1)!(n-|S|)!}{n!}$ for the case where $\game(S)=x_S=1$
and $\game(S\backslash\{i\})=x_{S\backslash\{i\}}=0$, since $y_{i,S} \le \frac{(|S|-1)!(n-|S|)!}{n!}\cdot x_S$.  

Generally, we remark that there is a rich modeling theory for ILPs and we want to briefly give a few, in our context,
relevant ideas. Let $a\in\{0,1\}$ be a variable representing a logical value. Assume that we want to use a
\textit{conditional inequality} $f(b)\le c$ just for the cases where $a=1$ and drop it for $a=0$. Using a sufficiently
large real constant $M$ we can formulate this situation as
$$
  f(b) \le c +(1-a)\cdot M.
$$ 
This technique is called \textit{Big-M method} and we have already seen its application in the case of weighted games in the
previous subsection. Since $f(b) \ge c$ is equivalent to $-f(b) \le -c$ and $f(b)= c$ is equivalent to
$f(b) \le c$ and $f(b) \ge c$, all types of conditional linear inequalities can be modeled. For all logic operations,
i.e.\ arbitrary combinations of \texttt{and}, \texttt{or}, and negation \texttt{not} of a finite number of Boolean
logical values, there exists a set of inequalities to model the corresponding logic gate.

In our example for the $\operatorname{SSI}$, things become more smoothly if we additionally introduce the binary
variables $y_{i,S}'\in\{0,1\}$ and set $y_{i,S}=\frac{(|S|-1)!(n-|S|)!}{n!}\cdot y_{i,S}'$.

We provide a list of ILP models for the power indices of Section~\ref{sec_power_indices}, which are based on 
counting functions, in appendix \ref{sec_ILP_counting_functions}. We have to remark that such ILP models can 
normally easily be obtained for the absolute versions of the power indices. For the respective normalized versions
we give a general reduction to a sequence of ILPs in the next subsection. In some cases it is possible to reformulate,
to be more precisely to exactly linearize, a linear fractional term to several linear constraints, see e.g.\ 
\cite{liberti2007techniques}.

\subsection{Measurement of the deviation}
\label{subsec_ILP_deviation}

Suppose we are given a target vector $\sigma=(\sigma_1,\dots,\sigma_n)$ of the desired power distribution.
Given a certain power index $P$, a class of binary voting games $\mathcal{V}_n$, and a norm $\Vert\cdot\Vert$,
the inverse power index problem asks for a game $\game\in\mathcal{V}_n$ such that the
deviation $\Vert P(\game)-\sigma\Vert$ from the desired power distribution is minimized. The details for the 
class of binary games and power indices have been given in the previous subsections. Now we go into the details
for the selectable norms. When staying in the class of ILPs we have to, more or less, restrict ourselves onto the
sum of absolute values $\Vert\cdot \Vert_1$ and the maximum norm $\Vert\cdot \Vert_\infty$.

Using more general optimization problems as a framework
we can also write down other norms easily or eventually solve the resulting problems with the corresponding 
optimization algorithms and software packages. In this context ILPs are far more innocent than their
generalization to MINLPs (Mixed Integer Nonlinear Programming).

In both treated cases we need the partial expression $\left|p_i-\sigma_i\right|$ for all $i\in N$. To this end
we introduce the nonnegative real variables $\delta_i$ and the inequalities
\begin{eqnarray}
  p_i-\sigma_i &\le& \delta_i \text{ and}\\
  \sigma_i-p_i &\le& \delta_i   
\end{eqnarray}
for all $i\in N$. With this, all feasible solutions satisfy $\delta_i\ge\left|p_i-\sigma_i\right|$. For the
$\Vert\cdot\Vert_1$ norm we can minimize the target function
\begin{equation}
  \sum_{i=1}^n \delta_i,
\end{equation}
so that we have $\delta_i=\left|p_i-\sigma_i\right|$ for each minimal solution. For the $\Vert\cdot\Vert_\infty$
norm we can identify the $\delta_i$, i.e.\ set $\delta=\delta_1=\dots=\delta_n$ and minimize $\delta$, so that
we have $\delta=\left|p_i-\sigma_i\right|$ for each minimal solution and all $i\in N$.

Given that the three parts of an inverse power index problem can be expressed linearly with integer (and real)
variables, we can make use of one of the many available software packages for ILPs. For our practical 
computations we have used the software CPLEX version 12.4 from IBM ILOG.  

Once we can formulate a certain inverse power index problem with power index $P$, class $\mathcal{V}_n$, and 
norm $\Vert\cdot\Vert$ as an ILP, we can solve the corresponding problem with the normalized power index
$\widehat{P}$ instead of $P$ by a sequence of ILPs. Actually we have to express 
$$
  \left|\frac{p_i}{\sum_{j=1}^n p_j}-\sigma_i\right| =\delta_i
$$
linearly (for all $i\in N$). Since $p_i\ge 0$, we can multiply both sides with $\sum_{j=1}^n p_j$ and obtain
$$
  \left|p_i-\sigma_i\cdot \sum_{j=1}^n p_j\right| =\delta_i\cdot \sum_{j=1}^n p_j.
$$ 
The absolute value on the left hand side can be easily linearized as show before. Unfortunately the right hand side
contains quadratic terms, i.e.\ the summands $\delta_i\cdot p_j$ depend on two variables. Since none of the two 
variable types is binary, there is no standard linearization available. As a workaround we introduce
$\delta_i'=\delta_i\cdot \sum_{j=1}^n p_j$ as new nonnegative real variables for all $i\in N$. With this, the
corresponding constraints are given by
\begin{eqnarray}
  p_i-\sigma_i\cdot \sum_{j=1}^n p_j &\le& \delta_i'\text{ and}\\
  \sigma_i\cdot \sum_{j=1}^n p_j-p_i &\le& \delta_i'  
\end{eqnarray}
for all $i\in N$. Minimizing $\sum_{i=1}^n \delta_i'$ (for $\Vert\cdot\Vert_1$) or $\delta'=\delta_1'=\dots=\delta_n'$
(for $\Vert\cdot\Vert_\infty$) would not yield the desired result. Instead we introduce a numerical parameter $\alpha$
and the constraint
\begin{equation}
  t(\delta_1',\dots,\delta_n') \le \alpha \cdot\sum_{j=1}^n p_j, 
\end{equation}
where $t(\delta_1,\dots,\delta_n)$ is the original target function. We completely drop the old target function and
obtain a so called feasibility problem, which can be solved with the same methods as ILPs including a target function.
The interpretation of $\alpha$ is as follows: If, for a given $\alpha$, the corresponding ILP contains a feasible solution,
then there exists a solution of $\Vert P(\game)-\sigma\Vert\le \alpha$ and each game $\game$ being described by 
such a feasible solution satisfies this inequality. If otherwise the set of feasible solutions is empty, we have
$\Vert P(\game)-\sigma\Vert> \alpha$ for all $\game\in\mathcal{V}_n$.

Initially we generally know that the minimum value of $\Vert P(\game)-\sigma\Vert$ is contained in the interval
$[0,\infty)$. Using a bisection algorithm for this interval and $\alpha$, we obtain a sequence 
$[l_1,r_1]\supseteq\dots\supseteq [l_m,r_m]$ of intervals of decreasing length such that the minimum value of
$\Vert P(\game)-\sigma\Vert$ is contained in $[l_j,r_j]$ for all $1\le j\le m$. Thus we can determine 
the minimum possible deviation (and a corresponding game) up to each given precision. If $\mathcal{V}_n$ is finite
there are only finitely many attainable power distributions $P(\game)$, so that this convergence suffices
to obtain the exact solution after a finite number of iterations, i.e.\ there exists a problem-dependent
upper bound on the length of $[l_m,r_m]$ which suffices to conclude that the best found solution is already
optimal. For the Banzhaf index the technical details have been described in \cite{kurz2012inverse} and, in
more detail, in \cite{kurz2012heuristic}.  

\section{Power distributions which are not concentrated on the first $k$~{\voter}s}
\label{sec_not_concentrated}

Let $\sigma=(\sigma_1,\dots,\sigma_n)\in\mathbb{R}_{\ge 0}^n$ be a desired power vector with $\Vert\sigma\Vert_1=1$.
Given a power index $P$, a norm $\Vert\cdot\Vert$, and a class $\mathcal{V}_n$ of binary games on $n$~{\voter}s, the
inverse power index asks for a game $\game\in\mathcal{V}_n$ minimizing the deviation $\Vert P(\game)-\sigma\Vert$. In
the previous section we have described an exact integer linear programming approach for those power indices from 
Section~\ref{sec_power_indices}, which are based on counting functions. For a concrete instance of the inverse 
power index problem we may eventually find the exact solution algorithmically in a reasonable amount of time. 

Are more general statements achievable? If the number $n$ of {\voter}s is small, there is only a relatively small finite
number of games in $\mathcal{V}_n$, so that $\min_{\game\in\mathcal{V}_n}\Vert P(\game)-\sigma\Vert$ may not become
too small. Whether $\min_{\game\in\mathcal{V}_n}\Vert P(\game)-\sigma\Vert$ is rather small or rather large can be
decided by just looking at the class of $k$-pure games in $\mathcal{V}_n$ with the aid of the Alon-Edelman type bounds from
Section~\ref{sec_alon_edelmann_type}, provided that most of the power of $\sigma$ is concentrated on the first $k$
coordinates, i.e.\ $\sum_{i=k+1} \sigma_i\ll 1$. Despite this rather vague description, things can be made very precise
by stating bounds in terms on a parameter $\varepsilon$, with $\sum_{i=k+1} \sigma_i\le\varepsilon$. Of course we cannot
say much more, since there are power distributions like e.g.\ $\sigma_n=(0.75,0.25,0\dots,0)$ which are rather hard to
approximate for most power indices and others like $\sigma'=(0.5,0.5,0\dots,0)$, which can be perfectly met. 

Is the assumption
that most of the power is concentrated on the first $k$ components a realistic assumption that is commonly satisfied in
practice? Instead of an answer we aim to classify all vectors of the unit simplex into different types of desired power
distributions $\sigma$. To this end let us call coordinates~$i$, where $\sigma_i$ is relatively small, \emph{oceanic}.
Those coordinates, where $\sigma_i$ is relatively large, are called \emph{islands}. If we have a sequence $\sigma^{(n)}$
of desired power distributions with increasing number of {\voter}s for all $n\in \mathbb{N}$, we can made this more precise
immediately: We call coordinate $i$ oceanic iff $\lim_{n\to\infty} \sigma_i^{(n)}=0$ and islands iff there exists a lower bound
$u_i>0$ such that $\sigma_i^{(n)}\ge u_i$ for all sufficiently large $n$. In principle there may be coordinates
which are neither oceanic nor islands, e.g.\ we may have $\sigma_i^{(n)}=1/7$ for odd $n$ and $\sigma_i^{(n)}=\frac{1}{n^2}$
for even $n$. Those cases may be considered as not \textit{well behaved} and we ignore them, if not all coordinates
are either oceanic or islands.

Now let $k$ denote the number of islands of such a sequence. If the aggregated desired power of the oceanic {\voter}s
is rather small or even tends to zero, we can apply the Alon-Edelman type results. The two other cases are that we have
no, i.e.\ $k=0$, islands or the aggregated desired power of the oceanic {\voter}s is non-vanishing. For simplicity, we assume
in the latter case that the desired power of the ocean is roughly given by a real constant $\alpha$. In both cases 
the so-called limit results for power indices give precise theoretical predictions about the power distribution when 
considering weighted games with weights $w_i$ equal to $\sigma_i$.

Let us first consider the case where no island is present, i.e.\ $k=0$ or in other words, where
$\max_{i} \sigma_{i}^{(n)}$ tends to zero as $n$ increases. The strongest statements are the
so-called Penrose's limit (type) theorems, see \cite{Penrose}: Under certain technical conditions we have
\begin{equation}
  \label{eq_relative_convergence}
  \lim_{n\to\infty} \frac{P_i\left(\left[q;\sigma_1^{(n)},\dots,\sigma_i^{(n)},\dots\right]\right)}{w_i}=1.
\end{equation}
Typical assumptions are that $\Vert\sigma^{(n)}\Vert_{\infty}$ tends to zero, and that the entries $\sigma_i^{(n)}$
only take a finite number of different values for each $n$ and that the number of occurrences of each type increases
with $n$ -- so-called replicative chains of weighted games. For the Banzhaf index such a result was proven for 
quota $q=\frac{1}{2}$ and for the Shapley-Shubik index for all fixed quotas $q\in(0,1)$, see \cite{lindner2004ls}. But
there are also cases where such a strong result does not hold, see e.g.\ \cite{chang2006ls,lindner2007cases}. The
technical assumption of replicative chains can be dropped if relative convergence of the fractions in 
Equation~(\ref{eq_relative_convergence}) is replaced by absolute convergence in e.g.\ the $\Vert\cdot\Vert_1$-norm.
\cite[Theorem 9.8]{neyman1982renewal} states, in different notation
\begin{equation}
  \label{eq_absolute_convergence}
  \lim_{n\to\infty} \left\Vert \operatorname{SSI}\!\left(\game^{(n)}\right)-\sigma^{(n)}\right\Vert_1=0,
\end{equation}  
where the weights of the weighted games $\game^{(n)}$ can be taken as $\sigma^{(n)}$ and quota $q$ can be taken
almost arbitrarily, i.e.\ there should be no cluster point at either $0$ or $1$. We remark that generally 
an absolute limit result like Equation~(\ref{eq_absolute_convergence}) implies a relative limit result like  
Equation~(\ref{eq_relative_convergence}) provided that replicative chains are considered.

A big drawback of those limit results might be that they cannot be directly applied to a given single desired 
power distribution $\sigma$. Whenever there are concrete error bounds we can say something for a given single 
desired power distribution $\sigma$. Recently such a bound has been given in \cite{kurz2013nucleolus} for
the nucleolus:
\begin{equation}
  \label{eq_nucleolus_convergence}
  \left\Vert \operatorname{Nuc}\!\left([q;w]\right)-w\right\Vert_1\le \frac{2\cdot\max_{1\le i\le n} w_i}{\min(q,1-q)},
\end{equation}
for all $q\in(0,1)$ and $w\in\mathbb{R}_{\ge 0}^n$ with $\Vert w\Vert_1=1$. For the $\operatorname{SSI}$ such
bounds should be hidden in the proofs of the technical lemmas of \cite{neyman1982renewal}.

In general we cannot expect much more than an inequality similar to (\ref{eq_nucleolus_convergence}): Let $P$
be a symmetric, efficient and positive power index satisfying the null {\voter} property. For 
$\sigma^{(n)}=\frac{1}{2n-1}\cdot(2,\dots,2,1)$ we have
\begin{equation}
  \label{ie_analytical_example}
  \left\Vert P\!\left(\left[q;\sigma^{(n)}\right]\right)-\sigma^{(n)}\right\Vert_1 \ge \frac{2}{2n-1}\cdot \frac{n-1}{n} 
\end{equation}
and $\max_i \sigma_i^{(n)}=\frac{2}{2n-1}$ for all $n$, see \cite{kurz2012heuristic}. So for this special
sequence of desired power distributions, taking the desired power as weights yields a 
$\theta\!\left(\frac{1}{n}\right)$-error, which is still considerably large for medium sized constitutions 
like e.g.\ $n=27$ or $28$. For $P=\operatorname{Bz}$ we can remark that solving the inverse power index yields
slightly better bounds than Inequality~(\ref{ie_analytical_example}) within the class $\mathcal{T}_n$ and 
an exact solution within the class of simple games for all $6\le n\le 18$, see \cite[Table 9]{kurz2012heuristic}.

For the remaining case of a small but positive number of islands and a non-vanishing ocean there are limit results too.
In \cite[Theorem 1]{shapiro1978values} an approximation formula with an estimate for the rate of convergence was given.
Here the assumptions on the weight distribution within the ocean are rather mild. Interestingly enough, provided
a small variance of the weight distribution within the ocean the error term is of the form $\theta\!\left(\frac{1}{n}\right)$.
If all oceanic {\voter}s have the same weight, then there is a similar result for the Banzhaf index, which fails to be
true for asymmetric weight distributions within the ocean, see \cite{dubey1979mathematical}. Some first results
in that direction for the nucleolus can be found in \cite{galil1974nucleolus}.

\section{Conclusion and future work}
\label{sec_conclusion}

Inspired by the seminal work of \cite{pre05681536}, we have considered Alon-Edelman type results for most of the
known power indices. It turned out that for the power indices $\operatorname{KB}$, $\operatorname{PHI}$, and 
$\operatorname{Chow}$ it makes no sense to ask for Alon-Edelman type results, since the respective power indices
do not admit power distributions where most of the power is concentrated on a small number of {\voter}s. For
the other considered power indices such a concentration is in principle possible. Nevertheless, there can be no
such result for the Johnston index, which we have shown by analytical power index calculations for a certain
class of parameterized weighted voting games, which might be interesting in its own right. For $p$-binomial
semivalues and $k$-rounding we similarly have obtained a negative result. On the other hand we conjecture that 
there exists an Alon-Edelman type result for $p$-binomial semivalues with $(k,p)$-rounding. The cases
of the Shapley-Shubik index and more general semivalues are left open. In order to classify and represent 
power indices in a unified way, we have introduced the concept of power indices based on counting functions. 
Except the nucleolus and the MSR index, due to somewhat \textit{global} properties, all of the presented power
indices of this paper, admit such a representation. For those power indices based on counting functions, we have
provided a theoretical and notational framework to formulate Alon-Edelman type results, if they exist at all,
in a unified and more or less compact way.

We are pretty sure that some of our bounds for the quality functions in Table~\ref{table_quality_functions} can 
be improved. It would also be interesting to construct \textit{worst case examples} showing how far our estimates
are from the real truth. In order to study the existence question of Alon-Edelman type results for the nucleolus
or the MSR index possibly other techniques are  necessary.

Maybe the concept of counting functions can be fruitfully used in different contexts in order to unify approaches
for several power indices. Examples might be generating function algorithms to compute values of power indices
for weighted voting games or the design of new, so far missed, power indices, see \cite{alonso2010new}.

Since $k$-rounding is the right shortening function for most power indices which admit an Alon-Edelman type result,
it would be interesting to study further of its theoretic properties. We have shown that $k$-rounding preserves 
weightedness. What about the generalizations? Does $k$-rounding preserve the dimension or stays within the class of
roughly weighted, $\alpha$-roughly weighted\footnote{See \cite{freixas2011alpha} and \cite{gvozdeva2013three},
where three hierarchies of simple games have been introduced.}, or homogeneous games?       

In the context of Alon-Edelman type results, of course one may find it interesting to enlarge the class of binary
voting games to games with several levels of approval in the input and output. If the number of option tends to 
infinity we end up with continuous models, see e.g.\ \cite{kurz2013measuring}. Also the $\Vert\cdot\Vert_1$ in the 
present Alon-Edelman type results may be replaced by different norms.

Alon-Edelman type results allow negative approximation results by reducing the $n$~{\voter} case to the
$k$~{\voter} case, provided that the desired power is concentrated only on a few {\voter}s. As a counterpart
we have outlined the known theory of limit results for power indices. For our purposes it would be very valuable
if some of these results could be turned into precise error estimates and possibly generalized to larger
classes of power indices. 

From the practical point of view we have presented a general exact algorithmic approach for the inverse power index
problem for power indices based on counting functions using integer linear programming. Any application of
more sophisticated techniques from integer linear programming would be beneficial to shift the computational
limits of this approach, which generally is NP complete so that exact solutions can be expected only for rather
small numbers of {\voter}s.

\cite{pre05681536} close by mentioning that they began the study which vectors in the unit simplex can be closely
approximated by Banzhaf vectors of simple games. We generalize their question to the list of known power indices and 
important subclasses of simple games. We agree that Alon-Edelman type results seem to be unable to provide a complete
solution.  If power has to be attained for regions, i.e.\ collections of {\voter}s, instead of single
{\voter}s, one can well approximate any distribution, as shown for the special case of the Banzhaf index in
\cite[Proposition 3.1]{pre05681536}. A similar statement is obviously true for all symmetric, positive and efficient
power indices, which satisfy the null {\voter} property.

\appendix

\section{Details for quality function results}
\label{sec_details_quality_functions}

In this section we want to prove the missing details for the quality function results announced in
Table~\ref{table_quality_functions}. Since we have started with the proof for the Public Good Index, in 
Subsection~\ref{subsec_quality_functions}, we want to continue with the Deegan-Packel index, which arises 
as the equal division of the $\operatorname{PGI}$, see Definition~\ref{def_equal_division}.

We can exploit this relation by setting $\widehat{C}^{\operatorname{DP}}=\sum_{i=k+1}^n C_i^{\operatorname{DP}}$
and denoting the number of minimal winning coalitions that contain at least one member of $(k,n]$ by 
$\widehat{M}^{\operatorname{PGI}}$. With this we have: 

\begin{lemma}
  \label{lemma_m_c_bound}
  $$
    \widehat{M}^{\operatorname{PGI}}\le \widehat{C}^{\operatorname{DP}}\cdot (k+1)
  $$
\end{lemma}
\begin{proof}
  Let $S$ be a minimal winning coalition that contains at least one player from $(k,n]$. We set $a:=\left|S\cap [1,k]\right|$
  and $b:=\left|S\cap (k,n]\right|$, so that $a+b=|S|$, $0\le a\le k$, and $1\le b\le n-k$. Thus, the stated inequality follows
  from $\sum\limits_{i=k+1}^n C_i^{\operatorname{DP}}(\game,S)=\frac{b}{a+b}\ge \frac{1}{k+1}$, where $\game$ denotes the
  respective game.
\end{proof}

\begin{corollary}
  For the Deegan-Packel index we can choose $f_1(k)=(k+1)^2$ and $f_2(k)=\frac{(k+1)^2(k+5)}{4}$.
\end{corollary}

Doing a tailored analysis similar to the one in the proof of Lemma~\ref{lemma_quality_functions_PGI}, allows us to provide
tighter bounds:

\begin{lemma}
  \label{lemma_quality_functions_DP}
  For the Deegan-Packel index we can choose $f_1(k)=k+2$ and $f_2(k)=\frac{k^2+6k+9}{4}$.
\end{lemma}
\begin{proof}
  Let $\game=(\mathcal{W},[n])$ be an arbitrary simple game. As defined before, $C_i^{\operatorname{DPI}}(\game,2^{[n]})$
  is the counting function of the Deegan-Packel index. For brevity, we just write $C_i$ and $\overline{C}$ for 
  $C_i^{\operatorname{DPI}}$ and $\overline{C}^{\operatorname{DPI}}$. As further abbreviations we use
  $\game'=(\mathcal{W}',[n])$ for the $k$-rounding of $\game$ and $\widehat{C}=\sum_{i=k+1}^n C_i$, i.e.\ the restriction of
  $\overline{C}$ to the {\voter}s in $(k,n]$. 
    
  Now we want to study the changes in the $C_i$ by going from $\game$ to $\game'$.  At first we consider the
  cases where coalition $S$ is a MWC in $\game$ but not in $\game'$.
  \begin{enumerate}
    \item[(1)]  If $S\cap (k,n]\neq \emptyset$ then removing $S$ from $\mathcal{W}$ results in a decrease of at least
                $\frac{1}{|S|}$ for {\voter}~$i$ and there is at least one {\voter} in $(k,n]$. Thus the negative change of the $C_i$ of
                that type is bounded by $\widehat{C}$.
    \item[(2)] If $S\subseteq [1,k]$, then $S$ has to be winning in $\game'$ too. Since $S$ is not a MWC in $\game'$
               (by assumption), there must be an index $j\in S$ such that $S-j$ is winning in $\game'$. Thus there must be
               a coalition $\emptyset\neq T\in 2^{(k,n]}$ such that $S-j\cup T$ is a MWC in $\game$. So we can 
               bound those cases by counting them at MWC $S-j\cup T$. The contribution of $S-j+T$ to $\widehat{C}$ is
               $\frac{t}{t+r}\ge \frac{1}{r+1}$, where $t=|T|$ and $r=|S-j|$. Given coalition $S-j+T$, there are $k-r\le k-1$
               possible choices for coalition $S$. In each choice the decrease of $C_i$ on $S$ is given by $\frac{1}{r+1}$ so
               that the total decrease of $C_i$ in those cases is bound from above by $(k-1)\widehat{C}$.
               For an upper bound of the decrease of all $C_i$ we have $k-r$ possibilities involving $r+1$ {\voter}s each. 
               Since $(k-r)(r+1)\le\left\lfloor\frac{(k+1)^2}{4}\right\rfloor$, we have the upper bound  
               $\left\lfloor\frac{(k+1)^2}{4}\right\rfloor \widehat{C}$.                
  \end{enumerate}  
  Next, we consider the cases where coalition $S$ is not a MWC in $\game$ but in $\game'$. Since the {\voter}s in
  $(k,n]$ are null {\voter}s we can deduce $S\in 2^{[1,k]}$.
  \begin{enumerate}
    \item[(3)] Assume that $S$ is losing in $\game$ but winning in $\game'$. According to the rounding procedure
               there exists a coalition $\emptyset\neq T\in 2^{(k,n]}$ such that $S\cup T$ is a MWC in $\game$.
               The contribution of $S\cup T$ to $\widehat{C}$ is $\frac{t}{t+r}\ge \frac{1}{r+1}$, where $t=|T|$ and $r=|S|$.
               Since $C_i(S)=\frac{1}{r}$ in $\game$ and $r\ge 1$ we have that the total increase for $C_i$ is bounded
               by $2\widehat{C}$. For the sum over the $C_i$ we note that the increase is at most $r\cdot \frac{1}{r}=1$ 
               for a single coalition compared to $\frac{1}{r+1}\ge \frac{1}{k+1}$ so that the overall change
               of those cases is at most $(k+1)\widehat{C}$.
    \item[(4)] If $S$ is winning in $\game$, then there must be a {\voter}~$j\in S$ such that $S-j$ is winning in 
               $\game$ but losing in $\game'$. We proceed similarly as in case (2) and deduce an upper bound of
               $k\widehat{C}$ for the change of $C_i$ and an upper bound of $\left\lfloor\frac{(k+1)^2}{4}\right\rfloor\widehat{C}$
               for the sum of changes of the $C_i$.
  \end{enumerate}
  Summarizing the four cases gives the mentioned two functions.
\end{proof}

\begin{lemma}
  \label{lemma_number_direct_left_shifts}
  Let $\game$ be a complete simple game with {\voter} set $N$ and $S\subseteq N$ a coalition. The number of direct 
  left-shifts of $S$ is bounded by $|S|+1$ and $\left\lfloor\frac{n+1}{2}\right\rfloor$. 
\end{lemma}
\begin{proof}
  W.l.o.g.\ we assume $1\succeq 2\succeq\dots\succeq n$. A direct left-shift of $S$ arises either by shifting a {\voter} 
  of $S$ one place to the left or by adding the weakest player $n$. For the other bound we observe that {\voter}~$i$ can
  only be shifted to position $i-1$ if $i-1\notin S$ and $i\neq 1$. Thus only $\left\lfloor\frac{n}{2}\right\rfloor$ {\voter}s
  can be shifted one position to the left, where equality is only possible if either $n$ is odd or $n$ is even and 
  $n\in S$. Considering the possible addition of {\voter}~$n$ given at most $\left\lfloor\frac{n+1}{2}\right\rfloor$ 
  cases.  
\end{proof}

\begin{lemma}
  \label{lemma_quality_functions_Shift}
  For the Shift index we can choose $f_1(k)=\frac{k+3}{2}$ and $f_2(k)=\frac{k^2+3k+2}{2}$.
\end{lemma}
\begin{proof}
  Let $\game=(\mathcal{W},[n])$ be an arbitrary simple game. As defined before, $C_i^{\operatorname{Shift}}(\game,2^{[n]})$ 
  counts the number of shift-minimal winning coalitions in $\game$ containing {\voter}~$i$ and we have 
  $\overline{C}^{\operatorname{Shift}}(\game,2^{[n]})=\sum_{i=1}^n C_i^{\operatorname{Shift}}(\game,2^{[n]})$. For brevity we just 
  write $C_i$ and $\overline{C}$. As further abbreviations we use $\game'=(\mathcal{W}',[n])$ for the $k$-rounding of $\game$,  
  $\widehat{C}=\sum_{i=k+1}^n C_i$, i.e.\ the restriction of $\overline{C}$ to the {\voter}s in $(k,n]$, and by 
  $\widehat{M}$ we denote the number of shift-minimal winning coalitions in $\game$ that contain at least one {\voter}
  from $(k,n]$. With this we have $\widehat{M}\le\widehat{C}$.

  Now we want to study the changes in the $C_i$ by going from $\game$ to $\game'$.  At first we consider the
  cases where coalition $S$ is a SMWC (shift-minimal winning coalition) in $\game$ but not in $\game'$.
 
  \begin{enumerate}
    \item[(1)]  If $S\cap (k,n]\neq \emptyset$ then removing $S$ from $\game$ results in a decrease of $1$ for
                {\voter}~$i$ and there is at least one {\voter} in $(k,n]$. Thus the negative change of the $C_i$ of
                that type is bounded by $\widehat{M}\le\widehat{C}$.
    \item[(2)]  If $S\subseteq [1,k]$, then $S$ has to be winning in $\game'$ too. Since $S$ is not a SMWC in
                $\game'$ (by assumption), there must be a direct right-shift $S'=S-j+h$ of $S$, where $j\in S$, 
                $h\in [1,k]\backslash S$ or $h=\emptyset$, such that $S'$ is winning in $\game'$. According
                to the rounding procedure, there has to be a subset $\emptyset\neq T\subseteq 2^{(k,n]}$ such
                that $S'\cup T$ is a SMWC in $\game$. Given $S'$ there are at most
                $\frac{k+1}{2}$ choices for $S$, see Lemma~\ref{lemma_number_direct_left_shifts}. Thus, we have that 
                the negative change is bounded by $\frac{k+1}{2}\cdot \widehat{M}\le \frac{k+1}{2}\cdot\widehat{C}$.            
  \end{enumerate}               
  Next, we consider the cases where coalition $S$ is not a SMWC in $\game$ but in $\game'$. Since the {\voter}s in
  $(k,n]$ are null {\voter}s, we can deduce $S\in 2^{[1,k]}$. 
  \begin{enumerate}
    \item[(3)] Assume that $S$ is losing in $\game$ but winning in $\game'$. According to the rounding procedure
               there exists  coalition $\emptyset\neq T\in 2^{(k,n]}$ such that $S\cup T$ is a SMWC in $\game$.
               Thus, the total increase for $C_i$ is bounded by $\widehat{M}\le\widehat{C}$.
    \item[(4)] If $S$ is winning in $\game$, then there must be a direct right-shift $S'$ such that $S'$ is winning in 
               $\game$ but losing in $\game'$. We proceed similarly as in case (2) and deduce an upper bound of
               $\frac{k+1}{2}\cdot\widehat{M}$ for the change of $C_i$.           
  \end{enumerate} 
  Summarizing the four cases gives the mentioned two functions.
\end{proof}
    
By proving the bound $\widehat{M}^{\operatorname{Shift}}\le \widehat{C}^{\operatorname{SDP}}\cdot (k+1)$, 
similar to the proof of Lemma~\ref{lemma_m_c_bound}, we can conclude:   
    
$f_1(k)=\frac{k+3}{2}$ and $f_2(k)=\frac{k^2+3k+2}{2}$.    
    
\begin{corollary}    
   For the Shift Deegan-Packel index we can choose $f_1(k)=\frac{(k+1)(k+3)}{2}$ and $f_2(k)=\frac{(k+1)^2(k+2)}{2}$.
\end{corollary}

But we can also perform a tailored analysis and improve the result a bit:
      
\begin{lemma}
  \label{lemma_quality_functions_SDP}
  For the Shift Deegan-Packel index we can choose $f_1(k)=2k+1$ and $f_2(k)=2k^2+k+1$.
\end{lemma}
\begin{proof}
  Let $\game=(\mathcal{W},[n])$ be an arbitrary simple game. As defined before, $C_i^{\operatorname{SDP}}(\game,2^{[n]})$
  is the counting function of the Shift Deegan-Packel index. For brevity we just write $C_i$ and $\overline{C}$ for 
  $C_i^{\operatorname{SDP}}$ and $\overline{C}^{\operatorname{SDP}}$. As further abbreviations we use
  $\game'=(\mathcal{W}',[n])$ for the $k$-rounding of $\game$ and $\widehat{C}=\sum_{i=k+1}^n C_i$, i.e.\ the restriction of
  $\overline{C}$ to the {\voter}s in $(k,n]$. 
    
  Now we want to study the changes in the $C_i$ by going from $\game$ to $\game'$.  At first we consider the
  cases where coalition $S$ is a SMWC in $\game$ but not in $\game'$.

  \begin{enumerate}
    \item[(1)]  If $S\cap (k,n]\neq \emptyset$ then removing $S$ from $\mathcal{W}$ results in a decrease of $\frac{1}{|S|}$
                for {\voter}~$i$ and there is at least one {\voter} in $(k,n]$. Thus the negative change of the $C_i$ of
                that type is bounded by $\widehat{C}$.
    \item[(2)]  If $S\subseteq [1,k]$, then $S$ has to be winning in $\game'$ too. Since $S$ is not a SMWC in
                $\game'$ (by assumption), there must be a direct right-shift $S'=S-j+h$ of $S$, where $j\in S$, 
                $h\in [1,k]\backslash S$ or $h=\emptyset$, such that $S'$ is winning in $\game'$. According
                to the rounding procedure there has to be a subset $\emptyset\neq T\subseteq 2^{(k,n]}$ such
                that $S'\cup T$ is a SMWC in $\game$. Let $u:=|S'|$ and $v:=|T|$. Since $\frac{v}{u+v}\ge \frac{1}{u+1}$
                the loss of the small ones is at least $\frac{1}{u+1}$. Since $u\le |S|\le u+1$ the decrease can be
                bounded by $2k\widehat{C}$.
  \end{enumerate}               
  Next, we consider the cases where coalition $S$ is not a SMWC in $\game$ but in $\game'$. Since the {\voter}s in
  $(k,n]$ are null {\voter}s we can deduce $S\in 2^{[1,k]}$. 
  \begin{enumerate}
    \item[(3)] Assume that $S$ is losing in $\game$ but winning in $\game'$. According to the rounding procedure,
               there exists  coalition $\emptyset\neq T\in 2^{(k,n]}$ such that $S\cup T$ is a SMWC in $\game$.
               Thus, the total increase for $C_i$ is bounded by $\widehat{C}$.
    \item[(4)] If $S$ is winning in $\game$, then there must be a direct right-shift $S'$ such that $S'$ is winning in 
               $\game{G}$ but losing in $\game'$. We proceed similarly as in case (2) and deduce an upper bound of
               $2k\widehat{C}$ for the change of $C_i$.           
  \end{enumerate} 
  Summarizing the four cases gives the mentioned two functions.
\end{proof}      
      
\begin{lemma}
  \label{lemma_quality_functions_Tijs}      
  For the Tijs index (and $k$-up-rounding) we can choose $f_1(k)=0$ and $f_2(k)=1$.
\end{lemma}
\begin{proof}
  Let $\game$ denote the original game and $\game'$ the $k$-up-rounding of $\game$ for an 
  integer $0<k<n$. By $m_1$ we denote the number of vetoers in $[1,k]$ and by $m_2$ the number
  of vetoers in $(k,n]$ of $\game$. $k$-up-rounding maps winning coalitions to winning coalitions. A
  coalition $S$ is winning in $\game'$ iff $S\cap [1,k]$ is winning in $\game'$ and $S'\subseteq[1,k]$ 
  is winning in $\game'$ iff there exists a subset $T\subseteq (k,n]$ such that $S'\cup T$ is winning in $\game$.
  Thus a {\voter} $i\in[1,k]$ is a vetoer in $\game$ if she is a vetoer in $\game'$ and no {\voter} $j\in(k,n]$
  can be a vetoer in $\game'$. For all $1\le i\le k$ we have $C_i^{\operatorname{Tijs}}(\Gamma\!\left(\game\right),2^N)=
  C_i^{\operatorname{Tijs}}\!\left(\game,2^N\right)$, so that we can choose $f_1(k)=0$, where $\Gamma$
  denotes $k$-up-rounding. Since we have $\overline{C}^{\operatorname{Tijs}}\!\left(\Gamma(\game),2^N\right)=
  \overline{C}^{\operatorname{Tijs}}\!\left(\game,2^N\right)=m_2$, we can choose $f_1(k)=1$.      
\end{proof}

\begin{lemma}
  \label{lemma_quality_functions_ColPrev}
  For $\operatorname{ColPrev}$ we can choose $f_1(k)=2$ and $f_2(k)=2k+1$. 
\end{lemma}  
\begin{proof}
  Let us write $\eta_i(\game,2^N)=C_i^{\operatorname{swing}}(\game,2^N)$ for the number of swings for {\voter}~$i$
  and $\overline{\eta}(\game,2^N)=\overline{C}^{\operatorname{swing}}(\game,2^N)$. By $\game'$ we denote
  $\Gamma(\game)$, where $\Gamma$ is the $k$-rounding function. From Lemma~\ref{lemma_quality_functions_absolute_banzhaf}
  and Lemma~\ref{lemma_scaled_counting_function} we conclude
  \begin{equation}
    \label{ie_qf_ColPrev_1}
    \left|\eta_i(\game',2^N)-\eta_i(\game,2^N)\right|\le \varepsilon\cdot\overline{\eta}(\game,2^N)   
  \end{equation}
  for all $1\le i\le k$. %%  and
  %% \begin{equation}
  %%   \left|\overline{\eta}(\game')-\overline{\eta}(\game)\right|\le (k+1)\cdot\varepsilon\cdot\overline{\eta}(\game,2^N).   
  %% \end{equation}
  By $\mathcal{W}(\game)$ and $\mathcal{W}(\game')$ we denote the set of winning coalitions of $\game$ and $\game'$,
  respectively. From the proof of Lemma~\ref{lemma_quality_functions_absolute_banzhaf} we obtain
  \begin{equation}
    \label{ie_qf_ColPrev_2}
    \Big\vert |\mathcal{W}(\game')|-|\mathcal{W}(\game)|\Big\vert\le\varepsilon\cdot \overline{\eta}(\game,2^N).
  \end{equation} 
  With this we conclude
  \begin{eqnarray*}
  \left|\frac{\eta_i(\game',2^N)}{|\mathcal{W}(\game')|}-\frac{\eta_i(\game,2^N)}{|\mathcal{W}(\game)|}\right|
  &\le& \left|\frac{\eta_i(\game',2^N)}{|\mathcal{W}(\game)|}-\frac{\eta_i(\game,2^N)}{|\mathcal{W}(\game)|}\right|
  +\left|\frac{\eta_i(\game',2^N)}{|\mathcal{W}(\game')|}-\frac{\eta_i(\game',2^N)}{|\mathcal{W}(\game)|}\right|\\
  &=& \frac{\left|\eta_i(\game',2^N)-\eta_i(\game,2^N)\right|}{|\mathcal{W}(\game)|}
  +\frac{\eta_i(\game',2^N)}{|\mathcal{W}(\game')|}\cdot
  \frac{\Big\vert|\mathcal{W}(\game')|-|\mathcal{W}(\game)|\Big\vert}{|\mathcal{W}(\game)|}\\
  &\le& \frac{\varepsilon\cdot\overline{\eta}(\game,2^N)}{|\mathcal{W}(\game)|}+
  \frac{\eta_i(\game',2^N)}{|\mathcal{W}(\game')|}\cdot
  \frac{\varepsilon\cdot\overline{\eta}(\game,2^N)}{|\mathcal{W}(\game)|}\\
  &\le& 2\cdot\frac{\varepsilon\cdot\overline{\eta}(\game,2^N)}{|\mathcal{W}(\game)|}
  \end{eqnarray*}
  where we have used the triangle inequality, Inequality~(\ref{ie_qf_ColPrev_1}), Inequality~(\ref{ie_qf_ColPrev_2}), 
  and the fact that the number of swings for {\voter}~$i$ cannot be larger then the number of winning coalitions. Thus
  we can choose $f_1(k)=2$. Due to Lemma~\ref{lemma_canonical_f_2} we can choose $f_2(k)=k\cdot f_1(k)+1=2k+1$. 
\end{proof}

\begin{lemma}
  \label{lemma_quality_functions_ColIni}
  For $\operatorname{ColIni}$ we can choose $f_1(k)=2$ and $f_2(k)=2k+1$. 
\end{lemma}  
\begin{proof}
  As pointed out in e.g.\ \cite{dubey1981value} $\operatorname{ColIni}(\game)$ equals
  $\operatorname{ColPrev}(\game^\star)$, where $\game^\star$ denotes the dual game of
  $\game$, see e.g.\ \cite{0943.91005} for a definition. Since the class of simple games is
  closed under taking the dual and $\widehat{\operatorname{ColIni}}=\widehat{\operatorname{ColPrev}}$,
  we can chose the same quality functions as in Lemma~\ref{lemma_quality_functions_ColPrev}.
\end{proof}

While we have obtained a negative Alon-Edelman-type result for $p$-binomial semivalues, there are subclasses of
semivalues where we can easily conclude a positive Alon-Edelman-type result from the proof of
Lemma~\ref{lemma_quality_functions_absolute_banzhaf}:
\begin{lemma}
  \label{lemma_quality_functions_bounded_semivalue}
  Let $\Psi^P$ be a semivalue such that there exist real constants $c_1,c_2\in\mathbb{R}_{>0}$ such that
  $c_1\le P_i\le c_2$ for all $0\le i\le n-1$. Then, we can choose $f_1(k)=\frac{c_2}{c_1}$ and
  $f_2(k)=\frac{c_2}{c_1}\cdot(k+1)$.
\end{lemma}

We remark that the above lemma is also valid for the more general probabilistic values if $c_1\le P_S\le c_2$
holds for all $S\subseteq N$.
      
\section{Details for a class of parameterized weighted games}
\label{sec_details_parameterized_wvg}

In this section we provide the delayed proofs from Subsection~\ref{subsec_parameterized_wvg}. Before going into the
details for $p$-binomial semivalues and the (absolute) Johnston index, we provide a useful bound for binomial coefficients:
 
\begin{lemma}
  \label{binomial_coefficient_bound}
  For integers $1\le m\le n$ we have
  $
    {{n-1}\choose {m-1}}\le \frac{2^{n-1}}{\sqrt{n}}
  $. 
\end{lemma}
\begin{proof}
  Since the binomial coefficients attain their maximum at the center, it suffices to prove the proposed inequality
  for $m-1=\left\lfloor\frac{n-1}{2}\right\rfloor$. For $t\ge 1$ we have ${{2t}\choose t}\le \frac{2^{2t}}{\sqrt{3t+1}}$.
  From ${{2t-1}\choose {t-1}}+{{2t-1}\choose {t}}={{2t}\choose {t}}$ and ${{2t-1}\choose {t-1}}={{2t-1}\choose {t}}$,
  we conclude ${{2t-1}\choose {t-1}}\le \frac{2^{2t-1}}{\sqrt{3t+1}}$ for all $t\ge 1$. Thus, we have
  ${{n-1}\choose {m-1}}\le \frac{2^{n-1}}{\sqrt{n}}$ for all $n\ge 1$.     
\end{proof}

\subsection{Technical details for $p$-binomial semivalues}
\label{subsec_technical_p_binomial}

We start our considerations with three technical estimates which will be used in the subsequent lemmas.
The common theoretical basis of the first two estimates is Hoeffding's inequality, which provides an upper
bound on the probability that the sum of random variables deviates from its expected value.

\begin{lemma}
  \label{lemma_hoeffding_application_1}
  If $m=\left\lceil\frac{n+1}{2}\right\rceil$ and $p=\frac{1}{2}+\delta$, where $\frac{1}{2}>\delta>0$, then
  $$
    \sum_{j=m}^n {n \choose j}p^j(1-p)^{n-j}\ge 1-\exp(-2n\delta^2).
  $$
\end{lemma}
\begin{proof}
  Since $\sum\limits_{j=0}^n {n \choose j}p^j(1-p)^{n-j}=1$, it suffices to prove
  $\sum\limits_{j=0}^{m-1} {n \choose j}p^j(1-p)^{n-j}\le \exp(-2n\delta^2)$. To this end, we
  apply Hoeffding's inequality for the special case of Bernoulli random variables and deduce
  $$
    \sum\limits_{j=0}^{m-1} {n \choose j}p^j(1-p)^{n-j}\le \exp\!\left(-2\frac{(np-(m-1))^2}{n}\right)
    \le\exp\!\left(-2\frac{(n\delta)^2}{n}\right)= \exp(-2n\delta^2).
  $$ 
\end{proof}

\begin{lemma}
  \label{lemma_hoeffding_application_2}
  If $m=\left\lfloor\frac{n}{2}\right\rfloor$ and $p=\frac{1}{2}-\delta$, where $\frac{1}{2}>\delta>0$, then
  $$
    \sum_{j=0}^{m-1} {n \choose j}p^j(1-p)^{n-j}\ge 1-\exp(-2n\delta^2+2\delta).
  $$
\end{lemma}
\begin{proof}
  Since $\sum\limits_{j=0}^{m-1} {n \choose j}p^j(1-p)^{n-j}=\sum\limits_{j=n-m+1}^{n}
  {n \choose j}(1-p)^j p^{n-j}$, we can similarly proceed as in the proof of
  Lemma~\ref{lemma_hoeffding_application_1} and prove 
  $\sum\limits_{j=0}^{n-m} {n \choose j}(1-p)^j p^{n-j}\le \exp(-2n\delta^2+2\delta)$.
  To this end, we apply Hoeffding's inequality for the special case of Bernoulli random variables and deduce
  \begin{eqnarray*}
    \sum\limits_{j=0}^{n-m} {n \choose j}p^j(1-p)^{n-j}&\le& \exp\!\left(-2\frac{(n(1-p)-(n-m))^2}{n}\right)\\
    &\le&\exp\!\left(-2\frac{(n\delta-\frac{1}{2})^2}{n}\right)\le \exp(-2n\delta^2+2\delta).
  \end{eqnarray*}
\end{proof}

\begin{lemma}
  \label{lemma_product_bound}
  If $m\in\big\{\left\lceil\frac{n+1}{2}\right\rceil,\left\lfloor\frac{n}{2}\right\rfloor\big\}$ and $p=\frac{1}{2}+\delta$
  or $p=\frac{1}{2}-\delta$, where $0\le\delta<\frac{1}{2}$, then
  $$
    p^{m-1}(1-p)^{n-m} \le 4\cdot \frac{(1-4\delta^2)^{(n-3)/2}}{2^{n-1}}.
  $$ 
\end{lemma}
\begin{proof}
  We have
  $$
    p^{m-1}(1-p)^{n-m}=\frac{(1\pm 2\delta)^{m-1}\cdot(1\mp 2\delta)^{n-m}}{2^{n-1}}
    \le \frac{(1-2\delta)^{(n-3)/2}\cdot(1+2\delta)^{(n+1)/2}}{2^{n-1}},    
  $$
  where the right hand side is bounded by $4\cdot \frac{(1-4\delta^2)^{(n-3)/2}}{2^{n-1}}$ due to $\delta<\frac{1}{2}$.
\end{proof}

Next we go on and determine formulas for the $p$-binomial semivalue $\Psi^p$ for the three games $\Gone$, $\Gtwo$, and
$\Gthree$ and the three types of {\voter}s:

\begin{lemma}
  \label{lemma_formula_for_p_binomial_semivalue_on_special_game}
  For $p\in(0,1)$ and integers $l,k,m,n$ satisfying the restrictions from Definition~\ref{def_parametric_game}
  and $1\le m\le n$, we have
  \begin{eqnarray*}
    \Psi_i^p(\Gone)   &=& {{k-1}\choose{l-1}} \cdot p^{l-1}\cdot (1-p)^{k-l}\\
                      && +1\cdot p^l(1-p)^{k-l-1}\cdot\sum_{j=0}^{m-1}{n\choose j}\cdot p^j (1-p)^{n-j}
                      \,\,\,\,\,\quad \forall i\in [1,k-l]\\
    \Psi_i^p(\Gone)   &=& \left({{k-1}\choose{l-1}}-1\right)\cdot p^{l-1}\cdot (1-p)^{k-l}\\              
                      &&+1\cdot p^{l-1}(1-p)^{k-l}\cdot\sum_{j=m}^{n}{n\choose j}\cdot p^j (1-p)^{n-j}
                      \,\,\,\,\,\quad \forall i\in(k-l,k]\\
    \Psi_i^p(\Gone)   &=& 1\cdot p^l(1-p)^{k-l}\cdot{{n-1}\choose{m-1}}\cdot p^{m-1}(1-p)^{n-m}
                      \quad\quad \forall i\in(k,k+n]\\
    \Psi_i^p(\Gtwo)   &=& {{k-1}\choose{l-1}} \cdot p^{l-1}\cdot (1-p)^{k-l}
                      +1\cdot p^l(1-p)^{k-l-1}
                      \quad \forall i\in[1,k-l]\\
    \Psi_i^p(\Gtwo)   &=& \left({{k-1}\choose{l-1}}-1\right)\cdot p^{l-1}\cdot (1-p)^{k-l}              
                      \quad\quad\quad\quad\quad\quad\quad \forall i\in(k-l,k]\\
    \Psi_i^p(\Gtwo)   &=& 0
                      \quad\quad\quad\quad\quad\quad\quad\quad\quad\quad
                      \quad\quad\quad\quad\quad\quad\quad\quad\quad\quad\quad 
                      \quad \forall i\in(k,k+n]\\
    \Psi_i^p(\Gthree) &=& {{k-1}\choose{l-1}} \cdot p^{l-1}\cdot (1-p)^{k-l}
                      \quad\quad\quad\quad\quad\quad\quad\quad\quad\,\,\,
                      \quad \forall i\in[1,k-l]\\
    \Psi_i^p(\Gthree) &=& {{k-1}\choose{l-1}}\cdot p^{l-1}\cdot (1-p)^{k-l}
                      \quad\quad\quad\quad\quad\quad\quad\quad\quad\,\,\,              
                      \quad \forall i\in(k-l,k]\\
    \Psi_i^p(\Gthree) &=& 0
                      \quad\quad\quad\quad\quad\quad\quad\quad\quad\,\,\,
                      \quad\quad\quad\quad\quad\quad\quad\quad\quad\quad\quad
                      \quad\,\,\, \forall i\in(k,k+n]\\                     
  \end{eqnarray*}                     
\end{lemma}
\begin{proof}
  We note that the $i$-swings of $\Gone$, $\Gtwo$, and $\Gthree$ all are of one of the forms
  $U\subseteq [k]$ with $|U|=l-1$ and $T\cup V$ with $V\subseteq (k,k+n]$.
\end{proof}

Now we are interested in the differences $\Psi^p(\Gone)-\Psi^p(\Gtwo)$ and $\Psi^p(\Gone)-\Psi^p(\Gthree)$.
Inserting into the terms of Definition~\ref{def_delta} yields:

\begin{lemma}
  \label{lemma_differences_for_p_binomial_semivalue_on_special_game}
  For $p\in(0,1)$, $P=\Psi^p$, and integers $l,k,m,n$ satisfying the restrictions from Definition~\ref{def_parametric_game}
  and $1\le m\le n$, we have
  \begin{eqnarray*}
    \Delta_{1,2}^{1,P} &=& p^{l}(1-p)^{k-l-1}\cdot \sum_{j=m}^n {n \choose j}p^j(1-p)^{n-j}, \\
    \Delta_{1,2}^{2,P} &=& p^{l-1}(1-p)^{k-l}\cdot \sum_{j=m}^n {n \choose j}p^j(1-p)^{n-j},\\
    \Delta_{1,2}^{3,P} &=& 1\cdot p^l(1-p)^{k-l}\cdot{{n-1}\choose{m-1}}\cdot p^{m-1}(1-p)^{n-m}, \\
    \Delta_{1,3}^{1,P} &=& p^{l}(1-p)^{k-l-1}\cdot \sum_{j=0}^{m-1} {n \choose j}p^j(1-p)^{n-j}, \\
    \Delta_{1,3}^{2,P} &=& p^{l-1}(1-p)^{k-l}\cdot \sum_{j=0}^{m-1} {n \choose j}p^j(1-p)^{n-j},\\
    \Delta_{1,3}^{3,P} &=& 1\cdot p^l(1-p)^{k-l}\cdot{{n-1}\choose{m-1}}\cdot p^{m-1}(1-p)^{n-m},\text{ and}\\
    \xi^P              &=& n\Delta_{1,2}^{3,P}=n\Delta_{1,3}^{3,P}=n\cdot p^l(1-p)^{k-l}\cdot{{n-1}\choose{m-1}}\cdot p^{m-1}(1-p)^{n-m}.
  \end{eqnarray*}
\end{lemma}

\begin{proof} (of Lemma~\ref{lemma_limit_results_p_binomial})
    \begin{itemize}
    \item[(1)] We set $p=\frac{1}{2}+\delta$, where $\frac{1}{2}>\delta>0$, and conclude 
               $$
                 \frac{\Delta_{1,2}^{1,P}}{\xi^P}= \frac{\sum\limits_{j=m}^n {n \choose j}p^j(1-p)^{n-j}}
                 {n(1-p)\cdot{{n-1}\choose{m-1}}\cdot p^{m-1}(1-p)^{n-m}}
                 \ge \frac{1}{n(1-p)}\cdot \frac{1-\exp(-2n\delta^2)}{\frac{2^{n-1}}{\sqrt{n}}\cdot
                 4\frac{(1-4\delta^2)^{(n-3)/2}}{2^{n-1}}}
               $$
               using Lemma~\ref{lemma_hoeffding_application_1}, Lemma~\ref{binomial_coefficient_bound}, and
               Lemma~\ref{lemma_product_bound}. The last expression can be rewritten to
               $$
                 \frac{1}{4(1-p)}\cdot \frac{1-\exp(-2n\delta^2)}{\sqrt{n}(1-4\delta^2)^{(n-3)/2}},
               $$
               which clearly tends to $\infty$ as $n$ increases. Since $\Delta_{1,2}^{2,P}=\Delta_{1,2}^{1,P}\cdot\frac{1-p}{p}$,
               we can also deduce the second statement.
    \item[(2)] Similarly, we set $p=\frac{1}{2}-\delta$, where $\frac{1}{2}>\delta>0$, and conclude 
               $$
                 \frac{\Delta_{1,3}^{1,P}}{\xi^P}= \frac{\sum\limits_{j=0}^{m-1} {n \choose j}p^j(1-p)^{n-j}}
                 {n(1-p)\cdot{{n-1}\choose{m-1}}\cdot p^{m-1}(1-p)^{n-m}}
                 \ge \frac{1}{n(1-p)}\cdot \frac{1-\exp(-2n\delta^2+2\delta)}{\frac{2^{n-1}}{\sqrt{n}}\cdot
                 4\frac{(1-4\delta^2)^{(n-3)/2}}{2^{n-1}}}
               $$
               using Lemma~\ref{lemma_hoeffding_application_2}, Lemma~\ref{binomial_coefficient_bound}, and
               Lemma~\ref{lemma_product_bound}. The last expression can be rewritten to
               $$
                 \frac{1}{4(1-p)}\cdot \frac{1-\exp(-2n\delta^2+2\delta)}{\sqrt{n}(1-4\delta^2)^{(n-3)/2}},
               $$
               which clearly tends to $\infty$ as $n$ increases. Since $\Delta_{1,3}^{2,P}=\Delta_{1,3}^{1,P}\cdot\frac{1-p}{p}$,
               we can also deduce the second statement.
    \item[(3)] From Lemma~\ref{binomial_coefficient_bound} and Lemma~\ref{lemma_product_bound} we conclude
               $$
                 \xi^P\le n\cdot{{n-1}\choose{m-1}}\cdot p^{m-1}(1-p)^{n-m}\le 4\sqrt{n}\cdot(1-4\delta^2)^{(n-3)/2},
               $$
               where the right hand side clearly tends to zero as $n$ increases.
  \end{itemize}  
\end{proof}

\begin{proof} (of Lemma~\ref{lemma_bounded_change_p_binomial})
  For $i\in[1,k-l]$ we have 
  $$
    \left|\Psi_i^p(\Gone)-\Psi_i^p(\Gtwo)\right|=
    \left|\Psi_i^p(\Gone)-\Psi_i^p(\Gthree)\right|
    \le p^l(1-p)^{k-l-1}\le p^{l-1}(1-p)^{k-l-1}
  $$
  and
  for $i\in(k-k,k]$ we have 
  $$
    \left|\Psi_i^p(\Gone)-\Psi_i^p(\Gtwo)\right|=
    \left|\Psi_i^p(\Gone)-\Psi_i^p(\Gthree)\right|
    \le p^{l-1}(1-p)^{k-l}\le p^{l-1}(1-p)^{k-l-1}.
  $$  
  Since we have defined $\frac{1}{n}\cdot \xi^P$ as 
  $$
    \left|\Psi_i^p(\Gone)-\Psi_i^p(\Gtwo)\right|=\left|\Psi_i^p(\Gone)-\Psi_i^p(\Gthree)\right|,
  $$ where $i\in(k,k+n]$, applying the triangle inequality for the absolute values gives the stated inequality.
\end{proof}

\pagebreak

\subsection{Technical details for the (absolute) Johnston index}
\label{subsec_technical_johnston}

\begin{lemma}
  \label{lemma_formula_for_johnston_on_special_game}
  For $l=1$ and integers $k,m,n$ satisfying the restrictions from Definition~\ref{def_parametric_game}
  and $1\le m\le n$, we have
  \begin{eqnarray*}
    \operatorname{JS}_i(\Gone) &=& \sum_{j=0}^n {n \choose j}\cdot 1 +\sum_{j=0}^{m-1}{n \choose j}\cdot 1
                =2^n+\sum_{j=0}^{m-1}{n \choose j} \quad\quad \forall i\in[1,k) \\
    \operatorname{JS}_k(\Gone) &=& \sum_{j=m+1}^{n}{n \choose j}\cdot 1\,+\,{n \choose m}\cdot\frac{1}{m+1}\\              
    \operatorname{JS}_i(\Gone) &=& {{n-1} \choose {m-1}}\cdot\frac{1}{m+1}
                \,\,\,\,\quad\quad\quad\quad\quad\quad\quad\quad\quad\quad\quad\quad\quad \forall i\in(k,k+n]   \\
    \operatorname{JS}_i(\Gtwo) &=& \sum_{j=0}^n {n \choose j}\cdot 1 =2^n
                \,\,\,\,\,\quad\quad\quad\quad\quad\quad\quad\quad\quad\quad\quad\quad\quad \forall i\in[1,k]   \\
    \operatorname{JS}_i(\Gtwo) &=& 0
                \,\,\,\,\,\quad\quad\quad\quad\quad\quad\quad\quad\quad\quad\quad\quad\quad\quad\quad\quad\quad\quad
                \quad\quad \forall i\in(k,k+n] \\
    \operatorname{JS}_i(\Gthree) &=& \sum_{j=0}^{n+1} {n \choose j}\cdot 1 =2^{n+1} 
                \,\,\,\quad\quad\quad\quad\quad\quad\quad\quad\quad\quad\quad\quad \forall i\in[1,k) \\
    \operatorname{JS}_i(\Gthree) &=& 0
                \,\,\,\,\quad\quad\quad\quad\quad\quad\quad\quad\quad\quad\quad\quad\quad\quad\quad\quad\quad\quad
                \quad\quad \forall i\in[k,k+n] \\
    \xi^{\operatorname{JS}} &=& n\cdot {{n-1} \choose {m-1}}\cdot\frac{1}{m+1}.              
  \end{eqnarray*}  
\end{lemma}
\begin{proof}
  We note that all swing coalitions are of the form $\{i\}\cup T$, $\{i,k\}\cup T$, and $\{k\}\cup T$, where $1\le i\le k-1$
  and $\emptyset\subseteq T\subseteq (k,n+k]$.   
\end{proof}

\begin{proof}(of Lemma~\ref{lemma_Johnston_bounds})
  From Lemma~\ref{binomial_coefficient_bound} and Lemma~\ref{lemma_formula_for_johnston_on_special_game} we conclude
  \begin{eqnarray*}
    \operatorname{JS}_i(\Gone) &=&   3\cdot 4^{\tilde{n}}\quad \quad\quad\quad\quad\quad \forall i\in[1,k)\\
    \operatorname{JS}_k(\Gone) &\ge& 4^{\tilde{n}} -\frac{\sqrt{2}\cdot 4^{\tilde{n}}}{\sqrt{\tilde{n}+1}}\\
    \operatorname{JS}_k(\Gone) &\le& 4^{\tilde{n}}\\
    \xi^{\operatorname{JS}} &\le& \frac{\sqrt{2}\cdot 4^{\tilde{n}}}{\sqrt{\tilde{n}}}.
  \end{eqnarray*}  
  For $\tilde{n}\ge 1$ we then have
  \begin{eqnarray*}
    \Vert \operatorname{JS}(\Gone)-\operatorname{JS}(\Gtwo)\Vert_1 
    &\ge& 
    k\cdot 4^{\tilde{n}}+\xi^{\operatorname{JS}}
    \ge 
    \left(k\cdot\frac{\sqrt{\tilde{n}}}{\sqrt{2}}+1\right)\xi^{\operatorname{JS}},\\
    \Vert \operatorname{JS}(\Gone)-\operatorname{JS}(\Gthree)\Vert_1&\ge&
    k\cdot 4^{\tilde{n}} -\frac{\sqrt{2}\cdot 4^{\tilde{n}}}{\sqrt{\tilde{n}+1}}+\xi^{\operatorname{JS}}
    \ge (k-1)4^{\tilde{n}}+\xi^{\operatorname{JS}}\\
    &\ge&
    \left((k-1)\cdot\frac{\sqrt{\tilde{n}}}{\sqrt{2}}+1\right)\xi^{\operatorname{JS}},\text{ and}\\
    \lim\limits_{\tilde{n}\to\infty} \frac{\xi^{\operatorname{JS}}}{\sum\limits_{i=1}^{n+k} \operatorname{JS}_i(\Gone)}&=& 0.    
   \end{eqnarray*}
\end{proof}

\begin{proof}(of Lemma~\ref{lemma_normalized Johnston_bounds})
  Summing up the inequalities from Lemma~\ref{lemma_Johnston_bounds} yields
  \begin{eqnarray*}
    (3k-2)\cdot 4^{\tilde{n}}-\frac{\sqrt{2}\cdot 4^{\tilde{n}}}{\sqrt{\tilde{n}}} \le\sum_{i=1}^{k+n} 
    \operatorname{JS}_i(\Gone)\le (3k-2)\cdot 4^{\tilde{n}}+\frac{\sqrt{2}\cdot 4^{\tilde{n}}}{\sqrt{\tilde{n}}},\\
    \sum_{i=1}^{k+n} \operatorname{JS}_i(\Gtwo)= (2k)\cdot 4^{\tilde{n}},\text{ and}\\
    \sum_{i=1}^{k+n} \operatorname{JS}_i(\Gthree)=(4k-4)\cdot 4^{\tilde{n}}.
  \end{eqnarray*}
  For $i\in[1,k)$ we have
  $$
    \left|\frac{\operatorname{JS}_i(\Gone)}{\sum_{j=1}^{n+k} \operatorname{JS}_j(\Gone)}
    -\frac{\operatorname{JS}_i(\Gtwo)}{\sum_{j=1}^{n+k} \operatorname{JS}_j(\Gtwo)}\right|\ge
    \left|\frac{3\cdot 4^{\tilde{n}}}{(3k-2)\cdot 4^{\tilde{n}}+\frac{\sqrt{2}\cdot 4^{\tilde{n}}}
    {\sqrt{\tilde{n}}}}-\frac{2\cdot 4^{\tilde{n}}}{2k\cdot 4^{\tilde{n}}}\right|,
  $$
  where the right hand side can be lower bounded by $\frac{1}{5k(k-1)}$ for sufficiently large $\tilde{n}$, so that
  $$
    \sum_{i=1}^{n+k} \left|\frac{\operatorname{JS}_i(\Gone}{\sum_{j=1}^{n+k} \operatorname{JS}_j(\Gone)}
    -\frac{\operatorname{JS}_i(\Gtwo)}{\sum_{j=1}^{n+k} \operatorname{JS}_j(\Gtwo)}\right|
    \ge \frac{1}{5k}.
  $$
  Similarly we deduce
  $$
    \sum_{i=1}^{n+k} \left|\frac{\operatorname{JS}_i(\Gone)}{\sum_{j=1}^{n+k} \operatorname{JS}_j(\Gone)}
    -\frac{\operatorname{JS}_i(\Gthree)}{\sum_{j=1}^{n+k} \operatorname{JS}_j(\Gthree)}\right|
    \ge \frac{1}{5k}.
  $$
  Finally, we conclude
  $$
    \frac{\xi^{\operatorname{Js}}}{\sum_{j=1}^{n+k} \operatorname{JS}_j(\Gone)}\le \frac{\sqrt{2}}{(3k-3)\sqrt{\tilde{n}}}
  $$ 
  for $\tilde{n}\ge 2$.
\end{proof}

\section{ILP formulations for the counting functions of almost all power indices from Section~\ref{sec_power_indices}}
\label{sec_ILP_counting_functions}

In Subsection~\ref{subsec_ILP_power_index} we have described the general idea of modeling power indices in the
framework of integer linear programming. Our approach is based on the idea of a counting function. Since we
can treat the normalized version of a power index, whenever we can treat the original version, see 
Subsection~\ref{subsec_ILP_deviation} for the algorithmic details, we restrict ourselves\footnote{Exact
linearization techniques for fractional linear terms can e.g.\ be found in \cite{liberti2007techniques}. See also
the approach described in Subsection~\ref{subsec_ILP_ColPrev}.} on
presenting ILP formulations for the absolute versions of the power indices introduced in Section~\ref{sec_power_indices}.
In the following subsections we will state sufficient constraints for each power index separately. Those constraints
for $y_{i,S}$ have to be required for all $i\in N$, which we will not repeat at each place.

\subsection{$\operatorname{SSI}$}
\label{subsec_ILP_SSI}
\begin{eqnarray}
  y_{i,S} &=& 0 \quad\quad\quad\quad\quad\quad\quad\quad\quad\quad\quad\quad\quad\quad\quad\quad
                \!\!\quad\forall S\subseteq N\backslash\{i\},\\
  y_{i,S} &\le& \frac{(|S|-1)!(n-|S|)!}{n!}\cdot x_S \quad\quad\quad\quad\quad\quad\,\,\forall \{i\}\subseteq S\subseteq N,\\
  y_{i,S} &\le& \frac{(|S|-1)!(n-|S|)!}{n!}\cdot \left(1-x_{S\backslash\{i\}}\right) \,\,\,\,\quad
                \forall \{i\}\subseteq S\subseteq N\text{ and}\\
  y_{i,S} &\ge& \frac{(|S|-1)!(n-|S|)!}{n!}\cdot \left(x_S-x_{S\backslash\{i\}}\right)\quad\forall \{i\}\subseteq S\subseteq N.  
\end{eqnarray}
For subclasses of simple games this simplifies to
\begin{eqnarray}
  y_{i,S} &=& 0 \quad\quad\quad\quad\quad\quad\quad\quad\quad\quad\quad\quad\quad\quad\quad\quad 
              \quad\forall S\subseteq N\backslash\{i\} \text{ and}\\
  y_{i,S} &=& \frac{(|S|-1)!(n-|S|)!}{n!}\cdot \left(x_S-x_{S\backslash\{i\}}\right)
              \quad\,\forall \{i\}\subseteq S\subseteq N.
\end{eqnarray}  

\subsection{$\operatorname{Tijs}$}
\label{subsec_ILP_Tijs}
As an incidence vector for minimal winning coalitions we introduce binary variables $z_S^m\in\{0,1\}$ for all
$S\subseteq N$. The constraints
\begin{eqnarray}
  z_S^m &\le& x_S\quad\quad\quad\quad\quad\quad\quad\,\forall S\subseteq N,\\
  z_S^m &\le& 1-x_{S\backslash\{j\}}\,\,\,\quad\quad\quad \forall S\subseteq N, j\in S,\text{ and}\\
  z_S^m &\ge& x_S-\sum_{j\in S} x_{S\backslash\{j\}}\quad\forall S\subseteq N
\end{eqnarray}
guarantee that $z_S^m=1$ if and only if $S$ is a minimal winning coalition. Next we introduce $m\in\{0,1\}$ with the
interpretation that $m=1$ if and only if $\left|\mathcal{W}^m\right|=1$. This can be enforced by
\begin{eqnarray}
  m &\le& \sum_{S\subseteq N} z_S^m,\\
  m &\ge& z_S^m-\sum_{T\subseteq N, T\neq S} z_T^m\quad\quad\quad\quad\forall S\subseteq N,\text{ and}\\
  \sum_{S\subseteq N} z_S^m&\le& 1+(1-m)\cdot \left(2^n-1\right).
\end{eqnarray}
With this we can complete the ILP formulation by
\begin{eqnarray}
  y_{i,S} &\le& z_S^m,\\
  y_{i,S} &\le& m, \text{ and}\\
  y_{i,S} &\ge& z_S^m+m-1
\end{eqnarray}
for all $S\subseteq N$.

\subsection{Semivalues $\Psi^{\mathbf{p}}$}
\label{subsec_ILP_semivalues}
As the Shapley-Shubik index is a special semivalue we might state an ILP formulation similar to those in
Subsection~\ref{subsec_ILP_SSI}. In order to better highlight the underlying concept we  introduce binary
variables $z_{i,S}^s\in\{0,1\}$ for all
$S\subseteq N$. The constraints
\begin{eqnarray}
  z_{i,S}^s &\le& x_S\quad\quad\quad\quad\quad\quad\quad\,\forall S\subseteq N,\\
  z_{i,S}^s &\le& 1-x_{S\backslash\{i\}}\,\,\,\quad\quad\quad \forall S\subseteq N,\text{ and}\\
  z_{i,S}^s &\ge& x_S-x_{S\backslash\{i\}}\!\quad\quad\quad\forall S\subseteq N
\end{eqnarray}
guarantee that $z_{i,S}^s=1$ if and only if $i\in S$ and $S\backslash\{i\}$ is a swing coalition for $i$. For
subclasses of simple games this simplifies to
\begin{equation}
  z_{i,S}^s = x_S-x_{S\backslash\{i\}}\quad\forall S\subseteq N.
\end{equation}
With this we can set $y_{i,S}=p_{|S|-1}\cdot z_{i,S}^s$. 

\subsection{Binomial semivalues $\Psi^{p}$}
Using the notation from Subsection~\ref{subsec_ILP_semivalues} we set $y_{i,S}=p^{|S|-1}(1-p)^{n-|S|}\cdot z_{i,S}^s$. 

\subsection{$\operatorname{Bz}$}
Using the notation from Subsection~\ref{subsec_ILP_semivalues} we set $y_{i,S}=1/2^{n-1}\cdot z_{i,S}^s$.

\subsection{$\operatorname{swing}$}
Using the notation from Subsection~\ref{subsec_ILP_semivalues} we set $y_{i,S}=z_{i,S}^s$.

\subsection{$\operatorname{ColPrev}$}
\label{subsec_ILP_ColPrev} 
We use the notation from Subsection~\ref{subsec_ILP_semivalues}. Things would be very easy if we could write
$y_{i,S}=\frac{1}{|\mathcal{W}|}\cdot z_{i,S}^s$, but unfortunately even $\frac{1}{|\mathcal{W}|}$ is nonlinear,
since $|\mathcal{W}|=\sum_{S\subseteq N}x_S$. Introducing the nonnegative real variables $a_S\in\mathbb{R}_{\ge 0}$
for all $S\subseteq N$, we can model the distribution of $\frac{1}{|\mathcal{W}|}$ to each winning coalition as follows:
\begin{eqnarray}
  \label{ie_ILP_ColPrev_1} a_S &\le& x_S\quad\quad\quad\quad\quad\,\, \forall S\subseteq N,\\
  \label{ie_ILP_ColPrev_2} a_S-a_T &\ge& x_S+x_T-2\quad \forall S,T\subseteq N,\text{ and}\\
  \label{ie_ILP_ColPrev_3} \sum_{S\subseteq N} a_S &=& 1.
\end{eqnarray}
By Inequality~(\ref{ie_ILP_ColPrev_1}) we have $a_S=0$ for all losing coalitions and $a_S\le 1$ for all
winning coalitions. So especially we have $0\le a_S\le 1$ for all $S\subseteq N$. Thus $a_S-a_T\ge -1$ for
all $S,T\subseteq N$, so that Inequality~(\ref{ie_ILP_ColPrev_2}) is only a restriction if both coalitions
$S$ and $T$ are winning. In this case we have $a_S-a_T\ge 0$ and due to the symmetric formulation also
$a_T-a_S\ge 0$, so that we finally have $a_S=a_T$ (Here we also have applied the Big-M method with $M=1$.). With
Equation~(\ref{ie_ILP_ColPrev_3}) we obtain $a_S=\frac{1}{|\mathcal{W}|}$ for all winning coalitions $S$. It remains
to assign the value $\frac{1}{|\mathcal{W}|}$ only to swing coalitions instead of winning coalitions:
\begin{eqnarray}
  \label{ie_ILP_ColPrev_4} y_{i,S} &\le& z_{i,S}^s,\\
  \label{ie_ILP_ColPrev_5} y_{i,S} &\ge& a_S-\left(1-z_{i,S}^s\right),\text{ and}\\
  \label{ie_ILP_ColPrev_6} y_{i,S} &\le& a_S+\left(1-x_S\right)
\end{eqnarray}   
for all $S\subseteq N$.
If either $i\notin S$ or $S\backslash\{i\}$ is not a swing coalition for {\voter}~$i$, then
Inequality~(\ref{ie_ILP_ColPrev_4}) ensures $y_{i,S}=0$ due to $y_{i,S}\ge 0$. Otherwise we have
$z_{i,S}^s=x_S=1$, so that inequalities (\ref{ie_ILP_ColPrev_5}) and (\ref{ie_ILP_ColPrev_6}) are
equivalent to $y_{i,S}\ge a_S$ and $y_{i,S}\le a_S$, respectively. Thus we have $y_{i,S}=\frac{1}{|\mathcal{W}|}$
as requested. For either $z_{i,S}^s=0$ or $x_S=0$ inequalities (\ref{ie_ILP_ColPrev_5}) and (\ref{ie_ILP_ColPrev_6})
are satisfied automatically due to $0\le a_S\le 1$.

\subsection{$\operatorname{ColIni}$}
We similarly proceed as in Subsection~\ref{subsec_ILP_ColPrev} and introduce the nonnegative real variables
$a_S\in\mathbb{R}_{\ge 0}$, for all $S\subseteq N$. With this we can model the distribution of $\frac{1}{|\mathcal{L}|}$
to each losing coalition as follows:
\begin{eqnarray}
  \label{ie_ILP_ColIni_1} a_S &\le& 1-x_S\,\,\quad\quad\,\, \forall S\subseteq N,\\
  \label{ie_ILP_ColIni_2} a_S-a_T &\ge& -x_S-x_T\quad \forall S,T\subseteq N,\text{ and}\\
  \label{ie_ILP_ColIni_3} \sum_{S\subseteq N} a_S &=& 1.
\end{eqnarray} 
It remains to assign the value $\frac{1}{|\mathcal{L}|}$ only to swing coalitions instead of losing coalitions:
\begin{eqnarray}
  \label{ie_ILP_ColIni_4} y_{i,S} &=& 0\quad\quad\quad\quad\quad\quad\quad\quad\quad\quad\quad\,\,\,\forall \{i\}\subseteq S\subseteq N,\\
  \label{ie_ILP_ColIni_5} y_{i,S} &\le& 1-x_S\quad\quad\quad\quad\quad\quad\quad\quad\quad\forall S\subseteq N\backslash\{i\},\\
  \label{ie_ILP_ColIni_6} y_{i,S} &\le& x_{S\cup\{i\}}\quad\quad\quad\quad\quad\quad\quad\quad\quad\forall S\subseteq N\backslash\{i\},\\
  \label{ie_ILP_ColIni_7} y_{i,S} &\ge& a_S-\left(x_S+1-x_{S\cup\{i\}}\right) \quad\forall S\subseteq N\backslash\{i\},\text{ and}\\
  \label{ie_ILP_ColIni_8} y_{i,S} &\le& a_S+\left(x_S+1-x_{S\cup\{i\}}\right) \quad\forall S\subseteq N\backslash\{i\}.
\end{eqnarray}   

\subsection{$\operatorname{Rae}$}
\begin{eqnarray}
  y_{i,S} &=& \frac{1}{2^{|N|}}\cdot x_S\quad\quad\quad\quad \forall \{i\}\subseteq S\subseteq{N}\text{ and}\\
  y_{i,S} &=& \frac{1}{2^{|N|}}\cdot \left(1-x_S\right)\quad\forall S\subseteq N\backslash\{i\}.
\end{eqnarray}

\subsection{$\operatorname{KB}$}
Using the notation from Subsection~\ref{subsec_ILP_ColPrev} we can state
\begin{eqnarray}
  y_{i,S} &=& a_S \quad\forall \{i\}\subseteq S\subseteq N\text{ and}\\
  y_{i,S} &=& 0 \,\,\,\,\quad \forall S\subseteq N\backslash\{i\}. 
\end{eqnarray}

\subsection{$\operatorname{PHI}$}
Similarly as in Subsection~\ref{subsec_ILP_ColPrev} we obtain:
\begin{eqnarray}
  y_{i,S} &=& 0\quad\quad\quad\quad\quad\quad\quad\forall S\subseteq N\backslash\{i\},\\
  y_{i,S} &\le & x_S\quad\quad\quad\quad\quad\quad\,\,\forall \{i\}\subseteq S\subseteq N,\\
  y_{i,S} &\ge & a_S-(1-x_S)\quad\,\forall \{i\}\subseteq S\subseteq N,\\
  y_{i,S} &\le & a_S+(1-x_S)\quad\,\forall \{i\}\subseteq S\subseteq N,\\
  a_S &\le& x_S\quad\quad\quad\quad\quad\quad\,\, \forall S\subseteq N,\\
  a_S-a_T &\ge& x_S+x_T-2\quad\quad \forall S,T\subseteq N,\text{ and}\\
  \sum_{S\subseteq N} |S|\cdot a_S &=& 1.
\end{eqnarray}

\subsection{$\operatorname{Chow}$}
\begin{eqnarray}
  y_{i,S} &=& x_S \quad\forall \{i\}\subseteq S\subseteq N\text{ and}\\
  y_{i,S} &=& 0 \,\,\,\,\quad \forall S\subseteq N\backslash\{i\}.
\end{eqnarray}

\subsection{$\operatorname{JS}$}
With the notation from Subsection~\ref{subsec_ILP_semivalues} we can state:
\begin{eqnarray}
  y_{i,S} &\le& z_{i,S}^s \quad\quad\quad\quad\quad\quad \forall S\subseteq N,\\ 
  y_{i,S}-y_{j,S} &\ge& z_{i,S}^s+z_{j,S}^s-2\quad \forall i,j\in N, S\subseteq N,\\
  \sum_{i=1}^n y_{i,S} &\le& 1\quad\quad\quad\quad\quad\quad\quad\, \forall S\subseteq N,\\
  \sum_{i=1}^n y_{i,S} &\le& \sum_{i=1}^n z_{i,S}^s\quad\quad\quad\quad\, \forall S\subseteq N,\text{ and}\\
  \sum_{j=1}^n y_{j,S} &\ge& z_{i,S}^s\quad\quad\quad\quad\quad\quad \forall S\subseteq N, i\in N.
\end{eqnarray}

\subsection{$\operatorname{PGI}$}
\begin{eqnarray}
  y_{i,S} &=& 0\quad\quad\quad\quad\quad\quad\quad\,\,\,\,\forall S\subseteq N\backslash\{i\},\\
  y_{i,S} &\le& x_S\quad\quad\quad\quad\quad\quad\,\,\,\,\,\,\forall \{i\}\subseteq S\subseteq N,\\
  y_{i,S} &\le& 1-x_{S\backslash\{j\}}\,\,\quad\quad\quad\forall \{i\}\subseteq S\subseteq N, j\in S,\text{ and}\\
  y_{i,S} &\ge& x_S-\sum_{j\in S} x_{S\backslash\{j\}}\!\quad\forall \{i\}\subseteq S\subseteq N.
\end{eqnarray}

\subsection{$\operatorname{DP}$}
\begin{eqnarray}
  y_{i,S} &=& 0\quad\quad\quad\quad\quad\quad\quad\quad\quad\quad\quad\,\,\,\,\,\,\forall S\subseteq N\backslash\{i\},\\
  y_{i,S} &\le& \frac{1}{|S|}\cdot x_S\quad\quad\quad\quad\quad\quad\quad\quad\,\,\,\,\,\,\forall \{i\}\subseteq S\subseteq N,\\
  y_{i,S} &\le& \frac{1}{|S|}\cdot\left(1-x_{S\backslash\{j\}}\right)\,\,\quad\quad\quad\quad\forall \{i\}\subseteq S\subseteq N, j\in S,\text{ and}\\
  y_{i,S} &\ge& \frac{1}{|S|}\cdot\left(x_S-\sum_{j\in S} x_{S\backslash\{j\}}\right)\!\quad\forall \{i\}\subseteq S\subseteq N.
\end{eqnarray}

\subsection{$\operatorname{Shift}$}
\begin{eqnarray}
  y_{i,S} &=& 0\quad\quad\quad\quad\quad\quad\quad\quad\quad\quad\quad\quad\,\forall S\subseteq N\backslash\{i\},\\
  y_{i,S} &\le& x_S\quad\quad\quad\quad\quad\quad\quad\quad\quad\quad\quad\,\,\,\forall \{i\}\subseteq S\subseteq N,\\
  y_{i,S} &\le& 1-x_{T}\,\,\quad\quad\quad\quad\quad\quad\quad\quad
                \,\quad\forall \{i\}\subseteq S\subseteq N, T\text{ direct right-shift of }S,\text{ and}\\
  y_{i,S} &\ge& x_S-\!\!\!\!\!\sum_{T\text{ direct right-shift of }S} \!\!\!\!\!\!\!\!\!\!x_{T}
                \quad\quad\quad\quad\forall \{i\}\subseteq S\subseteq N.
\end{eqnarray}

\subsection{$\operatorname{SDP}$}
\begin{eqnarray}
  y_{i,S} &=& 0\quad\quad\quad\quad\quad\quad\quad\,\,\quad\quad\quad\quad\quad\quad\quad\,\forall S\subseteq N\backslash\{i\},\\
  y_{i,S} &\le& \frac{1}{|S|}\cdot x_S\quad\quad\quad\quad\quad\quad\quad\quad\quad\quad\quad\,\,\,\forall \{i\}\subseteq S\subseteq N,\\
  y_{i,S} &\le& \frac{1}{|S|}\cdot\left(1-x_{T}\right)\,\,\quad\quad\quad\quad\quad\quad
                \,\quad\quad\forall \{i\}\subseteq S\subseteq N, T\text{ direct right-shift of }S,\text{ and}\\
  y_{i,S} &\ge& \frac{1}{|S|}\cdot\left(x_S-\!\!\!\!\!\sum_{T\text{ direct right-shift of }S} \!\!\!\!\!\!\!\!\!\!x_{T}\right)
                \quad\quad\forall \{i\}\subseteq S\subseteq N.
\end{eqnarray}

\section{Counting functions for almost all power indices from Section~\ref{sec_power_indices}}

\begin{eqnarray*}
  C_i^{\operatorname{SSI}}(\game,S) &=&
  \left\{\begin{array}{rcl}\frac{(|S|-1)!(n-|S|)!}{n!}&:&i\in S,\game(S)=1,\game(S-i)=0,\\0&:&\text{otherwise},\end{array}\right.\\
  C_i^{\operatorname{Tijs}}(\game,S) &=&
  \left\{\begin{array}{rcl}1&:&S\in \mathcal{W}_i^m,\left|\mathcal{W}^m\right|=1,\\0&:&\text{otherwise},\end{array}\right.\\
  C_i^{\Psi^{\mathbf{p}}}(\game,S) &=&
  \left\{\begin{array}{rcl}p_{|S|-1}&:&i\in S,\game(S)=1,\game(S-i)=0,\\0&:&\text{otherwise},\end{array}\right.\\
  C_i^{\Psi^{p}}(\game,S) &=&
  \left\{\begin{array}{rcl}p^{|S|-1}(1-p)^{n-|S|}&:&i\in S,\game(S)=1,\game(S-i)=0,\\0&:&\text{otherwise},\end{array}\right.\\
  C_i^{\operatorname{Bz}}(\game,S) &=&
  \left\{\begin{array}{rcl}1/2^{n-1}&:&i\in S,\game(S)=1,\game(S-i)=0,\\0&:&\text{otherwise},\end{array}\right.\\
  C_i^{\operatorname{swing}}(\game,S) &=&
  \left\{\begin{array}{rcl}1&:&i\in S,\game(S)=1,\game(S-i)=0,\\0&:&\text{otherwise},\end{array}\right.\\
  C_i^{\operatorname{ColPrev}}(\game,S) &=&
  \left\{\begin{array}{rcl}\frac{1}{|\mathcal{W}|}&:&i\in S,\game(S)=1,\game(S-i)=0,\\0&:&\text{otherwise},\end{array}\right.\\
  C_i^{\operatorname{ColIni}}(\game,S) &=&
  \left\{\begin{array}{rcl}\frac{1}{|\mathcal{L}|}&:&i\notin S,\game(S)=0,\game(S+i)=1,\\0&:&\text{otherwise},\end{array}\right.\\
  C_i^{\operatorname{Rae}}(\game,S)&=&
  \left\{\begin{array}{rcl}\frac{1}{2^{|N|}}&:& i\in S\in\mathcal{W}\text{ or }i\notin S\notin\mathcal{W},\\
  0&:&\text{otherwise},\end{array}\right.\\
  C_i^{\operatorname{KB}}(\game,S) &=&
  \left\{\begin{array}{rcl}\frac{1}{|\mathcal{W}|}&:&S\in\mathcal{W}_i,\\0&:&\text{otherwise},\end{array}\right.\\
  C_i^{\operatorname{PHI}}(\game,S) &=&
  \left\{\begin{array}{rcl}\frac{1}{\sum_{j=1}^n|\mathcal{W}_j|}&:&S\in\mathcal{W}_i,\\0&:&\text{otherwise},\end{array}\right.\\
  C_i^{\operatorname{Chow}}(\game,S) &=&
  \left\{\begin{array}{rcl}1&:&S\in\mathcal{W}_i,\\0&:&\text{otherwise},\end{array}\right.\\
  C_i^{\operatorname{JS}}(\game,S)&=&
  \left\{\begin{array}{rcl}
    \frac{1}{\left|\{j\in S\,:\,\game(S-j)=0\}\right|} &:&i\in S,\game(S)=1, \game(S-i)=0,\\0 &:& \text{otherwise,}
  \end{array}\right.\\
  C_i^{\operatorname{PGI}}(\game,S)&=&
  \left\{\begin{array}{rcl}1&:&S\in\mathcal{W}^m_i,\\0&:&\text{otherwise},\end{array}\right.\\
  C_i^{\operatorname{DP}}(\game,S)&=&
  \left\{\begin{array}{rcl}\frac{1}{|S|}&:&S\in\mathcal{W}^m_i,\\0&:&\text{otherwise},\end{array}\right.\\
  C_i^{\operatorname{Shift}}(\game,S)&=&
  \left\{\begin{array}{rcl}1&:&S\in \mathcal{W}^{sm}_i,\\0&:&\text{otherwise},\end{array}\right.\\
  C_i^{\operatorname{SDP}}(\game,S)&=&
  \left\{\begin{array}{rcl}\frac{1}{|S|}&:&S\in \mathcal{W}^{sm}_i,\\0&:&\text{otherwise}.\end{array}\right.
\end{eqnarray*}

\section{Properties of power indices}

In Section~\ref{sec_power_indices} we have introduced several properties of power indices. Whether a certain
power index satisfies such a property depends on the class of games where he is applied to. Here we want to
restrict ourselves onto the class of simple games and summarize the results in Table~\ref{table_properties_power_indices}.
Some of the stated entries can be e.g.\ found in \cite{bertini2013comparing}, are folklore, or can be easily 
proven. Since all mentioned power indices $P$ are positive, we can apply Lemma~\ref{lemma_efficient} and deduce
that the respective normalized power index $\widehat{P}$ is efficient. 

\begin{table}[htp]
  \begin{center}
    \begin{tabular}{rccccc}
      \hline
      \textbf{power index} & \textbf{symmetric} & \textbf{positive} & \textbf{efficient} & \!\!\!\textbf{null {\voter} property}\!\!\! & \!\!\!\textbf{null {\voter} removable}\!\!\!\\
      \hline
      $\operatorname{SSI}$     & \checked & \checked & \checked & \checked & \checked \\
      $\operatorname{Tijs}$    & \checked & -        & -        & \checked & \checked \\
      $\Psi^{\mathbf{p}}$      & \checked & \checked & -        & \checked & \checked \\
      $\Psi^{p}$               & \checked & \checked & -        & \checked & \checked \\
      $\operatorname{Bz}$      & \checked & \checked & -        & \checked & \checked \\
      $\operatorname{swing}$   & \checked & \checked & -        & \checked & -        \\
      $\operatorname{ColPrev}$ & \checked & \checked & -        & \checked & \checked \\
      $\operatorname{ColIni}$  & \checked & \checked & -        & \checked & \checked \\
      $\operatorname{Rae}$     & \checked & \checked & -        & -        & \checked \\
      $\operatorname{KB}$      & \checked & \checked & -        & -        & \checked \\
      $\operatorname{PHI}$     & \checked & \checked & \checked & -        & -        \\
      $\operatorname{Chow}$    & \checked & \checked & -        & -        & -        \\
      $\operatorname{JS}$      & \checked & \checked & -        & \checked & -        \\
      $\operatorname{PGI}$     & \checked & \checked & -        & \checked & \checked \\
      $\operatorname{DP}$      & \checked & \checked & -        & \checked & \checked \\
      $\operatorname{Shift}$   & \checked & \checked & -        & \checked & \checked \\
      $\operatorname{SDP}$     & \checked & \checked & -        & \checked & \checked \\
      $\operatorname{Nuc}$     & \checked & \checked & \checked & \checked & \checked \\
      $\operatorname{MSR}$     & \checked & \checked & \checked & \checked & \checked \\
      $\operatorname{Colomer}$ & \checked & \checked & -        & -        & -        \\
      \hline
    \end{tabular}
    \caption{Properties of power indices on simple games}
    \label{table_properties_power_indices}
  \end{center}
\end{table}

\section{Recursion formulas for power indices}
\label{sec_recursion_formulas_for_power_indices}

Given the {\voter}~set $N$, simple games are uniquely characterized by their set $\mathcal{W}$ of winning 
coalitions. With the inclusion operation for the sets of winning coalitions, simple games with {\voter}~set $N$ 
become a poset and even a graded poset using the rank function $|\mathcal{W}|$, i.e.\ the number of winning coalitions.
Starting from the simple game $\game$, where all non-empty subsets are winning, we can reach each simple game by
recursively turning a minimal winning coalition into a losing coalition. We can use this recursion for the computation 
of power indices of simple games. For each symmetric and efficient power index $P$ we have
$P(\game)=\frac{1}{n}\cdot(1,\dots,1)$. For other power indices the value $P(\game)$ can usually be  obtained easily. It
remains to provide a recursive formula for a given power index for the case that one minimal winning coalition is turned
into a losing coalition. For the swing count such a recursive formula is e.g.\ given in \cite[Lemma 3.3.12]{0954.91019}:  

\begin{lemma}
  Let $\game=(\mathcal{W},N)$ be a simple game and $T\neq N$ be one of its minimal winning coalitions.
  Turning $T$ into a losing coalition gives a simple game $\game'$ with
  $$
    \eta_i(\game')=\left\{\begin{array}{rcl}\eta_i(\game)-1 &:& i\in T,\\ \eta_i(\game)+1 &:& i\in N\backslash T\end{array}\right.
  $$
  for all $i\in N$. %, where $\eta_i(\game)$ denotes the number of swings for {\voter}~$i$ in $\game$.
\end{lemma}

Since $\operatorname{BZ}_i(\game)=\frac{\eta_i(\game)}{2^{n-1}}$ we obtain:
\begin{lemma}
  Let $\game=(\mathcal{W},N)$ be a simple game and $T\neq N$ be one of its minimal winning coalitions.
  Turning $T$ into a losing coalition gives a simple game $\game'$ with
  $$
    \operatorname{BZ}_i(\game')=\left\{\begin{array}{rcl}\operatorname{BZ}_i(\game)-\frac{1}{2^{n-1}} &:& i\in T,\\
    \operatorname{BZ}_i(\game)+\frac{1}{2^{n-1}} &:& i\in N\backslash T\end{array}\right.
  $$
  for all $i\in N$.
\end{lemma}

From $\operatorname{Rae}_i(\game)=\frac{1}{2}+\frac{1}{2}\cdot\operatorname{BZ}_i(\game)$ we conclude:
\begin{lemma}
  Let $\game=(\mathcal{W},N)$ be a simple game and $T\neq N$ be one of its minimal winning coalitions.
  Turning $T$ into a losing coalition gives a simple game $\game'$ with
  $$
    \operatorname{Rae}_i(\game')=\left\{\begin{array}{rcl}\operatorname{Rae}_i(\game)-\frac{1}{2^{n}} &:& i\in T,\\
    \operatorname{Rae}_i(\game)+\frac{1}{2^{n}} &:& i\in N\backslash T\end{array}\right.
  $$
  for all $i\in N$.
\end{lemma}

Also the Shapley-Shubik index counts swings in a weighted form, so that:
\begin{lemma}
  Let $\game=(\mathcal{W},N)$ be a simple game and $T\neq N$ be one of its minimal winning coalitions.
  Turning $T$ into a losing coalition gives a simple game $\game'$ with
  $$
    \operatorname{SSI}_i(\game')=\left\{\begin{array}{rcl}\operatorname{SSI}_i(\game)-\frac{(|T|-1)!\cdot(|N|-|T|)!}{|N|!} &:& i\in T,\\
    \operatorname{SSI}_i(\game)+\frac{|T|!\cdot(|N|-1-|T|)!}{|N|!} &:& i\in N\backslash T\end{array}\right.
  $$
  for all $i\in N$.
\end{lemma}

More generally for $p$-binomial semivalues we have:
\begin{lemma}
  Let $p\in(0,1)$ and $\game=(\mathcal{W},N)$ be a simple game and $T\neq N$ be one of its minimal winning coalitions.
  Turning $T$ into a losing coalition gives a simple game $\game'$ with
  $$
    \Psi^p_i(\game')=\left\{\begin{array}{rcl}\Psi^p_i(\game)-p^{t-1}(1-p)^{n-t} &:& i\in T,\\
    \Psi^p_i(\game)+p^t(1-p)^{n-t-1} &:& i\in N\backslash T\end{array}\right.
  $$
  for all $i\in N$, where $t=|T|$ and $n=|N|$.
\end{lemma}
We remark that the increments and the decrements differ by a factor of $\frac{p}{1-p}$ (or $\frac{1-p}{p}$). 
If only winning coalitions are counted, then only the value of the {\voter}s in $T$ are changed:
\begin{lemma}
  Let $\game=(\mathcal{W},N)$ be a simple game and $T\neq N$ be one of its minimal winning coalitions.
  Turning $T$ into a losing coalition gives a simple game $\game'$ with
  $$
    \operatorname{Chow}_i(\game')=\left\{\begin{array}{rcl}\operatorname{Chow}_i(\game)-1 &:& i\in T,\\
    \operatorname{Chow}_i(\game) &:& i\in N\backslash T\end{array}\right.
  $$
  for all $i\in N$.
\end{lemma}
For the K\"onig-Br\"auninger index and the Coleman power of a member to prevent action we have to consider a reweighting
of the values based on the change of the number of winning coalitions:
\begin{lemma}
  Let $\game=(\mathcal{W},N)$ be a simple game and $T\neq N$ be one of its minimal winning coalitions.
  Turning $T$ into a losing coalition gives a simple game $\game'$ with
  $$
    \operatorname{KB}_i(\game')=\left\{\begin{array}{rcl}\frac{|\mathcal{W}|}{|\mathcal{W}|-1}\cdot\operatorname{KB}_i(\game)-\frac{1}{|\mathcal{W}|-1} &:& i\in T,\\
    \frac{|\mathcal{W}|}{|\mathcal{W}|-1}\cdot\operatorname{KB}_i(\game) &:& i\in N\backslash T\end{array}\right.
  $$
  and
  $$
    \operatorname{ColPrev}_i(\game')=\left\{\begin{array}{rcl}\frac{|\mathcal{W}|}{|\mathcal{W}|-1}\cdot\operatorname{ColPrev}_i(\game)-\frac{1}{|\mathcal{W}|-1} &:& i\in T,\\
    \frac{|\mathcal{W}|}{|\mathcal{W}|-1}\cdot\operatorname{ColPrev}_i(\game)+\frac{1}{|\mathcal{W}|-1} &:& i\in N\backslash T\end{array}\right.
  $$
  for all $i\in N$.
\end{lemma}

\begin{lemma}
  Let $\game=(\mathcal{W},N)$ be a simple game and $T\neq N$ be one of its minimal winning coalitions.
  Turning $T$ into a losing coalition gives a simple game $\game'$ with
  $$
    \operatorname{JS}_i(\game')=\left\{\begin{array}{rcl}\operatorname{JS}_i(\game)-\frac{1}{|T|} &:& i\in T,\\ 
    \operatorname{JS}_i(\game)+\frac{1}{|T|+1} &:&
    i\in N\backslash T\end{array}\right.
  $$
  for all $i\in N$.
\end{lemma}

For $\operatorname{PGI}$, $\operatorname{DP}$, $\operatorname{Shift}$, and $\operatorname{SDP}$ it may be a non-trivial task 
to write down similar recursion formulas. 

\section{Technical details for Shapley-Shubik index}
\label{sec_technical_SSI}

For the parametric class of weighted games $\game_{n,m}^{k,l}$, defined in Subsection~\ref{subsec_parameterized_wvg}, 
and its possible $k$-roundings $\game_{n,n+1}^{k,l}$ and $\game_{n,0}^{k,l}$, we want to deduce an explicit formula 
for the computation of the Shapley Shubik index for each {\voter} type. 

For $\game_{n,0}^{k,l}$ the first $k$ {\voter}s are symmetric and the other {\voter}s are null {\voter}s, so that 
we obtain
$$
  \operatorname{SSI}(\game_{n,0}^{k,l})=\left(\frac{1}{k},\dots,\frac{1}{k},\frac{1}{k},\dots,\frac{1}{k},0,\dots,0\right).
$$
For $\game_{n,n+1}^{k,l}$ and a {\voter} $i\in[1,k-l]$ the $i$-swings are given by the coalitions $U\cup V$, where 
$U\subseteq [1,k]-i$ with $|U|=l-1$, $V\subseteq (k,k+n]$, and $T\cup V$, where $V\subseteq (k,k+n]$. Thus we obtain
for the corresponding value of the $\operatorname{SSI}$: 
\begin{eqnarray*}
  && \frac{1}{(n+k)!}\cdot\Bigg(\sum_{V\subseteq (k,k+n]}\Bigg(\sum_{\overset{U\subseteq [1,k]-i}{|U|=l-1}} 
  |U\cup V|!\cdot (n+k-|U\cup V|-1)!\\
 &&\,\,+\,\,|T\cup V|!\cdot (n+k-|T\cup V|-1)!\Bigg)\Bigg)\\
  &=& \frac{1}{(n\!+\!k)!}\cdot\left(\sum_{j=0}^n {{k\!-\!1}\choose {l\!-\!1}}{n\choose j}(l\!-\!1\!+\!j)!(n\!+\!k\!-\!l\!-\!j)
  +\sum_{j=0}^n {n\choose j}(l\!+\!j)!(n\!+\!k\!-\!l\!-\!j\!-\!1)!\right)\\
  &=& \frac{1}{k}+\frac{1}{k\cdot{{k-1}\choose l}} 
\end{eqnarray*}

Since the last $n$ {\voter}s, in $(k,k+n]$, are null {\voter}s, the central $l$~{\voter}s are symmetric,
and the Shapley-Shubik index is efficient, the Shapley vector for $\game_{n,n+1}^{k,l}$ is given by 
$\operatorname{SSI}(\game_{n,n+1}^{k,l})=$
$$
  \left(\frac{1}{k}+\frac{1}{k\cdot{{k-1}\choose l}},\dots,\frac{1}{k}+\frac{1}{k\cdot{{k-1}\choose l}},
  \frac{1}{k}-\frac{k-l}{kl\cdot{{k-1}\choose l}},\dots,\frac{1}{k}-\frac{k-l}{kl\cdot{{k-1}\choose l}},0,\dots,0\right).
$$
For $\game_{m,n}^{k,l}$ the analysis is a bit more complicated. For $i\in(k,k+n)$ the $i$-swings are given by 
$T\cup V$ with $V\subseteq (k,k+n]-i$ and $|V|=m-1$, so that the Shapley value for those {\voter}s is given by
\begin{eqnarray*}
  && \frac{1}{(n+k)!}\cdot\sum_{\overset{V\subseteq (k,k+n]-i}{|V|=m-1}} |T\cup V|!\cdot (n-|T\cup V|-1)!\\
  &=& \frac{1}{(n+k)!}\cdot {n \choose {m-1}}(l+m-1)!(n+k-m-l)!
\end{eqnarray*}  
For $i\in[1,k-l]$ the $i$-swings are given by $U\cup V$ and $T\cup V'$, where $V\subseteq (k,k+n]$, 
$U\subseteq [1,k]-i$, $|U|=l-1$, and $V'\subseteq (k,k+n]$, $|V'|\le m-1$, so that the Shapley
value for those {\voter}s is given by
\begin{eqnarray*}
  && \frac{1}{(n+k)!}\cdot\left(\sum_{\overset{U\subseteq [1,k]-i}{|U|=l-1}}\sum_{V\subseteq (k,k+n]}
  |U\cup V|!\cdot (n+k-|U\cup V|-1)!\right)\\
  && +\frac{1}{(n+k)!}\cdot\left(\sum_{\overset{V'\subseteq (k,k+n]}{|V'|\le m-1}}|T\cup V'|!\cdot(n+k-|T\cup V'|-1)!\right)\\
  &=& \frac{1}{(n+k)!}\cdot\Bigg(\sum_{j=0}^n {{k-1}\choose {l-1}}{n\choose j}(l-1+j)!(n+k-l-j)!\\
  &&\,+\,\sum_{j=0}^{m-1} {n\choose j}(l+j)!(n+k-l-j-1)!\Bigg)\\
  &=& \frac{1}{k}+r_{k,l,n}^m,
\end{eqnarray*}  
where $r_{k,l,n}^m:=\frac{1}{(k+n)!}\cdot \sum\limits_{j=0}^{m-1} {n\choose j}(l+j)!(n+k-l-j-1)!$ with
$0\le r_{k,l,n}^m \le \frac{1}{k {{k-1} \choose l}}$.

We have done some numerical experiments with those formulas suggesting the existence of an Alon-Edelman type result
for the Shapley-Shubik index. 

\end{document}